\RequirePackage{fix-cm}
\documentclass[twocolumn]{svjour3}          
\smartqed  
\usepackage{graphicx}
\usepackage{subfigure}
\usepackage{multirow}

\usepackage{amsthm}
\usepackage[explicit]{titlesec}
\usepackage[noend, linesnumbered, ruled, lined]{algorithm2e}
\usepackage{hyperref}
\usepackage{balance}
\usepackage{color}
\usepackage{pifont}
\usepackage{amsmath}
\usepackage{flushend}
\usepackage{bm}
\usepackage{array}
\usepackage{setspace}
\usepackage{bbding}
\usepackage{booktabs}
\usepackage[misc]{ifsym}
\begin{document}

\title{Efficient Graph Embedding at Scale: Optimizing CPU-GPU-SSD Integration} 



\author{Zhonggen Li  \and
        Xiangyu Ke  \and
        Yifan Zhu \and
        Yunjun Gao \and
        Feifei Li
}


\institute{
           Zhonggen Li\and Xiangyu Ke \and Yifan Zhu \and Yunjun Gao (\Letter) \at
           Zhejiang University, Hangzhou, China \at
              \email{\{zgli, xiangyu.ke, xtf\_z, gaoyj\}@zju.edu.cn} 
            \and
        Feifei Li \at
        Alibaba Group, Hangzhou, China \at
            lifeifei@alibaba-inc.com
}

\date{Received: date / Accepted: date}

\maketitle

\begin{abstract}
Graph embeddings map graph nodes to continuous vectors and are foundational to community detection, recommendation, and many scientific applications. 
{At billion-scale, however, existing graph embedding systems face a trade-off: they either rely on large in-memory footprints across many GPUs (limited scalability) or repeatedly stream data from disk (incurring severe I/O overhead and low GPU utilization).}

In this paper, we propose {\sf Legend}, a \underline{l}ightweight h\underline{e}te-
ro\underline{g}\underline{en}eous system for graph embe\underline{d}ding that systematically redesigns data management across CPU, GPU, and NVMe SSD resources. 
{{\sf Legend} combines three practical ideas: (1) a prefetch-friendly embedding-loading order that lets GPUs efficiently prefetch necessary embeddings directly from NVMe SSD with low I/O amplification; (2) a high-throughput GPU–SSD direct-access driver tuned for the access patterns of embedding training; and (3) a customized parallel execution strategy that maximizes GPU utilization. 
Together, these components let {\sf Legend} store and stream vast embedding data without overprovisioning GPU memory or suffering I/O stalls.}
Extensive experiments on billion-scale graphs demonstrate that {\sf Legend} speeds up end-to-end workloads by up to 4.8$\times$ versus state-of-the-art systems, and matches their performance on the largest workloads while using only one quarter of the GPUs.


\end{abstract}

\section{Introduction}
\label{sec:introduction}
Graphs are a fundamental model for representing relationships across domains~\cite{kim2020densely,li2021analyzing,fang2020effective}. 
Recent advances in graph machine learning have expanded the scope of graph analytics to tasks such as link prediction~\cite{zhu2019aligraph,zou2023embedx,yuan2023comprehensive} and node classification~\cite{song2022towards,zhang2022mul,wen2021meta}. 
Central to these successes are graph embeddings -- continuous vector representations that capture structural properties of nodes and serve as the basis for applications in recommendation systems~\cite{xia2021multi,zhang2021group,miao2022het}, dialogue and conversational agents~\cite{zhou2024atom,edge2024local}, drug discovery~\cite{zhong2023knowledge,zhang2025predicting,he2022webmile}, etc. 
{Despite the success, acquiring high-quality graph embeddings entails prohibitive computational overhead, hindering the scalability of large-scale graphs~\cite{zheng2024ge2}.} 

To scale embedding training to large graphs, a range of systems have been proposed, including Amazon’s {\sf DGL-KE}~\cite{zheng2020dgl}, Meta’s {\sf PyTorch Big Graph} ({\sf PBG})~\cite{lerer2019pytorch}, {\sf Marius}~\cite{mohoney2021marius}, and {\sf GE$^2$}~\cite{zheng2024ge2}. 
These systems leverage GPUs to accelerate computation, but real-world graphs often contain millions or even billions of nodes, so storing embeddings and optimizer state entirely in GPU memory\footnote{For instance, the Freebase86M dataset—comprising 86 million nodes—demands 68 GB of memory to store 100-dimensional embeddings and optimizer states.} becomes impractical~\cite{mohoney2021marius,zheng2024ge2}. 
As a result, two dominant engineering strategies have emerged:
\textbf{(i)} {\em RAM-based systems} (e.g., {\sf DGL-KE} and {\sf GE$^2$}) keep embeddings and optimizer state in host RAM and transfer required partitions to GPUs over PCIe on demand.
\textbf{(ii)} {\em Disk-based systems} (such as {\sf PBG} and {\sf Marius}) persist embeddings on disk and stream partitions into RAM (and then to GPUs) as needed.
Figures~\ref{fig:architrcture}(a) and (b) illustrate the architectures of {\em RAM-based systems} and {\em Disk-based systems}, respectively.
Although both approaches enable large-scale training, each has important limitations. 
{
RAM-based designs are constrained by host-memory capacity and incur high provisioning costs: on a billion-node dataset (\textit{Twitter}), RAM-based deployments inflate device costs substantially (Table~\ref{tab:metrics_comparison}).
Disk-based designs avoid large RAM footprints but introduce long data paths (disk-CPU-GPU) and suffer from suboptimal GPU utilization; in our measurements, disk-based systems exhibit average GPU utilization of only about 58\%, resulting in reduced end-to-end efficiency. 
These trade-offs motivate a rethinking of data placement and I/O strategies across CPU, GPU, and NVMe tiers to improve both scalability and utilization for bill-
ion-scale graph embedding workloads.
}

\begin{figure}
    \centering
    \includegraphics[width=0.48\textwidth]{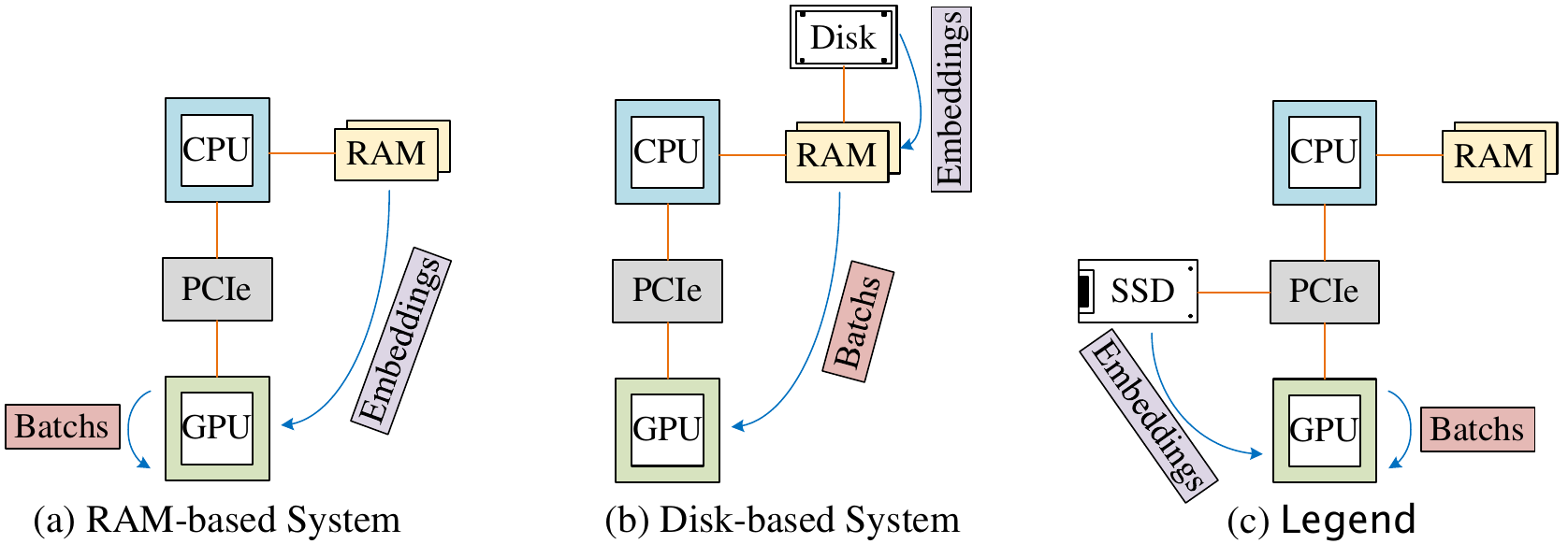}
    \caption{Comparison of system architectures.}
    \label{fig:architrcture}
\end{figure}


\begin{table*}[tbp]
    \centering
    \renewcommand{\arraystretch}{1.25}
    \small
    \caption{{Statistics of cost and efficiency on dataset \textit{Twitter} (Edges: 1.3B, Nodes: 41.6M). Prices of devices are taken from Amazon when writing the paper. Our experiments evaluate the bandwidths, throughputs, and execution overheads.} }
    \vspace{1mm}
    \begin{tabular}{p{1.4cm}<{\centering}p{1.7cm}<{\centering}p{2.5cm}<{\centering}p{3.5cm}<{\centering}p{2.6cm}<{\centering}p{1.55cm}<{\centering}p{1.45cm}<{\centering}}
    \toprule
       {Systems} & {Storage cost} & {Computing cost} & {Comm. bandwidth} & {Comp. throughput} & {Batch time} & {Total time} \\
    \midrule
       {\sf GE$^2$~\cite{zheng2024ge2}} & 2.02 \$/GB & 33.2k \$ (A100$\times$4) & 10.05 GB/s (CPU-GPU) & 6.75 $\times 10^{6}$ edges/s & 18.5 ms & 32 min \\
       {\sf Marius~\cite{mohoney2021marius}} & 0.13 \$/GB & 8.3k \$ (A100$\times$1) & 3.12 GB/s (SSD-CPU) & 1.49 $\times 10^{6}$ edges/s & 315.6 ms & 146 min \\
       {\sf {Legend}} & {0.13 \$/GB} & {8.3k \$ (A100$\times$1)} & {3.06 GB/s (SSD-GPU)} & {7.18 ${\times 10^{6}}$ edges/s} & {12.0 ms} & {30 min} \\
    \bottomrule
    \end{tabular}
    \label{tab:metrics_comparison}
\end{table*}

In addition to the {\it storage-related inefficiencies}, current graph-embedding systems suffer from two critical {\it computational bottlenecks}: 
{\em First}, many systems (e.g., {\sf Marius}) remain CPU-centric for tasks such as batch construction, negative sampling, and embedding updates. This CPU dependence increases CPU–GPU communication and underutilizes accelerator compute, producing substantial per-batch slowdowns (up to 26$\times$ in our measurements; see Table~\ref{tab:metrics_comparison}). 
{\em Second}, even systems that offload batch construction to the GPU (e.g., {\sf GE$^2$}) can still exhibit suboptimal GPU utilization because their designs emphasize embedding-framework features over training-path optimizations. 
In our experiments, the GPU-side batch computation dominates end-to-end batch time -- accounting for more than 80\% of processing -- thereby becoming the principal performance bottleneck. 
Together, these two limitations, {\it excessive CPU dependence} and {\it inefficient GPU use}, significantly reduce throughput and hamper the scalability of billion-scale embedding workloads.

Recently, NVMe SSDs have gained traction due to their favorable balance of cost and performance~\cite{haas2023modern,park2022ginex,fang2024enabling}. 
Recent advances in NVMe SSDs offer an attractive cost–performance trade-off for large-scale storage~\cite{haas2023modern,park2022ginex,fang2024enabling}. 
Accordingly, we adopt the architecture in Figure~\ref{fig:architrcture}(c), which replaces conventional RAM- or disk-backed storage with NVMe SSDs and enables GPU–SSD direct access to provide both economical capacity and lower-latency data movement. 
To realize efficient embedding training on this platform, we 
(i) map graph-embedding tasks across CPU, GPU, and SSD, 
and (ii) optimize GPU-side computations. 
Although prior work has investigated GPU–SSD direct access in other domains~\cite{huang2024neos,bae2021flashneuron,jeongmin2024accelerating}, applying this model to graph embedding and maximizing GPU utilization during training raises three key challenges:
\begin{itemize}
\vspace{0.5mm}
    \item \textbf{Task Mapping.} 
    Existing embedding workflows are not engineered for a CPU–GPU–SSD hardware stack and thus perform poorly when naively ported. 
    For example, systems that treat the NVMe device as a slow disk (e.g., {\sf Marius}) overload the CPU with transfer work, while approaches that treat the SSD like RAM (e.g., {\sf GE$^2$}) suffer because SSD-GPU bandwidth is substantially lower than RAM-GPU bandwidth (three times lower as shown in Table~\ref{tab:metrics_comparison}).
    \item \textbf{I/O Bottlenecks.} 
    Even with GPU–SSD direct access (eliminating some CPU-mediated copies), I/O remains a nontrivial fraction of end-to-end training time -- in our measurements, it accounts for over 25\% of the pipeline. The lower raw transfer rates of SSDs versus RAM, therefore, remain a limiting factor for throughput.
    \item \textbf{Computational Intensity.} 
    Once the required embedding shards and optimizer state are resident on the GPU, downstream steps (batch construction, negative sampling, and weight updates) involve computationally intensive element-wise and reduction operations (e.g., exponentiation, matrix multiplication). 
    These operations must be carefully organized to saturate the GPU; otherwise, GPU-side batch computation becomes the dominant bottleneck.
\end{itemize}

In this paper, we propose {\sf Legend}, a scalable graph-embedding system that tightly integrates CPU, GPU and NVMe SSD resources and applies three complementary optimization families to address the aforementioned challenges of SSD-backed training. 
Below we summarize each component; implementation and formal details are presented in \S\ref{sec:workflow}–\S\ref{subsec:optGPU}.

\vspace{0.2cm}
\noindent{\bf Storage Arrangement and Task Allocation. }
To exploit the capacity and bandwidth characteristics of the platform in Figure~\ref{fig:architrcture}(c), {\sf Legend} separates hot and cold state: graph topology and frequently accessed metadata remain in host RAM, while large, infrequently accessed arrays (node embeddings and optimizer state) are resident on NVMe SSD. 
Control and orchestration tasks are mapped to the CPU, whereas heavy linear algebra and per-batch computation execute on the GPU. 
This partitioning reduces host-memory pressure, keeps latency-sensitive graph data immediately available in RAM, and leverages the GPU for compute-intensive kernels, yielding a pragmatic trade-off between capacity, bandwidth, and compute.

\vspace{0.1cm}
\noindent{\bf I/O Optimizations.} We introduce two techniques tailored for GPU–SSD embedding training: 
\textbf{(i)} {\textbf{\em Prefetch-friendly embedding ordering}.} 
{For existing graph embedding systems, embeddings are usually partitioned and loaded to the GPU in an I/O-optimized order~\cite{lerer2019pytorch,zheng2020dgl}. 
Although prefetching is supported in some existing graph embedding systems~\cite{mohoney2021marius}, the loading orders hinder the effective overlap of the I/O and computation, and the CPU-managed prefetching to RAM fails to tackle the bottleneck of CPU-GPU data transfer, leading to significant data transfer overhead.}
We prove that generating a prefetch-friendly order while minimizing I/O times is an NP-hard problem. 
{To get practical, we implement an efficient order-generation algorithm based on a column-separation covering strategy that produces a partition-swap order supporting effective I/O–compute overlap; this order achieves I/O latency comparable to state-of-the-art heuristics while enabling GPU-driven prefetches directly from SSD into GPU memory. }
\textbf{(ii)} \textbf{{\em Customized GPU-SSD direct access driver.}} Off-the-shelf GPU–SSD drivers expose substantial overhead from fine-grained locking and doorbell operations when used naively for embedding workloads~\cite{qureshi2023gpu,jeongmin2024accelerating,Markussen2021}. 
{In {\sf Legend} we redesign NVMe queue management for our access patterns: queue positions are precomputed so threads can concurrently enqueue and dequeue without costly locks, and we use a full-coalesced doorbell ringing and batch polling strategy to minimize doorbell overhead and maximize throughput. These changes remove the partial-coalescing inefficiencies of prior approaches and substantially increase GPU–SSD communication bandwidth (\S~\ref{subsec:optNVMe}). }

\vspace{0.1cm}
\noindent{\bf Computation Optimizations.} 
To efficiently perform the calculations during embedding, we design a parallel strategy tailored to graph embedding learning workloads that fully leverages Tensor cores, registers, and shared memory in the GPU. This strategy reduces heavy memory access while ensuring high parallelism. 
Additionally, we identify redundant calculations and reuse intermediate results to further reduce computational costs (\S~\ref{subsec:optGPU}). 
These optimizations address the neglected issue of embedding computation in existing graph embedding systems, fully utilizing GPU resources, and achieving higher GPU utilization. 


Comprehensive experiments demonstrate the efficiency and scalability of {\sf Legend}. It achieves a speedup of up to 4.8$\times$ compared to the state-of-the-art graph embedding systems. {\sf Legend} is also lightweight and exhibits comparable performance on a single GPU to the state-of-the-art system using 4 GPUs on the billion-scale \textit{Twitter} dataset (\S~\ref{sec:evaluation}).

In summary, our key contributions are as follows: 
\begin{itemize}
    \item We design a workflow to reasonably allocate tasks for graph embedding in the CPU-GPU-NVMe SSD heterogeneous systems, considering the respective characteristics of each hardware component (\S~\ref{sec:workflow}). 
     \item We prove the NP-hardness of identifying a prefetch-friendly iteration order and propose a heuristic algorithm to solve the problem (\S~\ref{sec:order}). And we devise a customized GPU-SSD direct access mechanism (\S~\ref{subsec:optNVMe}) to achieve efficient I/O during embedding training. 
    \item We optimize the batch computation on the GPU by devising a specific parallel strategy and exploiting the computation and storage resources to enhance GPU utilization (\S~\ref{subsec:optGPU}). 
    \item We conduct comprehensive evaluations demonstrating that {\sf Legend} outperforms existing graph embedding systems, achieving up to 4.8$\times$ speedup for large-scale graph embedding (\S~\ref{sec:evaluation}). 
\end{itemize}

The rest of this paper is organized as follows.  
Section \ref{sec:preliminaries} briefly introduces the background of graph embedding, GPU architectures, and data access mechanism in NVMe SSD. 
Section \ref{sec:workflow} presents the workflow design in {\sf Legend}. 
Section \ref{sec:order} describes the prefetch-friendly loading order. 
Section \ref{subsec:optNVMe} and Section \ref{subsec:optGPU} illustrate the optimization on SSD and GPU, respectively. 
Section \ref{sec:evaluation} exhibits the experimental results. 
Section \ref{relatedworks} reviews the related studies.
We conclude the paper in Section \ref{sec:conclusion}. 
\section{Preliminaries}
\label{sec:preliminaries}
In this section, we first provide an overview of graph embedding learning. Subsequently, we offer a concise description of GPU architecture, following the data access mechanism of NVMe SSD. 

\begin{figure*}
    \centering
    \includegraphics[width=0.95\textwidth]{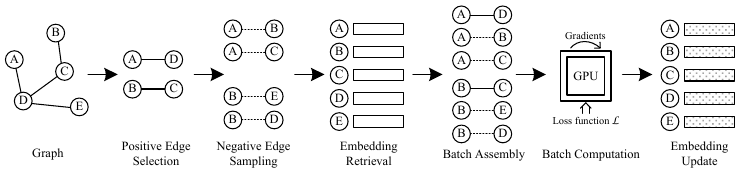}
    \caption{{Example of graph embedding.} }
    \label{fig:comExample}
\end{figure*}

\subsection{Graph Embedding Learning}
\label{subsec:graph_embedding_pre}
Following \textsf{PBG}~\cite{lerer2019pytorch}, \textsf{Marius}~\cite{mohoney2021marius} and \textsf{GE$^2$}~\cite{zheng2024ge2}, we focus on the multi-relation graphs denoted by $G=(V,R,E)$, where $V$ represents the set of nodes (entities), $R$ is the set of edge (relation) types and $E$ is the set of edges. Each edge $e\in E$ is a triplet denoted as $e=(s,r,d)$, where $s$ is the source node, $r$ is the relation type, and $d$ is the destination node. The triplet $(s,r,d)$ signifies that entity $s$ has a relationship $r$ with entity $d$, indicating the presence of an edge between $s$ and $d$. 
Although {\sf Legend} primarily targets multi-relation graphs, it's also capable of handling graphs without relation types. 

An embedding is a vector $\theta$ of fixed dimension. During graph embedding learning, the elements in the embedding vectors of each node and relation type are iteratively updated based on their previous values. Specifically, graph embedding learning uses a score function $f(\theta_s,\theta_r,\theta_d)$, where $\theta_s$, $\theta_r$ and $\theta_d$ represent the embedding vectors of $s,r,d$ in the triplet $e=(s,r,d)$. 
There are various score functions (i.e., embedding models) proposed for multi-relation graph embedding, such as ComplEx~\cite{trouillon2016complex} and DistMult~\cite{yang2015embedding}, which share a similar computing procedure. Our work aims to accelerate the computing process of graph embedding, which is orthogonal to specific score functions.  
The goal of graph embedding learning is to maximize $f(\theta_s,\theta_r,\theta_d)$ for $(s,r,d)\in E$ and minimize $f(\theta_{s^{\prime}},\theta_{r^{\prime}},\theta_{d^{\prime}})$ for $(s^{\prime},r^{\prime},d^{\prime})\\ \notin E$, where $e=(s,r,d)$ is referred to as a positive edge and $e^{\prime}=(s^{\prime},r^{\prime},d^{\prime})$ is known as a negative edge, respectively. 
This objective is achieved using the contrastive loss function, as shown in Equation \ref{equa1}. 

\begin{equation}
\label{equa1}
    \mathcal{L} = -\sum_{(s,r,d)\in E}(f(\theta_s,\theta_r,\theta_d)
    -\log(\sum_{(s^{\prime},r^{\prime},d^{\prime})\notin E}e^{f(\theta_{s^{\prime}},\theta_{r^{\prime}},\theta_{d^{\prime}})}))
\end{equation}

{The loss function is optimized by stochastic gradient descent (SGD), which further promotes the updates of embeddings. The embedding models employ Adagrad as the optimizer for embedding updates.} To fully utilize the parallelism of GPUs, SGD is performed on batches. Specifically, a batch is composed of embeddings corresponding to both positive edges and negative edges. For each positive edge (edges from $E$), several negative edges are sampled using negative sampling algorithms.  
All edges in $E$ are iterated by being split into batches. 
{As shown in Figure \ref{fig:comExample}, the batch process comprises 6 key stages: \textbf{(1)} \textit{Positive edge selection:} fetch a fixed number of edges from $E$ in order (e.g., $(A,D)$ and $(B,C)$ in Figure \ref{fig:comExample}), serving as positive edges that reflect true graph connectivity. \textbf{(2)} \textit{Negative edge sampling}: generate negative samples by sampling node pairs from $V$ (the sampled node pairs are highly likely to be disconnected due to the sparsity of the graph, e.g., $B$ and $C$ for edge $(A,D)$, $E$ and $D$ for edge $(B,C)$ in Figure \ref{fig:comExample}). \textbf{(3)} \textit{Embedding retrieval}: retrieve current embeddings for nodes and relations associated with both positive and negative edges. These embeddings act as trainable parameters in the downstream loss computation. \textbf{(4)} \textit{Batch assembly}: package the edges with their corresponding embeddings into a batch, enabling parallelized computation of the contrastive loss (Equation \ref{equa1}). \textbf{(5)} \textit{Batch computation}: computes the loss function based on the current embeddings and calculates gradients according to the loss function. \textbf{(6)} \textit{Embedding update}: the original node embeddings are updated using the calculated gradients. }
When all edges are traversed once, an epoch is completed. It always requires several epochs to ensure the convergence of the updated embeddings. 
It's worth noting that existing graph embedding systems ignore the optimization of the time-consuming batch computation, which significantly restricts the training efficiency and GPU utilization.

As the number of nodes in a graph can easily reach hundreds of millions, the limited memory cannot accommodate such large-scale embedding data. To enable scalable graph embedding training, 
we adopt a partition-based scheme similar to {\sf PBG}. As illustrated in Figure \ref{fig:pbg}, \textsf{PBG} divides the node embeddings into several equal-sized partitions ($\{P_0,P_1,P_2,P_3\}$) based on the node IDs, and stores them on the disk. In practical implementations, optimizer states are stored alongside the node embeddings, although they are omitted in Figure \ref{fig:pbg} for simplicity. Correspondingly, the edges are grouped into several buckets, where the bucket ID $(i,j)$ indicates that the source nodes of these edges are located in node partition $P_i$, and the destination nodes are located in node partition $P_j$. 

During training, the edge buckets are processed in a specific order, such as the order denoted by ``$[k]$" inside each edge bucket in Figure \ref{fig:pbg}. To retrieve the embeddings related to the edges within each edge bucket, the corresponding node partitions are required to be loaded into the memory buffer from the disk. For example, the buffer in Figure \ref{fig:pbg} contains $P_0$ and $P_1$, supporting the training of edge buckets $\{(0,0),(0,1),(1,0),(1,1)\}$, as the source nodes and destination nodes of edges within these edge buckets are all from node partition $P_0$ and $P_1$. As the edge buckets are processed in order, the node partitions in the memory buffer are continuously updated. The node partitions in the memory buffer at any given time are referred to as the {\em buffer state}. For instance, the current buffer state in Figure \ref{fig:pbg} is $\{P_0,P_1\}$.    

It is important to note that the order in which node partitions are loaded and edge buckets are processed significantly affects the I/O times between the disk and the memory buffer. For example, if the edge buckets are iterated in the order of $\{(0,0),(1,3),(1,0)\}$, the I/O time is 4, as $P_0$ is loaded twice and other partitions are loaded once. In contrast, iterating in the order of $\{(0,0),(1,0),(1,3)\}$ reduces the I/O time to 3, since $P_0$ is loaded only once. 

\begin{figure}
    \centering
    \includegraphics[width=0.36\textwidth]{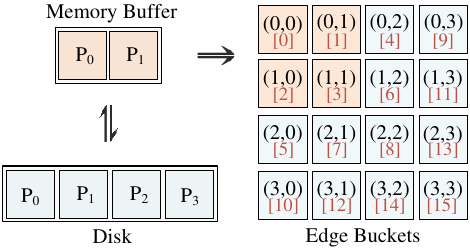}
    \caption{Partition-based training scheme. $P_i$ denotes the node partition and $[j]$ denotes the calculating order. }
    \label{fig:pbg}
\end{figure}

\subsection{GPU Architecture}
Modern graphics processing units (GPUs) are composed of numerous streaming multiprocessors (SMs), where each SM operates as an independent processing unit containing dozens or hundreds of computational cores. 
Within the CUDA programming framework, threads are grouped into 32-member execution units called {\em warps} that follow a SIMT (Single Instruction Multiple Threads) paradigm, requiring all threads in a warp to execute identical instructions synchronously.
The hierarchical execution model in CUDA further aggregates warps into {\em thread blocks}, which are scheduled on individual SMs for computation.

Modern GPUs feature two types of computing cores: CUDA cores and Tensor cores. CUDA cores serve as the primary computational units for general-purpose tasks, while Tensor cores are specialized for efficient matrix multiplication, enabling the multiplication of fixed-size matrices within a single clock cycle~\cite{zhu2019sparse}. 
The GPU typically has multiple levels of memory hierarchy, consisting of global memory, shared memory, and registers. 
Global memory, shared among all threads on the GPU, provides the largest capacity but has relatively low I/O bandwidth. When threads in a warp collectively read or write contiguous addresses of global memory, it can be performed by a single I/O transaction, which is called \textit{coalesced memory access}. 
Shared memory, accessible by all threads within each thread block, offers higher bandwidth but with a limited capacity of only several tens of KBs. 
Registers provide the fastest data access among these memory structures, which are private to individual threads once declared~\cite{cao2023gpu,owens2008gpu}. 

\subsection{Data Access Mechanism of NVMe SSD}
The NVMe SSD facilitates data transmission by leveraging queue pairs, which consist of submission queues (SQs) and completion queues (CQs). Multiple queue pairs in an NVMe SSD enable parallel responses to requests, ensuring high throughput~\cite{qureshi2023gpu,Markussen2021}. When a CPU or GPU requests data, it first constructs NVMe commands following the NVMe protocol. These commands specify the request type (read/write), request size (typically 512B or 4KB), request address, data placement address, and other parameters. Subsequently, it places the command at the end of an SQ. {Afterwards, it signals to the NVMe controller by writing the updated tail pointer into the doorbell register of the NVMe SSD via PCIe, indicating that new commands have been added to the SQ. This operation is known as doorbell ringing, which comes at a high cost. Therefore, some efforts are dedicated to reducing the doorbell ringing times.} The NVMe controller processes the commands, transfers the requested data to the host memory, and places completion entries into the CQ. Finally, the CPU or GPU retrieves the completion entries from the CQ and informs the NVMe controller by writing the new head pointer to the doorbell register, signifying that the new entries in the CQ have been processed~\cite{qureshi2023gpu}. 

{Although data access to the NVMe SSD follows the above protocol, the access driver is highly customizable for different workloads in distinct applications. For example, some customized NVMe drivers implement the protocol in the userspace of the operating system to reduce the overhead of the CPU software stack~\cite{yang2017spdk}, while others implement the protocol on the GPU to achieve higher throughput by employing substantial threads~\cite{qureshi2023gpu}. The queue management and doorbell ringing strategies are also redesigned to reduce redundant overhead. However, existing NVMe drivers are not specifically optimized for graph embedding tasks. Therefore, we propose a new driver customized for the graph embedding workload to achieve significant throughput, which will be detailed in Section \ref{subsec:optNVMe}.}
\section{Workflow in {\sf Legend}}
\label{sec:workflow}
In this section, we introduce the workflow of {\sf Legend}, focusing on storage arrangement, task assignment across different hardware components, and the overall data flow among each hardware component.

Following the partition schema used in {\sf PBG}~\cite{lerer2019pytorch}, {\sf Marius}~\cite{mohoney2021marius} and {\sf GE$^2$}~\cite{zheng2024ge2} (see Section \ref{subsec:graph_embedding_pre} and Figure \ref{fig:pbg}), {\sf Legend} divides the graph's nodes into $n$ equal-sized partitions based on their IDs ($n=4$ in Figure \ref{fig:workflow}). As a result, the node embeddings are split into $n$ corresponding partitions, and the edges are distributed into buckets. For example, the edge bucket $(1,2)$ in Figure \ref{fig:workflow} indicates that the source nodes of the edges in this edge bucket belong to node partition 1, while the destination nodes are from node partition 2.
{We carefully design the data placement and task mapping strategies during graph embedding to make full use of the unique characteristics of the CPU, GPU, and NVMe SSD. 
Next, we will introduce how {\sf Legend} maps storage and tasks to the architecture of CPU-GPU-SSD. }

\begin{figure}
    \centering
    \includegraphics[width=0.46\textwidth]{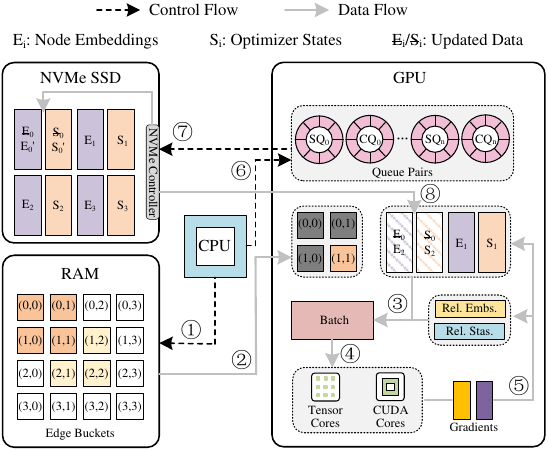}
    \caption{Workflow of {\sf Legend}. }
    \label{fig:workflow}
\end{figure}

\noindent\textbf{Storage Arrangement.} {{\sf Legend} adopts a three-tiered storage architecture to separately store node embeddings, edges, and relation embeddings. 
\textbf{(1)} Node embeddings and optimizer states are stored in the NVMe SSD, which occupies the majority of memory space during the graph embedding learning process. To maximize the bandwidth and make full use of the high parallelism of NVMe SSD, the embeddings and optimizer states of each partition are stored in consecutive memory addresses. This allows embedding and optimizer states of a partition to be loaded simultaneously with a single kernel on the GPU. As shown in Figure \ref{fig:workflow}, $E_0$ and $S_0$ are embeddings and optimizer states of node partition 0, which are stored consecutively. Otherwise, if the complete node embeddings and optimizer states are stored in the NVMe SSD instead of being stored in partitions, accessing one partition is required to perform with two requests, leading to additional data transfer overhead and failing to fully utilize the bandwidth between the SSD and GPU.}
Considering the significant overhead during embedding transmission, we design an I/O-efficient partition loading order and a customized high-throughput GPU-SSD direct access driver to reduce the I/O overhead between the GPU and the NVMe SSD, which will be introduced in Section \ref{sec:order} and \ref{subsec:optNVMe}.
{\textbf{(2)} The edges, which require significantly less space compared to embeddings and optimizer states, are stored in RAM and partitioned into edge buckets according to their source and destination nodes.} Storing the edges in RAM rather than NVMe SSD offers two key advantages. First, since the CPU controls the graph embedding learning process, it can effectively track the GPU's progress—specifically, which edge bucket is currently being processed. As a result, the CPU can transfer new edge buckets to the GPU on time and instruct the GPU to fetch the required embedding data from the NVMe SSD. {Second, although the theoretical bandwidth of RAM and SSD to GPU is the same, the actual bandwidth between RAM and GPU is more than 3 times higher than that between the NVMe SSD and GPU in our experiments (Table \ref{tab:metrics_comparison}) due to the hardware restriction.} Thus, storing edge buckets in RAM allows for efficient and synchronous transfers from the CPU to the GPU, reducing the GPU's idle time. 
\textbf{(3)} For multi-relation graphs, the number of relation types is typically small, necessitating frequent synchronous updates~\cite{mohoney2021marius}. Consequently, we store the relation embeddings (denoted as Rel. Embs.) and optimizer states (denoted as Rel. Stas.) in the global memory of the GPU, following the design of existing graph embedding systems~\cite{zheng2020dgl,mohoney2021marius,zheng2024ge2}. 
{Besides, there is a buffer in the GPU to temporarily hold the embeddings and optimizer states of partial node partitions. Although {\sf Legend} differs from prior systems by placing the buffer on the GPU and transferring data directly from SSDs to the buffer, its buffer management strategy is conceptually similar: embeddings of node partitions are evicted and loaded in a specific order. }

\noindent\textbf{Tasks Mapping.} Considering the powerful ability of the CPU to handle complex logic and control tasks, the CPU is responsible for moving data and sending commands to the GPU and NVMe SSD in {\sf Legend}, coordinating and controlling the processes of tasks on various hardware components. Meanwhile, considering the powerful parallel computing capability of the GPU, it takes on all computing tasks to achieve more efficient graph embedding computation. 
Considering the underutilization of the GPU during computation, we design an optimized parallel strategy and reorganize the computing procedures to fully utilize the resources on the GPU, which will be illustrated in Section \ref{subsec:optGPU}.
Based on this strategy, the CPU first transfers edge buckets from RAM to the GPU, and the GPU subsequently constructs batches as well as computes gradients. Once the CPU detects that the edges on the GPU are going to be used up, it instructs the GPU to fetch the next embedding partition from the SSD. Afterwards, the CPU transfers new edge buckets to the GPU, and a new round of processing begins in the same way. 

Specifically, assume that the nodes are divided into four partitions and that the buffer in the GPU global memory can accommodate two partitions at a time, i.e., the buffer capacity is 2. Initially, the embeddings and optimizer states of partition 0 and partition 1 reside in the GPU global memory and are randomly initialized, as shown in Figure \ref{fig:workflow}. With partition 0 and partition 1, the GPU conducts the computation of 4 edge buckets, namely $\{(0,0),(0,1),(1,0),(1,1)\}$, as the source and destination nodes of the edges within these edge buckets are located in these two node partitions. 
To do this, as depicted in Figure \ref{fig:workflow}, the CPU \ding{172} fetches edge buckets $\{(0,0),(0,1),(1,0),(1,1)\}$ from RAM and \ding{173} transfers them to the GPU global memory. The GPU then \ding{174} fetches a fixed number of edges (i.e., positive edges) from the edge buckets, samples negative edges for each positive edge, and retrieves the corresponding embeddings and optimizer states from $\{\{E_0, E_1\}, \{S_0, S_1\}\}$ to construct a batch. Next, the gradients of this batch \ding{175} are calculated using Tensor cores and CUDA cores, which will be detailed in Section \ref{subsec:optGPU}. The embeddings and the optimizer states in the global memory are \ding{176} updated by the GPU with the computed gradients. 
The advantages of sampling negative edges and constructing batches on the GPU are threefold. First, the corresponding embeddings of the trained edge buckets are stored in the GPU rather than in RAM, minimizing the need for frequent communication between the CPU and GPU. Second, the node embeddings are updated synchronously, avoiding the staleness issues encountered in some graph embedding systems, such as {\sf Marius}. Third, both negative edge sampling and embedding retrieval are parallelizable tasks, making them well-suited for execution on the GPU. 

When {\sf Legend} finishes the calculation of all four edge buckets in Figure \ref{fig:workflow}, it has to exchange an embedding and optimizer state partition in the GPU buffer ($E_0$ and $S_0$ in our example) with another partition in the NVMe SSD ($E_2$ and $S_2$ in our example). Due to the limited bandwidth between the SSD and GPU, the data transfer is not completed immediately. However, the GPU has no computational tasks during data exchange, resulting in low utilization. Consequently, necessary data is prefetched by {\sf Legend} at an appropriate time before being used, which will be introduced in Section \ref{sec:order}. The CPU \ding{177} launches a data access kernel at the appropriate time to have the GPU offload $E_0$ and $S_0$ to the NVMe SSD, and to load $E_2$ and $S_2$ into the GPU's global memory, which is the key operation of prefetching. The GPU data access kernel employs several thread blocks, each with several threads to simultaneously construct NVMe commands and enqueue them into the submission queues. Subsequently, each submission queue \ding{178} has a dedicated thread to ring the doorbells located in the controller of the NVMe SSD, informing the NVMe SSD that there are new data access requests to process. The NVMe controller \ding{179} retrieves the data and transfers the data to the required addresses in the GPU global memory using DMA (details will be introduced in Section \ref{subsec:optNVMe}). Meanwhile, the GPU calculates gradients for the remaining edge buckets during data exchange. 
For instance, in Figure \ref{fig:workflow}, when the GPU has completed the calculation of edge buckets $\{(0,0),(0,1),(1,0)\}$, the CPU launches a data access kernel on the GPU to exchange $E_0$ and $S_0$ with $E_2$ and $S_2$. This is because $E_0$ and $S_0$ will not be used for the computation of next edge buckets (i.e., $\{(1,1),(1,2),(2,1),(2,2)\}$), instead, $E_2,S_2$ will be used. At the same time of this partition exchange between the GPU and NVMe SSD, the GPU continues computing the edge bucket $(1,1)$, which will not be affected by this exchange. 
{\sf Legend} implements the kernel parallelism of data access and edge bucket computation through CUDA streams. By parallelizing the data access kernel and the computing kernel, the data transfer overhead can be hidden in the computation of the remaining edge buckets, achieving overall performance improvement. Furthermore, GPU waiting time is eliminated by prefetching the data that needs to be calculated, thereby increasing GPU utilization. 
\section{Edge Bucket Ordering}
\label{sec:order}
In this section, we illustrate our proposed partition loading order and the corresponding edge bucket iteration order to enhance the overlap of I/O and computation. 

\begin{figure}
    \centering
    \includegraphics[width=0.48\textwidth]{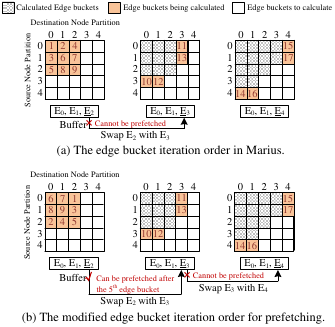}
    \caption{Partition loading order in {\sf Marius}. The numbers inside the edge buckets denote their calculated order. }
    \label{fig:beta}
\end{figure}
As discussed in Section \ref{sec:workflow}, when the GPU completes the training of all edge buckets associated with the current partitions in its global memory, it has to wait for the transfer of the next partitions, leading to reduced GPU utilization. If embeddings and optimizer states of the next partition are prefetched into the GPU's global memory before being used, the GPU can perform calculations for the subsequent batch without incurring waiting time. {Although prefetching is supported in existing graph embedding systems~\cite{mohoney2021marius}, their partition loading orders focus on reducing the node partition exchange counts, missing many opportunities to overlap the I/O and computation during prefetching. }

We illustrate the issue with {\sf Marius}'s loading order through an example. As depicted in Figure \ref{fig:beta}(a), the memory buffer is initialized with partitions $\{E_0, E_1, E_2\}$. {\sf Marius} computes the edge buckets in the order of $\{(0,0),\\(0,1),(1,0),(0,2),(2,0),(1,1),(1,2),(2,1),(2,2)\}$. Next, partition $E_2$ is evicted, and $E_3$ will be swapped in. However, $E_3$ cannot be prefetched at this time because there is no remaining edge bucket in the GPU memory to calculate while fetching $E_3$. 
{The order in which node partitions are loaded or evicted is noted as the {\em partition loading order} throughout the rest of this paper}. 

To achieve prefetching, the edge bucket iteration order can be adjusted to $\{\{(0,2),(2,0),(1,2),(2,1),(2,2)\},\\ \{(0,0),(0,1),(1,0),(1,1)\}\}$ as shown in Figure \ref{fig:beta}(b). Th-
is exchange prioritizes the computation of edge buckets related to the partition $E_2$ before swapping it out for $E_3$. Thus, we can exchange $E_2$ with $E_3$ during calculating the remaining edge buckets (i.e., $\{(0,0),(0,1),(1,0),\\(1,1)\}$). However, this adjustment is only applicable in the initial buffer state. In the subsequent buffer state, as depicted in the second subfigure in Figure \ref{fig:beta}(b), the edge buckets are all related to $E_3$, hindering the eviction of $E_3$ and the prefetching of the next partition. The underlying issue is that the partition swapped in during the previous buffer state is immediately evicted in the next buffer state, leaving insufficient time for prefetching. Although {\sf GE$^2$}~\cite{zheng2024ge2} employs another strategy to reduce the I/O times (i.e., node partition exchange counts), it also fails to support prefetching due to a similar problem. 

From the preceding discussion, we observe that partition prefetching can be achieved if the recently loaded partition is not immediately swapped out of the GPU memory buffer. During the partition exchange, we calculate the edge buckets unrelated to the exchanged partitions, ensuring the overlap of computation and data transfer. {In our problem setting, we assume the number of partitions is $n (n>3)$ and the buffer capacity (the number of partitions that can be loaded) in the GPU's global memory is \textbf{fixed at 3} for problem simplification. This configuration enables training on datasets of any size, as the value of $n$ is arbitrary.} To identify a loading order that supports prefetching, we first define the concept of a \textit{Prefetching Supported Order}. 

\begin{definition}[\textbf{Prefetching Supported Order}]
    \label{def:prefetch}
    A \textit{Prefetching Supported Order} is a node partition loading order such that there is at least one edge bucket not related to the node partition scheduled for eviction in each buffer state. 
\end{definition}

According to Definition \ref{def:prefetch}, in each buffer state, we can first compute the edge buckets related to the partition that will be evicted and simultaneously load a new partition during the computation of the edge buckets unrelated to this partition. The loading orders in {\sf Marius} and {\sf GE$^2$} do not qualify as a {\em Prefetching Supported Order}, as they have no edge bucket unrelated to the partition that will be evicted in most buffer states. To effectively identify the {\em Prefetching Supported Order}, we explore its properties in Theorem \ref{theo:1}. 

\begin{theorem}
    \label{theo:1}
    For the buffer capacity of 3, a partition loading order is classified as a Prefetching Supported Order if it satisfies two properties: (1) The partition that has just been swapped in each buffer state will not be immediately evicted in the subsequent buffer state. (2) Any two partitions may appear concurrently in multiple buffer states, but only in consecutive buffer states. 
\end{theorem}

\begin{proof}
    Without loss of generality, we consider the consecutive buffer states $\{E_i,\textcolor[rgb]{0.29,0.64,0.78}{E_j},\underline{E_k}\},\{E_i,\underline{E_j},\textcolor[rgb]{0.29,0.64,0.78}{E_l}\}, \{E_i,\textcolor[rgb]{0.29,0.64,0.78}{E_m}\\,E_l\}$, where the \textcolor[rgb]{0.29,0.64,0.78}{blue} color and the \underline{underline} denote the loaded node partition in the current buffer state and the partition to be evicted, respectively. $E_i$ and $E_j$ appear concurrently in the first two consecutive buffer states. For the buffer state $\{E_i,\textcolor[rgb]{0.29,0.64,0.78}{E_j},\underline{E_k}\}$, the edge buckets related to $E_i$ and $E_k$ have been calculated during loading $E_j$. $E_k$ is the node partition in the current buffer state that is going to be evicted. The edge buckets $\{(i,j),(j,i)\}$ that are not related to $E_k$ must not have been calculated due to the property (2). Otherwise, if the property (2) is not satisfied, $E_i$ and $E_j$ can appear concurrently in a previous buffer state. Their corresponding edge buckets $(i,j)$ and $(j,i)$ have already been calculated at that buffer state. Therefore we cannot exchange $E_k$ with $E_l$ while calculating other edge buckets at buffer state $\{E_i,\textcolor[rgb]{0.29,0.64,0.78}{E_j},\underline{E_k}\}$, leading to the failure of prefetching. Similarly, for next buffer state $\{E_i,\underline{E_j},\textcolor[rgb]{0.29,0.64,0.78}{E_l}\}$, $E_j$ is going to be evicted. Edge buckets $\{(i,l),(l,i)\}$ are not related to $E_j$ and have not been calculated. We can evict $E_j$ and prefetch the next node partition while calculating $\{(i,l),(l,i)\}$. 
\end{proof} 

It is important to note that the impact of the property (2) in Theorem \ref{theo:1} is minimal in practical applications. Without considering property (2), only 4 out of 36 I/O times fail to support prefetching for 12 partitions, as demonstrated in experiments. Therefore, we exclude property (2) from the algorithm design. Our objective is to design an efficient algorithm to find an order that satisfies property (1) while minimizing the I/O times. We adopt the same swapping strategy as {\sf Marius}~\cite{mohoney2021marius}, which allows a single partition to be swapped in each buffer state. Generating an order that satisfies property (1) is straightforward. However, identifying an order that meets property (1) while minimizing I/O times is an NP-hard problem, as proved in Theorem \ref{theo:2}. 

\begin{theorem}
    \label{theo:2}
    With $n$ partitions and a buffer capacity of 3, the problem of identifying an order that meets property (1) while minimizing I/O times is NP-hard.     
\end{theorem}

\begin{proof}
    {We demonstrate that the problem is NP-hard via a reduction from the covering design problem~\cite{gordon1995new}, a well-known NP-hard problem. Specifically, an instance of the covering design problem with parameters $(n,3,2)$, which seeks the minimum number of 3-element subsets (blocks) covering all $C_n^2$ pairs, is mapped to our problem as follows. The buffer is first initialized with a buffer state, corresponding to an initial block in the covering design problem. A node partition in the buffer is subsequently swapped with a partition out of the buffer. Each exchange operation above generates a new block in the covering design problem. The requirement that all pairs of node partitions must coexist in the buffer is equivalent to ensuring all $C_n^2$ pairs are covered by the sequence of blocks. 
    A covering design problem solution with $m$ blocks implies a valid exchange sequence of $m-1$ steps, as each block requires one exchange. Conversely, an exchange sequence of length $k$ produces $k+1$ blocks covering all pairs. 
    To address the constraint that is introduced in the property (1) of Theorem \ref{theo:1}, intermediate blocks may be inserted (e.g., exchanging two elements sequentially), which only polynomially inflates the sequence length and preserves the reduction’s validity. Since the covering design problem is NP-hard, our problem is also NP-hard. }
\end{proof}

\begin{figure*}
    \centering
    \includegraphics[width=0.8\linewidth]{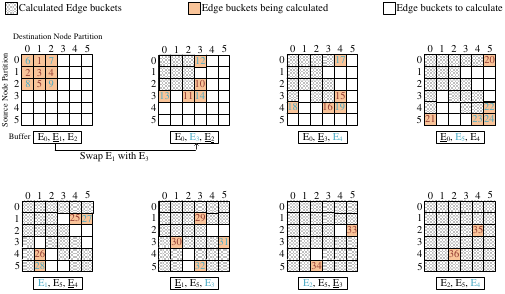}
    \caption{Order for prefetching in {\sf Legend}. The numbers inside the edge buckets denote their calculated order. The \textcolor[rgb]{0.29,0.64,0.78}{blue} color indicates the edge buckets that can be calculated while prefetching the next partition. \textcolor[rgb]{0.29,0.64,0.78}{$E_i$} is the prefetched partition. $\underline{E_j}$ is the node partition to be evicted in the next buffer state. }
    \label{fig:order}
\end{figure*}

The NP-hardness of this problem motivates us to devise an efficient heuristic algorithm. 
To this end, we propose a column separation covering strategy to generate an order supporting prefetching while minimizing I/O times within one second. The key idea of the loading order is to {\em sequentially cover each column of edge buckets, greedily maximizing coverage in each column}. Figure \ref{fig:order} depicts an example of our proposed node partition loading order and edge bucket iteration order with 6 node partitions. 

Initially, we cover edge buckets in the first column by swapping in each node partition in order of their ID. For example, we cover edge buckets $\{(0,0),(1,0),(2,0),\\(3,0),(4,0),(5,0)\}$ in column 0 by sequentially swapping in node partitions $\{E_0, E_1, E_2, E_3, E_4, E_5\}$ (the first four buffer states in Figure \ref{fig:order}). For subsequent columns, we swap in node partitions starting from the maximal ID in the current buffer state. If all edge buckets from the maximal ID to $n$ are covered, we then switch to the minimal ID. 
For instance, after transitioning to column 1 with the buffer state $\{E_1,E_5,E_4\}$, we start with the node partition having the minimal ID to swap in, which is $E_3$. Subsequently, the buffer state changes to $\{E_1,E_5,E_3\}$ and edge buckets $\{(3,1),(1,3),(3,5),(5,3)\}$ are covered. 
Since all edge buckets in column 1 are covered, we move to column 2 after the buffer state $\{E_1,E_5,E_3\}$ by loading node partition $E_2$, and edge buckets $\{(5,2),(2,5)\}$ are covered. 
For column 2, we again start with the minimal ID, which is 4, to swap in. Therefore, the buffer state changes to $\{E_2,E_5,E_4\}$, and all edge buckets have been covered by now. 
The procedures are formalized in Algorithm \ref{alg:loading}. 

\setlength{\textfloatsep}{0pt}
\makeatletter
\patchcmd{\@algocf@start}
  {-1.5em}
  {0pt}
  {}{}
\makeatother
\begin{algorithm}[tp]
    \caption{Node partition loading order}
    \label{alg:loading}
    \LinesNumbered
    \setstretch{0.82}
    \KwIn{Node partitions $n$, buffer capacity $3$}
    \KwOut{Buffer states in order.}
    $EdgeBuckets \leftarrow \{0\}_{n*n}$, $BufStates \leftarrow \{\}$;\\
    $CurCol \leftarrow 0$;\\
    $BufStates.append(\{0,1,2\})$;\\ 
    $Buf \leftarrow BufStates[-1]$;\\ 
    \For{$i$ in range(3,n)}{
        $Buf \leftarrow Buf - \{i-2\} + \{i\}$;\\
        $BufStates.append(Buf)$;\\
    }
    Set the covered edge buckets to 1;\\
    \While{$sum(EdgeBuckets)<n^2$}{
        $ToEvict \leftarrow -1, ToLoad \leftarrow -1$;\\
        \If{$sum(EdgeBuckets[CurCol])=n$}{
            $Buf \leftarrow Buf - \{CurCol\} + \{CurCol+1\}$;\\
            $BufStates.append(Buf)$;\\
            Set the covered edge buckets to 1;\\
            \If{sum(EdgeBuckets[CurCol][max(Buf)+1:n]) $<$ n-max(Buf)}{
                $ToEvict \leftarrow max(Buf)$;\\
            }
            \Else{
                $ToEvict \leftarrow min(Buf)$;\\
            }
            $CurCol \leftarrow CurCol + 1$;\\
        }
        \Else{
            $ToEvict \leftarrow id \in BufStates[-1]\cap BufStates[-2]$ and $id\neq CurCol$;\\
        }
        $BeginID \leftarrow (Buf-\{ToEvict, CurCol\}) + CurCol + 2$;\\
        $ToLoad \leftarrow$ the id that covers the most edge buckets from $BeginID$ to $BeginID-1$;\\ 
        $Buf \leftarrow Buf - \{ToEvict\} + \{ToLoad\}$;\\
        $BufStates.append(Buf)$;\\
        Set the covered edge buckets to 1;\\
    }
    \Return{$BufStates$}
\end{algorithm}
\setlength{\textfloatsep}{12pt plus 2pt minus 2pt}

\vspace{0.7mm}
In Algorithm \ref{alg:loading}, we first generate the buffer states related to node partition 0 (lines 3-6). In the {\em while} loop within lines 8-24, if all edge buckets in the current column $CurCol$ have been covered, we advance to the next column by swapping node partition $E_{CurCol}$ with $E_{CurCol+1}$ and mark the corresponding edge buckets as covered (lines 11-13). Then we select the maximal ID in the current buffer as the node partition to evict, provided that subsequent IDs have corresponding edge buckets that remain uncovered. Otherwise, we opt for the minimal ID (lines 13-17). If the edge buckets in the current column have not been fully accessed, we evict a node partition that was not just swapped in the last buffer state (line 19). Finally, we greedily select a node partition that covers the most edge buckets from the $BeginID$ to swap in (lines 20-24). 

\setlength{\textfloatsep}{0pt}
\begin{algorithm}[tp]
    \caption{Edge buckets iterating order}
    \label{alg:iterating}
    \LinesNumbered
    \setstretch{0.82}
    \KwIn{Buffer states $BufStates$, node partitions $n$}
    \KwOut{Edge buckets iterating order.}
    $EdgeBuckets \leftarrow \{0\}_{n*n}$, $BufStates \leftarrow \{\}$;\\
    $IterOrder \leftarrow \{(0,1),(1,1),(1,0),(1,2),(2,1)\}$;\\
    Set the covered edge buckets to 1;\\
    $LoadedPar \leftarrow 3$;\\
    \For{i in range(len(BufStates)-1)}{
        $ToEvict \leftarrow BufStates[i] - BufStates[i+1]$;\\
        \For{$k \in BufStates[i]-\{LoadedPar\}$}{
            \If{EdgeBuckets[ToEvict][k]=0}{
                $EdgeBuckets[ToEvict][k]=1$;\\
                $IterOrder.append((ToEvict,k))$;\\
            }
            \If{EdgeBuckets[k][ToEvict]=0}{
                $EdgeBuckets[k][ToEvict]=1$;\\
                $IterOrder.append((k,ToEvict))$;\\
            }
        }
        \If{EdgeBuckets[ToEvict][LoadedPar]=0}{
            $EdgeBuckets[ToEvict][LoadedPar]=1$;\\
            $IterOrder.append((ToEvict,LoadedPar))$;\\
        }
        \If{EdgeBuckets[LoadedPar][ToEvict]=0}{
            $EdgeBuckets[LoadedPar][ToEvict]=1$;\\
            $IterOrder.append((LoadedPar,ToEvict))$;\\
        }
        $LoadedPar \leftarrow BufStates[i+1] - BufStates[i]$;\\
    }
    \Return{IterOrder}
\end{algorithm}
\setlength{\textfloatsep}{12pt plus 2pt minus 2pt}

Algorithm \ref{alg:iterating} generates the edge bucket iteration order according to the output of Algorithm \ref{alg:loading}. It first covers the edge buckets related to the node partition scheduled for eviction in the next buffer state (lines 7-13). Subsequently, it calculates the edge buckets related to both the node partition that will be evicted and the one that was just swapped in (lines 14-19). 

\noindent \textbf{Example.} Figure \ref{fig:order} exhibits an example of the node partition loading order output by Algorithm \ref{alg:loading} and the edge buckets iterating order generated by Algorithm \ref{alg:iterating}. For the initial buffer state $\{E_0,E_1,E_2\}$, the edge buckets related to $E_1$ (i.e., $\{(0,1),(1,0),(1,1),(1,2),(2,1)\}$) are first calculated. 
Subsequently, $E_1$ can be swapped with $E_3$ in advance. The remaining edge buckets in the initial buffer state (i.e., $\{(0,0),(0,2),(2,0),(2,2)\}$) are calculated during the data exchange. 
After the calculation of the initial buffer state, the second buffer state $\{E_0,E_3,E_2\}$ first calculates edge buckets related to $E_2$. Similarly, $E_2$ is replaced with $E_4$ in advance right after the calculation of $\{(2,3),(3,2)\}$. And the remaining edge buckets are calculated at the same time. 
The following buffer states adopt a similar process to prefetch the next node partition while calculating the remaining edge buckets, significantly reducing the I/O overhead. 

Although prefetching hides I/O overhead in the computation, it also raises the problem of whether the I/O overhead can be completely covered. To this end, Theorem \ref{theo:3} discusses this problem and proves that it has to do with the dataset characteristics. 

\begin{theorem}
\label{theo:3}
    Using the loading order generated by Algorithm \ref{alg:iterating}, the I/O overhead can be completely covered by the computation when ${\scriptstyle \frac{|E|}{|V|^2}\ge \frac{96d^2}{Mt(w+r)}}$, where $|E|$ and $|V|$ are the number of edges and nodes, $d$ is the embedding dimension, $M$ is the buffer size in the global memory of GPU, $t$ is the average computing time of an edge, $w$ and $r$ are the writing and reading throughput between the GPU and NVMe SSD. 
\end{theorem}

\begin{proof}
    Suppose the node embeddings are divided into $n$ partitions. For each edge bucket, the average number of edges is ${\scriptstyle \frac{|E|}{n^2}}$. Consequently, the average time to calculate an edge bucket is ${\scriptstyle t*\frac{|E|}{n^2}}$. An exchange of a partition, including writing the old partition into the NVMe SSD and loading the new one into the GPU buffer. Each partition contains embeddings and optimizer states, whose total size is $2*P$. As a result, an exchange of a partition requires time of ${\scriptstyle \frac{2*P}{w+r}}$. Following the order output by Algorithm \ref{alg:iterating} ensures that there are at least 2 edge buckets for computing during partition exchange. So if the inequality ${\scriptstyle 2t*\frac{|E|}{n^2}\ge 2*\frac{P}{w+r}}$, the I/O overhead can be completely covered by the calculation of edge buckets. As the buffer size is $M$ and it can contain 3 partitions in our hypothesis, $P$ can be calculated as ${\scriptstyle \frac{M}{6}}$ and the minimum $n$ can be calculated as ${\scriptstyle \frac{|V|*d*4*2}{M/3}}$, where $4$ denotes the bytes of a float type. Substituting $P$ and $n$ into the inequality yields ${\scriptstyle \frac{|E|}{|V|^2}\ge \frac{96d^2}{Mt(w+r)}}$. 
   
\end{proof}

Theorem \ref{theo:3} displays the condition that I/O overhead can be completely covered by the computation using our prefetching strategy. {\sf Legend} has the metrics of $t\approx10^{-7}s$, $w\approx2G/s$ and $r\approx3G/s$ in our experimental setting (Section \ref{subsec:setting}). With $M=15G$ and $d=100$, the I/O overhead can be completely covered by the computation if ${\scriptstyle \frac{|E|}{|V|^2}\ge 10^{-7}}$. {For instance, the Twitter dataset has $|E|=1.46\times 10^9$ and $|V|=4.16\times 10^7$ (see Table \ref{tab:dataset} for detailed information). The condition ${\scriptstyle \frac{|E|}{|V|^2}\approx 8\times 10^{-7}\ge 10^{-7}}$ is satisfied, confirming that I/O overhead can be completely hidden by computation.}

\noindent {\textbf{Discussion.} Our proposed ordering strategy focuses on single-GPU optimization, which is valuable especially for users such as researchers and students who have limited GPU resources. The experimental results in Section \ref{subsec:overall} and Table \ref{tab:multigpus} also demonstrate that Legend is a cost-effective solution. Extending the ordering strategy to multi-GPU settings introduces additional challenges and is left to future work. In this paragraph, we only discuss possible strategies for multi-GPU environments as a potential direction for future research. Prior work supporting multi-GPU, such as DGL-KE~\cite{zheng2020dgl}, employs the METIS algorithm to partition a graph and assign subgraphs to GPUs. As shown in Figure \ref{fig:multiGPU}, with METIS partitioning, most edges reside within diagonal blocks, which aligns well with our prefetch-aware ordering strategy. Specifically, each subgraph on the GPU can be further partitioned by node IDs, enabling the straightforward application of the ordering strategy in Legend for diagonal blocks. The remaining partitions require a new ordering strategy that supports prefetching, as the partition ID on their horizontal and vertical axes is not identical. Another option is to only prefetch partitions for diagonal blocks, which have already covered the majority of edges. In this way, we can achieve partial prefetching in multi-GPU scenarios. } 

\begin{figure}
    \centering
    \includegraphics[width=0.32\textwidth]{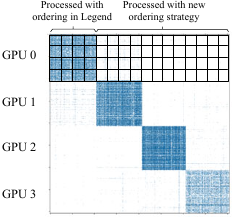}
    \caption{{Adjacent matrix of a large graph after applying METIS partitioning~\cite{zheng2020dgl}.} }
    \label{fig:multiGPU}
\end{figure}
\section{Optimizations on GPU Direct Access to SSD}
\label{subsec:optNVMe}
In this section, we introduce our proposed optimizations for the GPU-SSD direct access mechanism, {including batch enqueue, full-coalesced doorbell ringing, and batch polling techniques.} These techniques are specifically designed for graph embedding workloads and enhance the bandwidth between the GPU and SSD.  


{Previously, access to NVMe SSD relied on the kernel I/O stacks of the operating system, which involve context switching, data copying, interrupts, resource synchronization, etc. 
As the latency of storage devices decreases, the CPU software stack becomes a bottleneck for I/O access~\cite{qureshi2023gpu}. Consequently, customized NVMe SSD drivers such as SPDK have emerged to move all of the necessary operations into userspace~\cite{yang2017spdk}, reducing the CPU software stack's overhead.} 
Recently, to achieve high-performance direct access between the GPU and NVMe SSD, research has shifted towards offloading I/O tasks from the CPU and reconstructing the user-level I/O stack on the GPU, aiming to reduce the stacks' overhead and enhance throughput by leveraging the massive parallelism of GPU threads. 

Among these, {\sf BaM} achieves the state-of-the-art performance~\cite{qureshi2023gpu}. However, {\sf BaM} is designed to handle general workloads across various scenarios, incorporating complex mechanisms including parallel queue management strategies, atomic operations, caching strategies, etc. 
{The enqueue operation is designed to be serial by atomic operations to avoid concurrency conflicts, and the doorbell ringing strategy only coalesces doorbell writes of some threads to overcome complex read/write workload, leading to redundant overhead of lock and doorbell ringing for data transfer in graph embedding. }
Moreover, {\sf BaM} employs numerous thread blocks on the GPU to achieve high throughput between the GPU and NVMe SSD, which consumes valuable GPU resources and hinders the simultaneous execution of the data access kernel and computing kernel. 

{Other methods employ GPU-SSD direct access for specific applications such as DNN and GNN training~\cite{bae2021flashneuron,jeongmin2024accelerating}. These methods leverage existing access mechanisms implemented in {\sf GDRCopy}~\cite{refgdrcopy} or {\sf BaM}~\cite{qureshi2023gpu}, without optimizing the underlying access mechanism specific for the I/O patterns of various applications. 
For graph embedding, the embeddings and optimizer states are stored continuously in large partitions in SSD, leading to redundant door ringing and locks during data access when employing existing data access mechanisms such as {\sf BaM}. }
To implement a lightweight yet high-throughput NVMe SSD access kernel, we analyze the specific workload of graph embedding learning and optimize the direct access mechanism. 

In the context of graph embedding learning, embeddings and optimizer states are loaded from NVMe SSD to the GPU buffer only after the GPU has completed the computation of the edge buckets related to the node partitions in the current buffer state. The data loading times are determined once Algorithm \ref{alg:iterating} provides the order. Additionally, the size of the embedding and optimizer states for each node partition is fixed, allowing for sequential access page by page. Such a workload leads to opportunities to reduce the complexity of the queue management mechanism. 

To avoid building complex I/O stacks from scratch, similar to {\sf BaM}, we implement the GPU-SSD direct access driver based on an open-source codebase~\cite{Markussen2021}. We will only introduce our contributions below. To maximize the parallelism and the I/O throughput of the NVMe SSD, we employ multiple NVMe queues and utilize several thread blocks, with each thread block managing one NVMe queue pair. All threads within a thread block can enqueue and dequeue on the corresponding queue pair. This thread block allocation strategy simplifies the management of queue pairs, as synchronization among threads within the same block is straightforward and has low overhead. 

\begin{figure}
    \centering
    \includegraphics[width=0.4\textwidth]{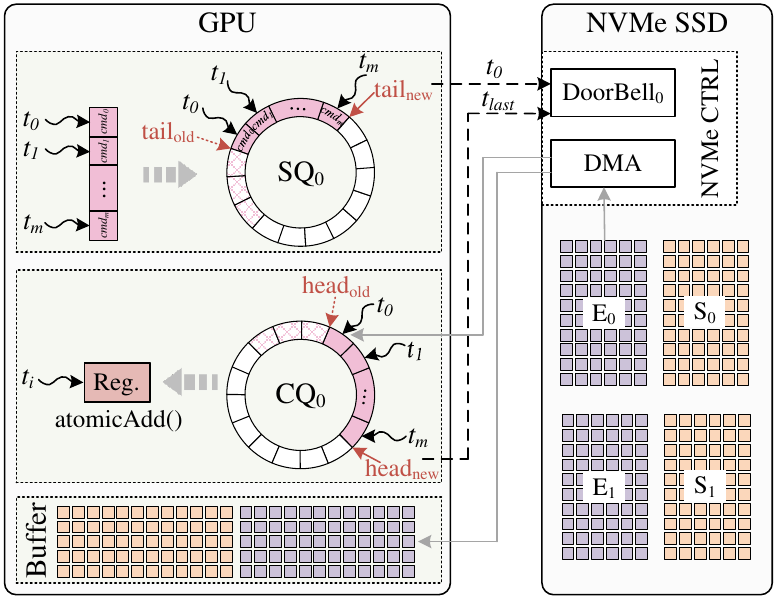}
    \caption{Procedure of the GPU direct access to NVMe SSD.}
    \label{fig:nvme}
\end{figure}

Figure \ref{fig:nvme} depicts our proposed optimized procedure for GPU direct access to NVMe SSD. 
{For clarity, Figure \ref{fig:nvme} exhibits a single thread block with a single queue. In practical implementation, we employ multiple queues to maximize the bandwidth between the GPU and SSD. }
The key idea of our proposed GPU direct access mechanism is to utilize the regular embedding access characteristics to precompute the positions of queue entries and minimize the doorbell ringing time, avoiding the use of locks and atomic operations, as well as reducing the overhead of doorbell ringing. 

During the command construction phase, threads in a block construct read/write commands in parallel, requesting consecutive NVMe addresses. Subsequently, these threads need to insert the commands into the submission queue (SQ). To achieve lock-free enqueue, {we design a {\em batch enqueue} strategy.} Each thread $t_i$ inserts its command into the position $tail_{old}+i$ of the corresponding SQ, where $tail_{old}$ denotes the current head pointer and $i$ denotes the thread ID. 
This enqueue process is parallelized among threads since each thread has a unique and determined position in the SQ. 
The fixed size of embeddings simplifies the enqueue process, allowing parallel operation without the complex data structures and atomic operations in {\sf BaM} for correct enqueueing. 
Following the enqueue process, the tail of SQ is updated to $tail_{new}$ as shown in Figure \ref{fig:nvme}, which is synchronized with the doorbell registers subsequently. 
The doorbell registers in the NVMe controller are write-only, necessitating serial writing from threads. Furthermore, the writing overhead of doorbell registers is high because they are located in the NVMe SSD, and the writing needs to be performed through PCIe. 
As a result, {we design a {\em full-coalesced doorbell ringing} mechanism to reduce the high cost}, where a single thread ($t_0$ in Figure \ref{fig:nvme}) is assigned to ring the doorbell only after all the threads within a thread block ($t_0\sim t_m$) have completed the enqueue process, rather than ringing the doorbell multiple times. 

Once the doorbell rings, the NVMe controller fetches commands from the SQ in the GPU's global memory. The NVMe controller analyzes these commands, retrieves data from the NVMe SSD, and transfers the data to the specified addresses in the GPU buffer according to the commands via Direct Memory Access (DMA). Following this, the NVMe controller inserts an entry into the completion queue (CQ) corresponding to the entry in the SQ.

\begin{figure*}
    \centering
    \includegraphics[width=0.8\textwidth]{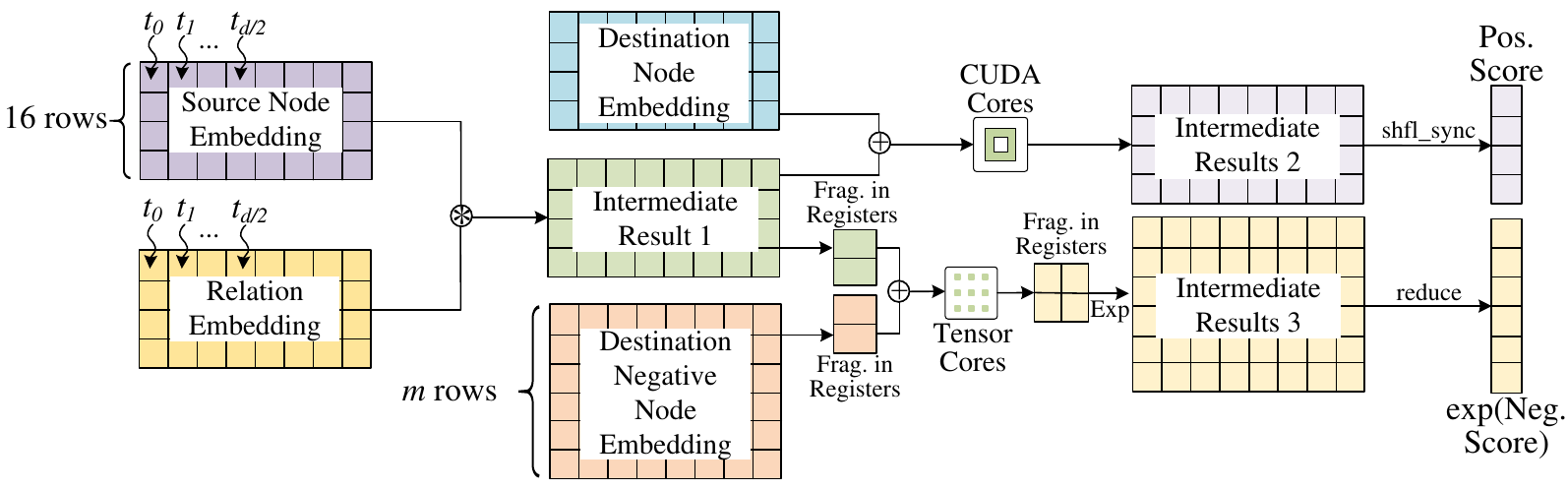}
    \caption{Optimized procedure of the forward phase. }
    \label{fig:forward}
    \vspace{-2mm}
\end{figure*}

To wait for the completion entries in the CQ, existing methods employ a polling strategy, where each thread polls the CQ for the completion entry it enqueues during the enqueue phase, and the doorbell ringing strategy is the same as the enqueue phase. 
However, the na\"ive polling strategy causes redundant overhead to calculate the exact polling positions and ring the doorbell. To reduce the polling overhead, {we design a {\em batch polling} strategy specific for graph embedding data transfer workload.}  
Specifically, during the polling phase, each thread $t_i$ within a thread block checks the position of $head_{old}+i$ in the CQ, where $head_{old}$ is the head pointer of the CQ and $i$ denotes the thread ID. 
We maintain a counter for each CQ in the GPU's shared memory, which is initialized with 0. 
When a thread detects that an entry has been inserted by the NVMe controller, it atomically adds 1 to the counter by using the $atomicAdd()$ operation in CUDA to avoid concurrency conflict. 
The last thread $t_{last}$ to increment this counter updates the head pointer of CQ to $head_{new}$ and rings the doorbell to transfer the updated position of the latest head pointer. 
The atomic operation has a low cost because the counter is located in shared memory, and the number of threads within a block is limited. This polling strategy can not only fully utilize thread parallelism, but also reduce the overhead of doorbell ringing. 

As the size of required embeddings and optimizer states is determined, threads in a thread block repeat the data access procedures synchronously until all the pages of embedding and optimizer states in the NVMe SSD are loaded. This GPU direct access driver and the embedding training kernel run simultaneously on the GPU by utilizing CUDA streams.  



\section{Optimizations on GPU}
\label{subsec:optGPU}
In this section, we describe our optimization techniques for the batch computation on the GPU, which fully exploit the resources on the GPU and significantly enhance GPU utilization. 
{Batch computation is the core process in the graph embedding pipeline, which involves two procedures. For each batch, it first calculates the score $f(\theta_s, \theta_r, \theta_d)$ for positive edges and $f(\theta_{s'}, \theta_{r'}, \theta_{d'})$ for sampled negative edges (see Equation \ref{equa1}). Then it calculates the gradients of the loss function $\mathcal{L}$ (Equation 1) for the embedding values ($\theta_s, \theta_r, \theta_d$, $\theta_{s'}, \theta_{r'}, \theta_{d'}$). The embedding values are finally updated based on the calculated gradients. }

As mentioned above, existing systems often overlook GPU computation optimization, leading to underutilization of the GPU. In graph embedding systems, the time cost can be divided into three parts: CPU processing, CPU-GPU communication, and GPU computing~\cite{zheng2024ge2}. As some tasks on the CPU are offloaded to the GPU and CPU-GPU communication is optimized, the overhead of GPU computing becomes more pronounced. For instance, both {\sf Marius}~\cite{mohoney2021marius} and {\sf GE$^2$}~\cite{zheng2024ge2} have similar GPU computing overhead because they use the same GPU computing engine. However, {\sf GE$^2$} offloads some tasks to the GPU and reduces the CPU-GPU communication cost by customized loading order, which results in the GPU computing becoming the most time-consuming part (more than 1/3). In {\sf Legend}, the communication overhead is further reduced due to the node partition loading order that minimizes the I/O times as well as the support of embedding prefetching, which makes GPU computing the primary bottleneck of graph embedding learning. 

{Different from the computing process of DNNs and GNNs which directly perform matrix multiplication between weights (or adjacent matrics of graphs) and embeddings~\cite{wang2022lightseq2,peng2022sancus}, the computing process of graph embedding is more complex as shown in Equation \ref{equa1}, which involves computation between $\theta_s$, $\theta_r$, $\theta_d$ and their corresponding negative samples, leading to more memory access transactions and intermediate results. 
Moreover, different graph embedding models (i.e., $f$ in Equation \ref{equa1}) have distinct computing processes. For example, the graph embedding model {\sf ComplEx}~\cite{trouillon2016complex} requires cross-calculation of the first half dimensions and the last half dimensions of the embedding.
To optimize the computation for graph embedding, our key idea is to design a customized parallel computational pattern for graph embedding on the GPU to reduce memory access, maximize the reuse of the intermediate results, and leverage various GPU resources, including Tensor cores and registers, to enhance GPU utilization. }

As illustrated in Figure \ref{fig:forward}, we horizontally split a batch of embeddings into several chunks, distributing each chunk across thread blocks in the GPU for simultaneous computation, fully utilizing the high parallelism of the GPU. 
For clarity, Figure \ref{fig:forward} depicts the computing procedure of a single thread block on the GPU. 
Each thread block contains several warps, with $\lceil d/64 \rceil$ warp(s) collaborating to process one row, where $d$ represents the embedding dimension. This design allows the $\lceil d/64 \rceil$ warp(s) to calculate the first half of the elements in a row before processing the last half. It leads to only one access to each element for the cross-calculation between the first and last half elements in some embedding models, such as ComplEx~\cite{trouillon2016complex}, avoiding redundant memory access. 
It is also suitable for other embedding models without cross-calculation. Each thread block handles 16 rows of embeddings in a batch, as subsequent calculations utilize Tensor cores, which necessitate a fixed input size of $16\times 8\times 16$ submatrices in each thread block. 

Threads within a thread block retrieve the corresponding elements from the source node embeddings ($\theta_s$) and the relation embeddings ($\theta_r$), and calculate the intermediate results according to the operator $\otimes$ defined by the embedding model. Following the calculation of $\theta_s \otimes \theta_r$, we obtain Intermediate Results 1, which are stored in the registers of each thread and have not yet been written to the GPU's memory. 
To minimize the access times of global memory, we first use the results stored in registers to compute the positive scores before writing them to global memory. 
Therefore, we first employ a similar parallel strategy to calculate Intermediate Results 2 by the destination node embeddings ($\theta_d$) and Intermediate Result 1 ($\theta_s\otimes\theta_r$) using CUDA cores, with the equation of $(\theta_s\otimes \theta_r)\oplus\theta_d$, where $\otimes$ and $\oplus$ are defined by the adopted embedding model. 
Notably, Intermediate Result 2 remains stored in registers and is distributed among threads. To subsequently calculate the positive scores from Intermediate Result 2, we implement a two-phase reduction strategy, which first reduces elements within threads in each warp using the inter-thread data exchange function \textit{\_\_shfl\_sync()}, and second reduces the elements within warps in each row. 
The two-phase reduction strategy leverages shared memory only in the second phase, thereby reducing memory access overhead. 
During the calculation of positive scores, global memory access only happens when the data is loaded at the beginning and the positive scores are written at the end. 
Consequently, we improve the efficiency of positive score calculation by optimizing computation and memory access. 

For the efficient calculation of negative scores, we design an optimized kernel specifically for multiplication-based embedding models such as ComplEx~\cite{trouillon2016complex} and DistMult~\cite{yang2015embedding}, whose $\oplus$ is a multiplication operation in $(\theta_s \otimes \theta_r)\oplus \theta^{\prime}_d$. Normally, a source node embedding is required to perform element-wise multiplication with a group of negative node embeddings, followed by a reduction of elements in each row to calculate the negative scores. Given the substantial number of multiplication operations, we utilize Tensor cores, which can execute fixed-size matrix multiplications within a single clock cycle. We adopt the TF32 data type for matrix multiplication in Tensor cores, requiring input matrices to be sized $16\times 8$. As shown in Figure \ref{fig:forward}, in each thread block, the Intermediate Result 1 contains exactly 16 rows, facilitating horizontal iteration. In a thread block, we employ multiple warps to iterate over the negative node embeddings horizontally, with each warp handling 16 rows of the negative embeddings. Each warp loads a fragment of embeddings into registers and feeds them into the Tensor cores to get the multiplication results. Considering that we need to use the exponent results of the negative scores in the loss and gradients calculation, we perform the exponent operation in advance in registers before writing the results to global memory (Intermediate Result 3 in Figure \ref{fig:forward}), which further reduces redundant memory access. Finally, we employ the reduction API in {\sf libtorch} to reduce the elements in Intermediate Result 3. 

During gradient computing, we reuse the Intermediate Result 1, 2, and 3 in Figure \ref{fig:forward} to eliminate redundant calculations. We also apply the same parallel strategy as in Figure \ref{fig:forward} to compute the gradients efficiently on the GPU, which shares a similar computing process. The parallel strategy, memory access strategy, and the intermediate results reusing perform collaboratively to enhance GPU utilization and enable high-performance gradient computation on large datasets.

\section{Experiments}
\label{sec:evaluation}
In this section, we evaluate the performance of our proposed \textsf{Legend} and conduct a comparative evaluation with state-of-the-art graph embedding systems. Source code of {\sf Legend} is publicly available~\footnote{https://github.com/ZJU-DAILY/Legend}.

\subsection{Experiment Settings}
\label{subsec:setting}
\noindent \textbf{Datasets}. For comprehensive evaluations, we use 4 da-
tasets with varying volumes, previously employed in related works~\cite{zheng2020dgl,mohoney2021marius,zheng2024ge2}. Table \ref{tab:dataset} summarizes their properties, where \textit{FB} and \textit{FM} are multi-type knowledge graphs, while \textit{LJ} and \textit{TW} are social networks without relation types. In Table \ref{tab:dataset}, \textit{Dim.} denotes the embedding dimension, and \textit{Size} indicates the storage requirements for embeddings and optimizer states. Each dataset is divided into training, test, and validation subsets for embedding training and evaluation.

\begin{table}[tbp]
\renewcommand{\arraystretch}{1.2}
    \centering
    \small
    \caption{Details of Datasets. }
    \begin{tabular}{p{2.68cm}<{\centering}p{0.6cm}<{\centering}p{0.8cm}<{\centering}p{0.6cm}<{\centering}p{0.4cm}<{\centering}p{0.8cm}<{\centering} }
        \toprule
         Graphs & $|V|$ & $|E|$ & $|R|$ & \textit{Dim.} & \textit{Size} \\
         \toprule
         FB15k (\textit{FB}) & 15k & 592k & 1345 & 100 & 13MB \\
         LiveJournal (\textit{LJ}) & 4.8M & 68M & - & 100 & 3.8GB \\
         Twitter (\textit{TW}) & 41.6M & 1.46B & - & 100 & 32GB \\
         Freebase86M (\textit{FM}) & 86.1M & 304.7M & 14824 & 100 & 68GB \\
        \bottomrule
    \end{tabular}
    \label{tab:dataset}
\end{table}

\begin{figure*}
\centering 

    \subfigure[\textit{{FB}}]{
    \includegraphics[width=0.225\textwidth]{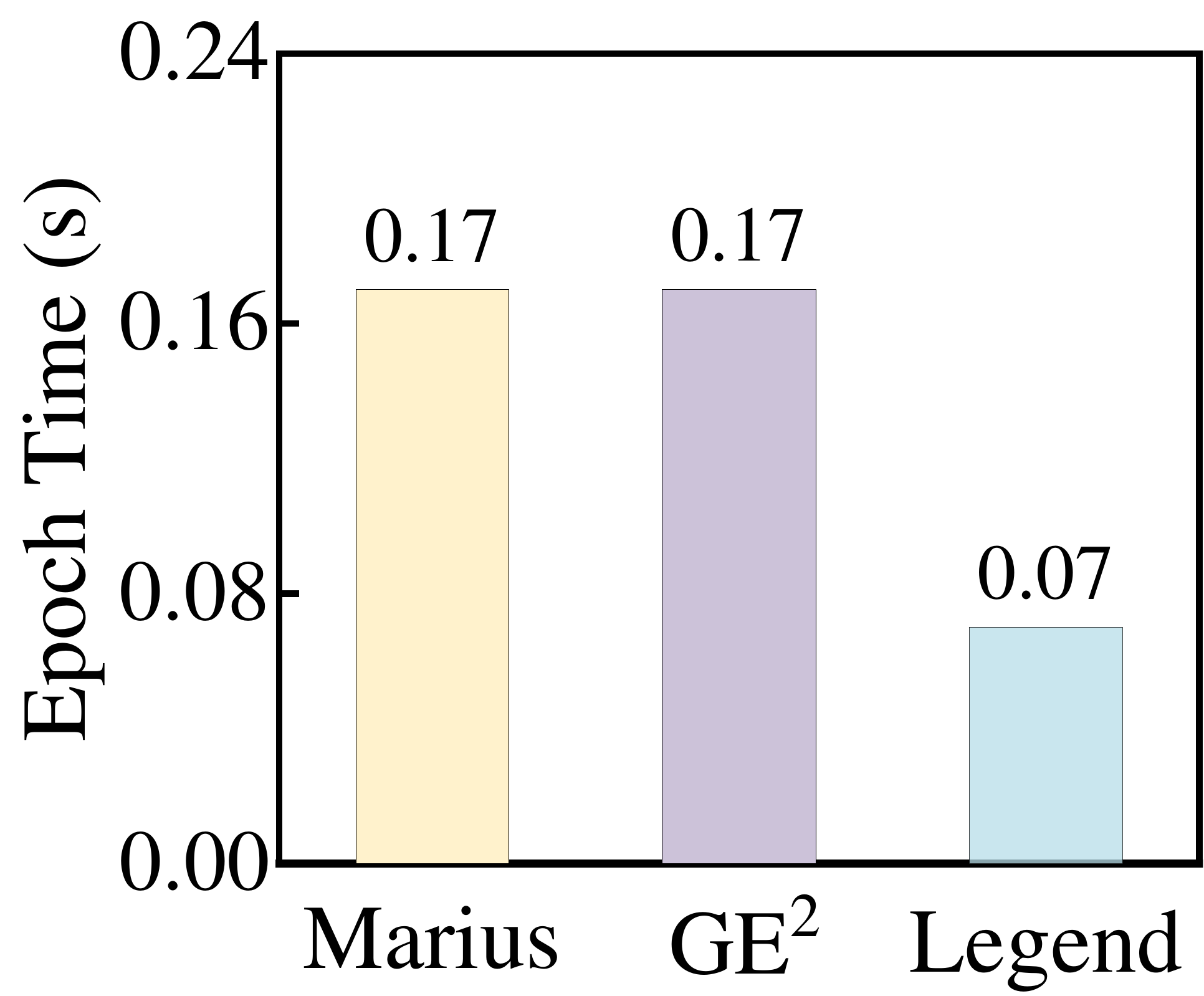}}
    \hspace{2mm}
    \subfigure[\textit{LJ}]{
    \includegraphics[width=0.213\textwidth]{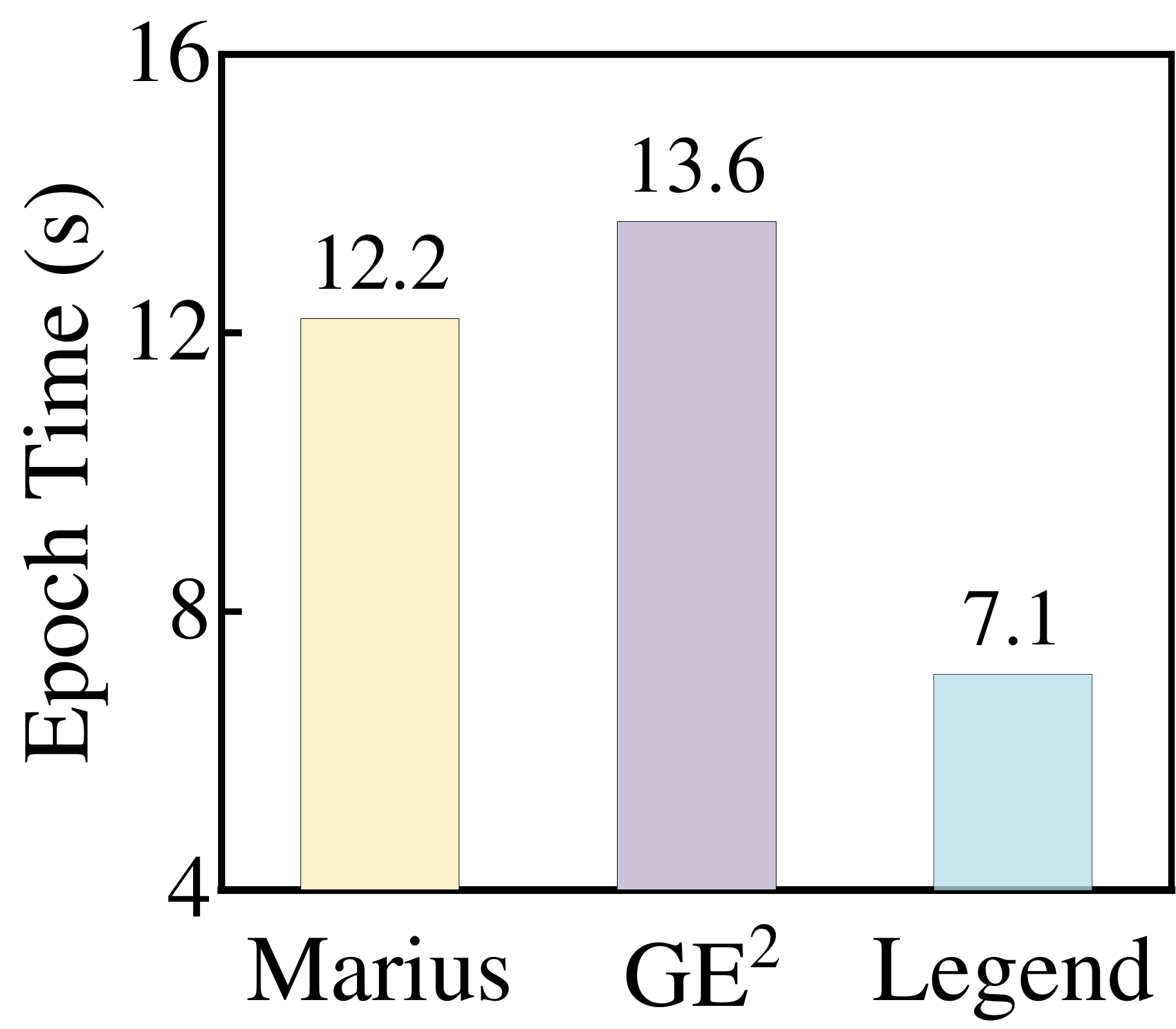}}
    \hspace{2mm}
    \subfigure[\textit{{TW}}]{
    \includegraphics[width=0.23\textwidth]{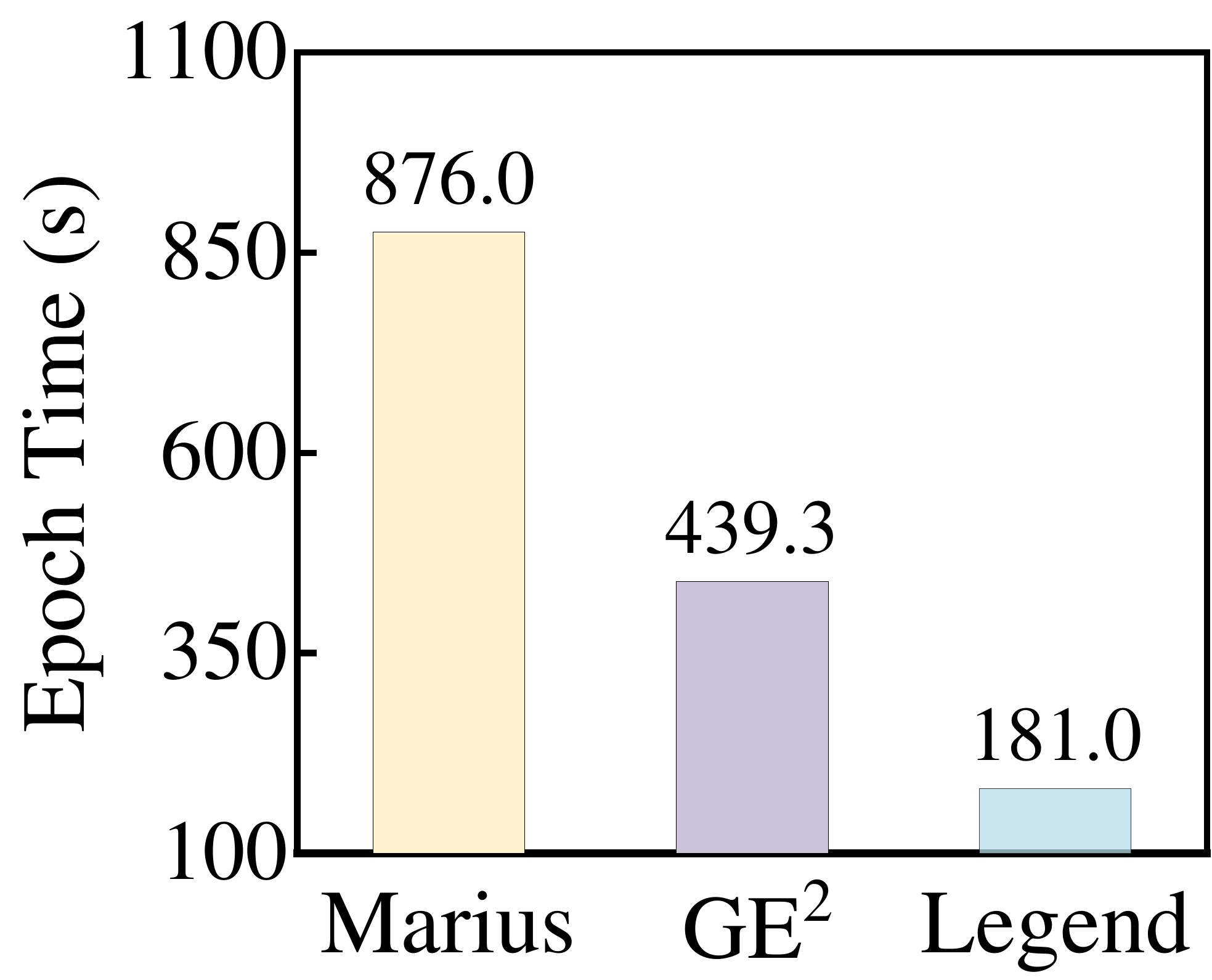}}
    \hspace{2mm}
    \subfigure[\textit{{FM}}]{
    \includegraphics[width=0.22\textwidth]{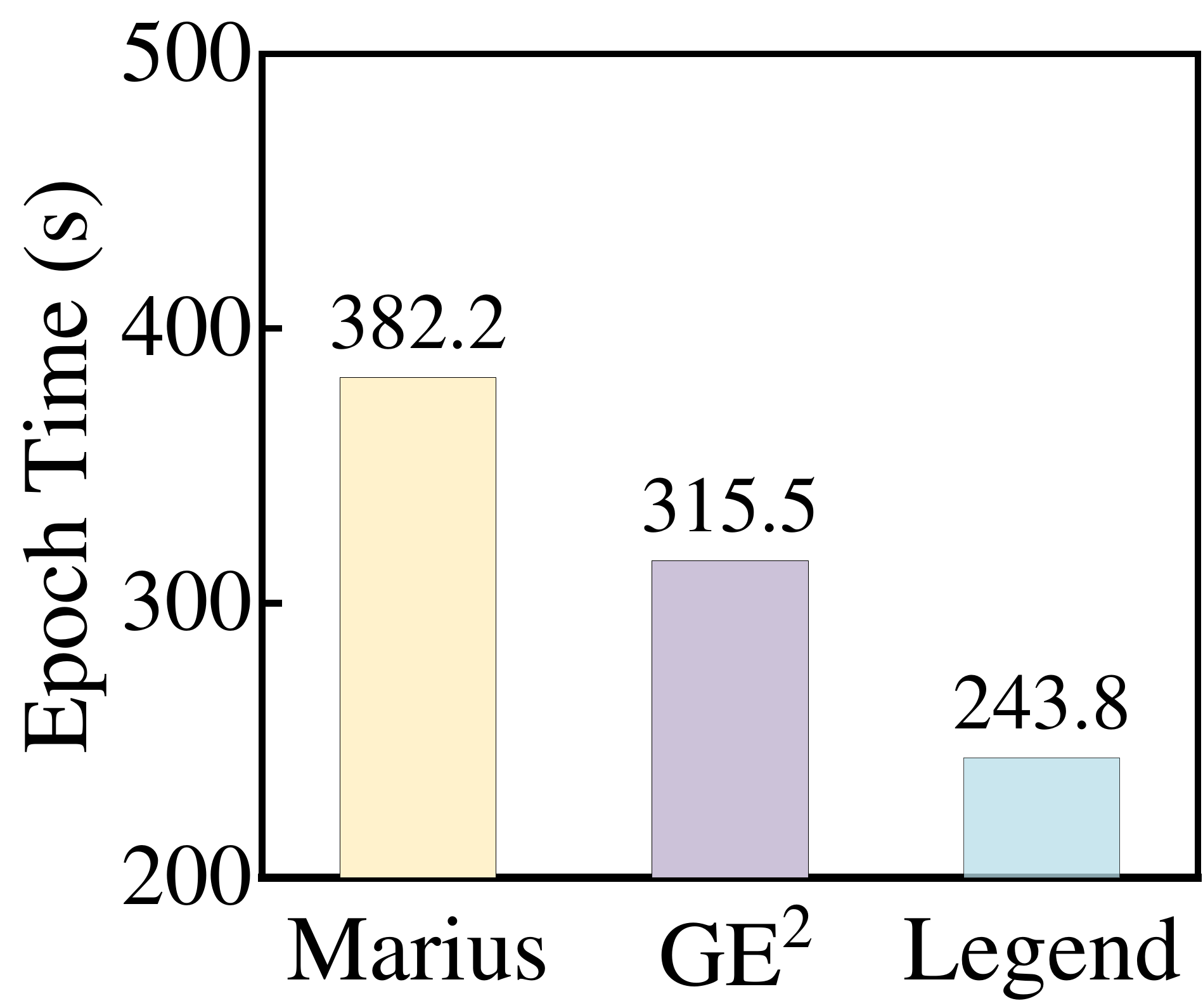}}

    \caption{Comparison of the average epoch time using a single GPU.}
    \label{fig:over_performance_time}
\end{figure*}

\begin{figure*}
\centering 
    \subfigure[\textit{{FB}}]{
    \includegraphics[width=0.22\textwidth]{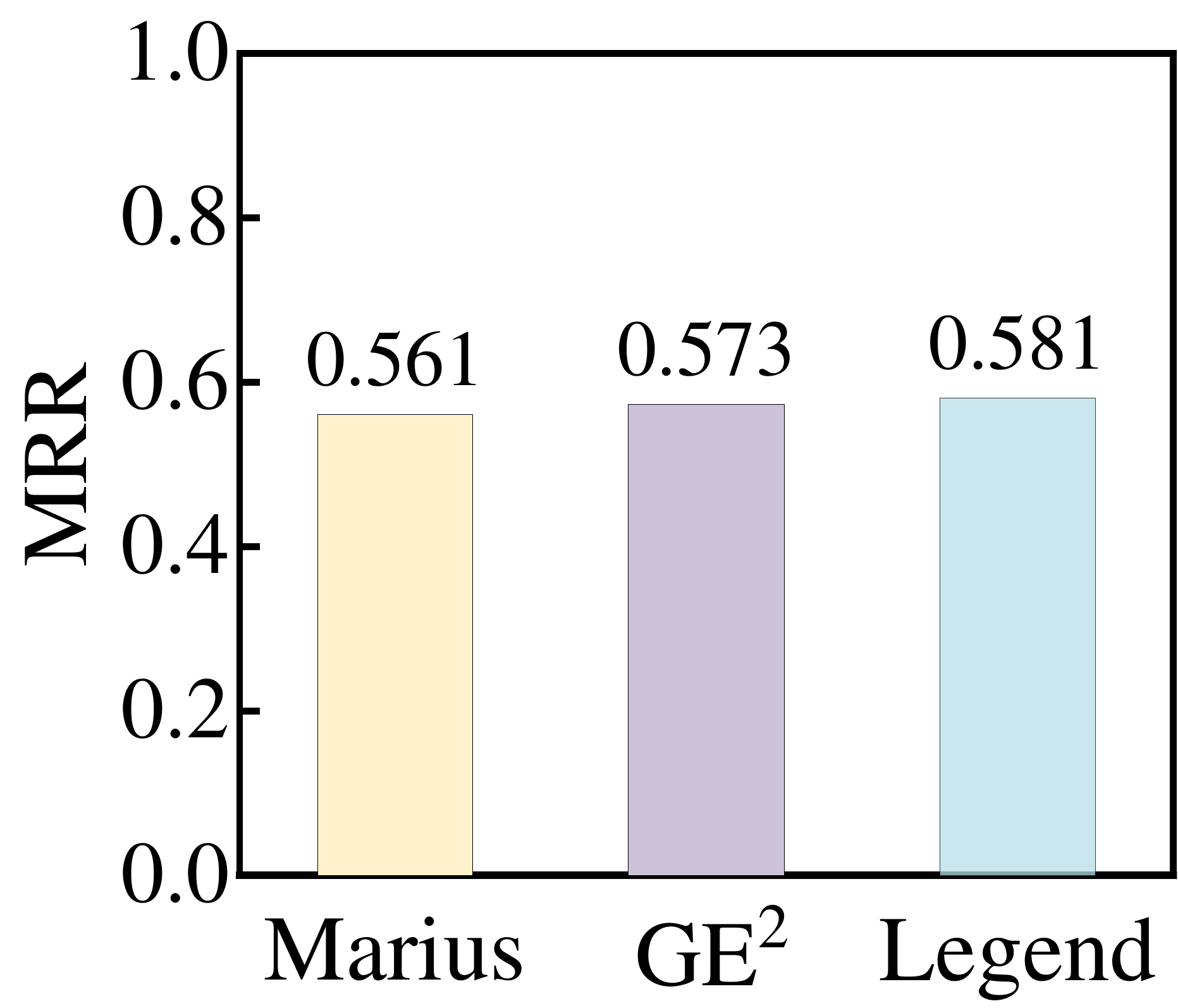}}
    \hspace{2mm}
    \subfigure[\textit{LJ}]{
    \includegraphics[width=0.22\textwidth]{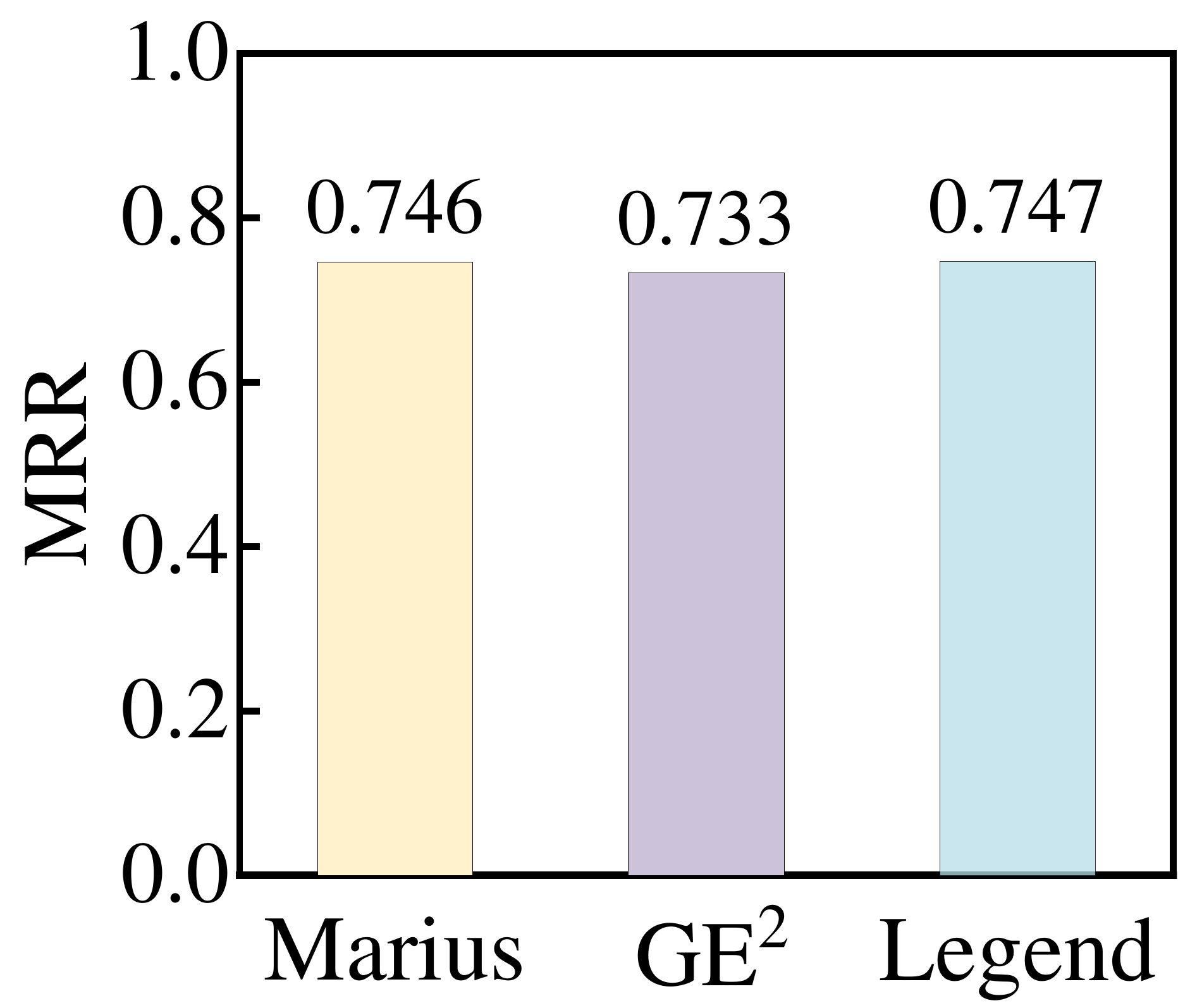}}
    \hspace{2mm}
    \subfigure[\textit{TW}]{
    \includegraphics[width=0.22\textwidth]{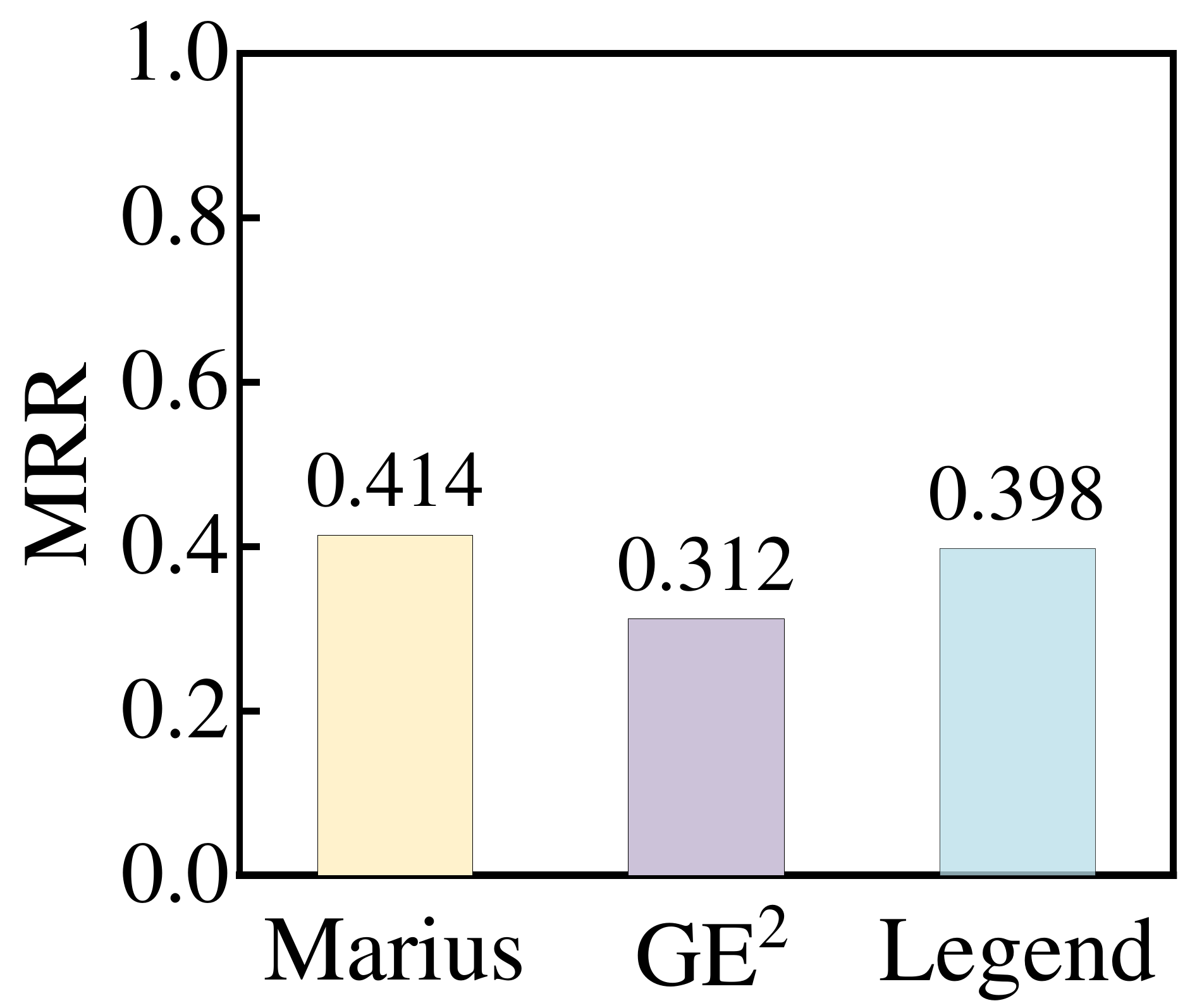}}
    \hspace{2mm}
    \subfigure[\textit{FM}]{
    \includegraphics[width=0.22\textwidth]{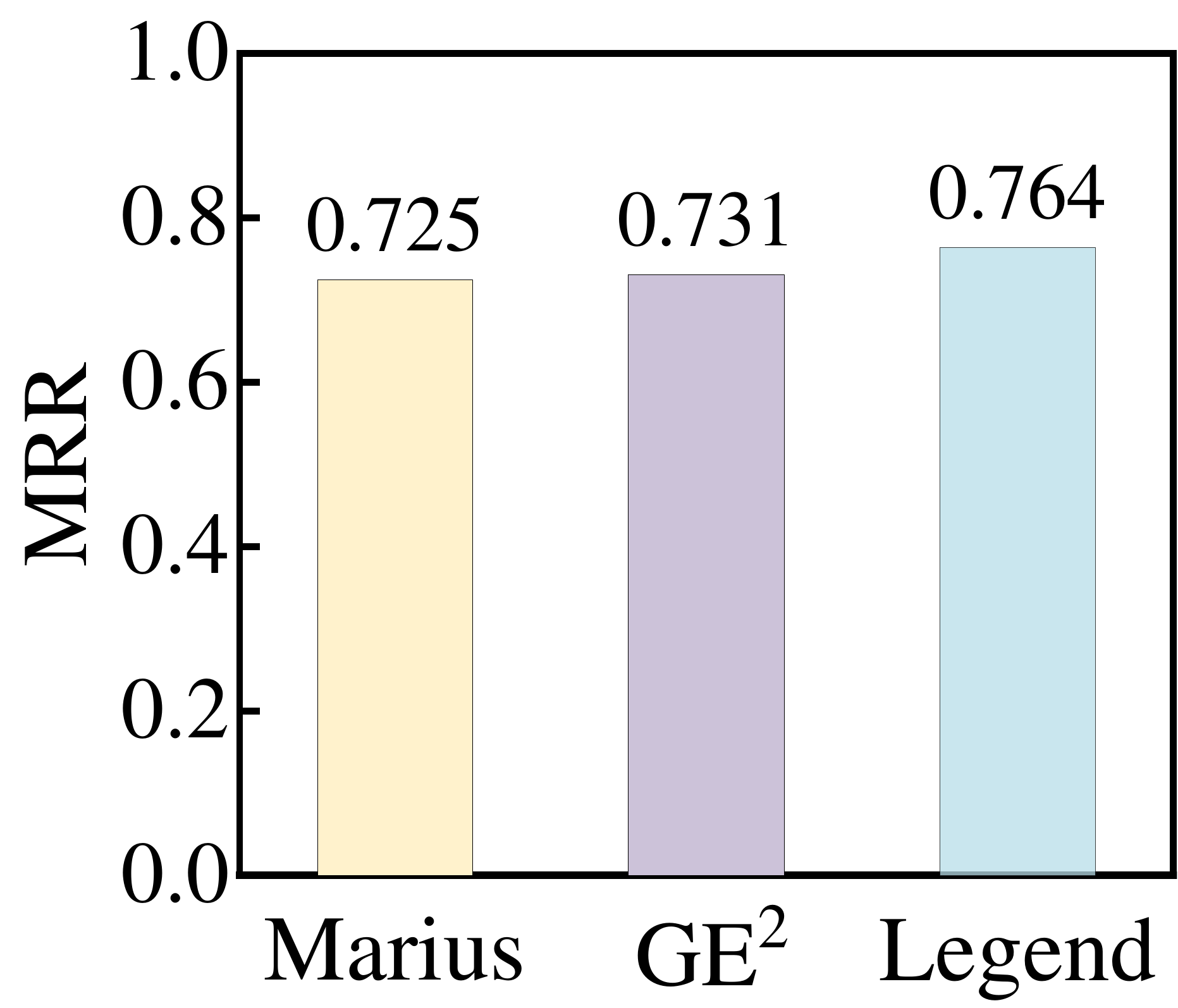}}

    \caption{Comparison of MRR using a single GPU. }
    \label{fig:over_performance_mrr}
\end{figure*}

\begin{figure*}
\centering 
    \subfigure[\textit{{FB}}]{
    \includegraphics[width=0.22\textwidth]{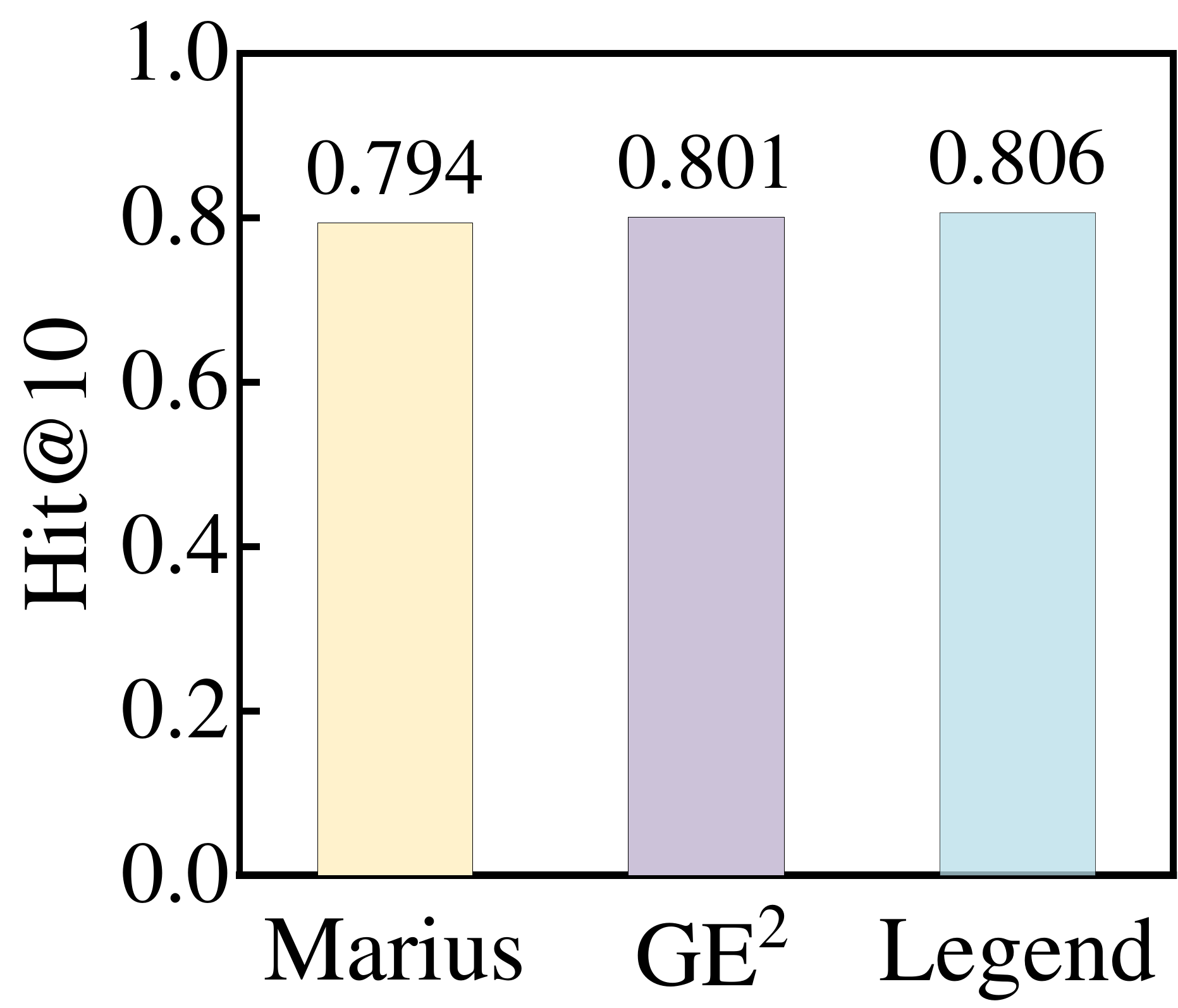}}
    \hspace{2mm}
    \subfigure[\textit{LJ}]{
    \includegraphics[width=0.22\textwidth]{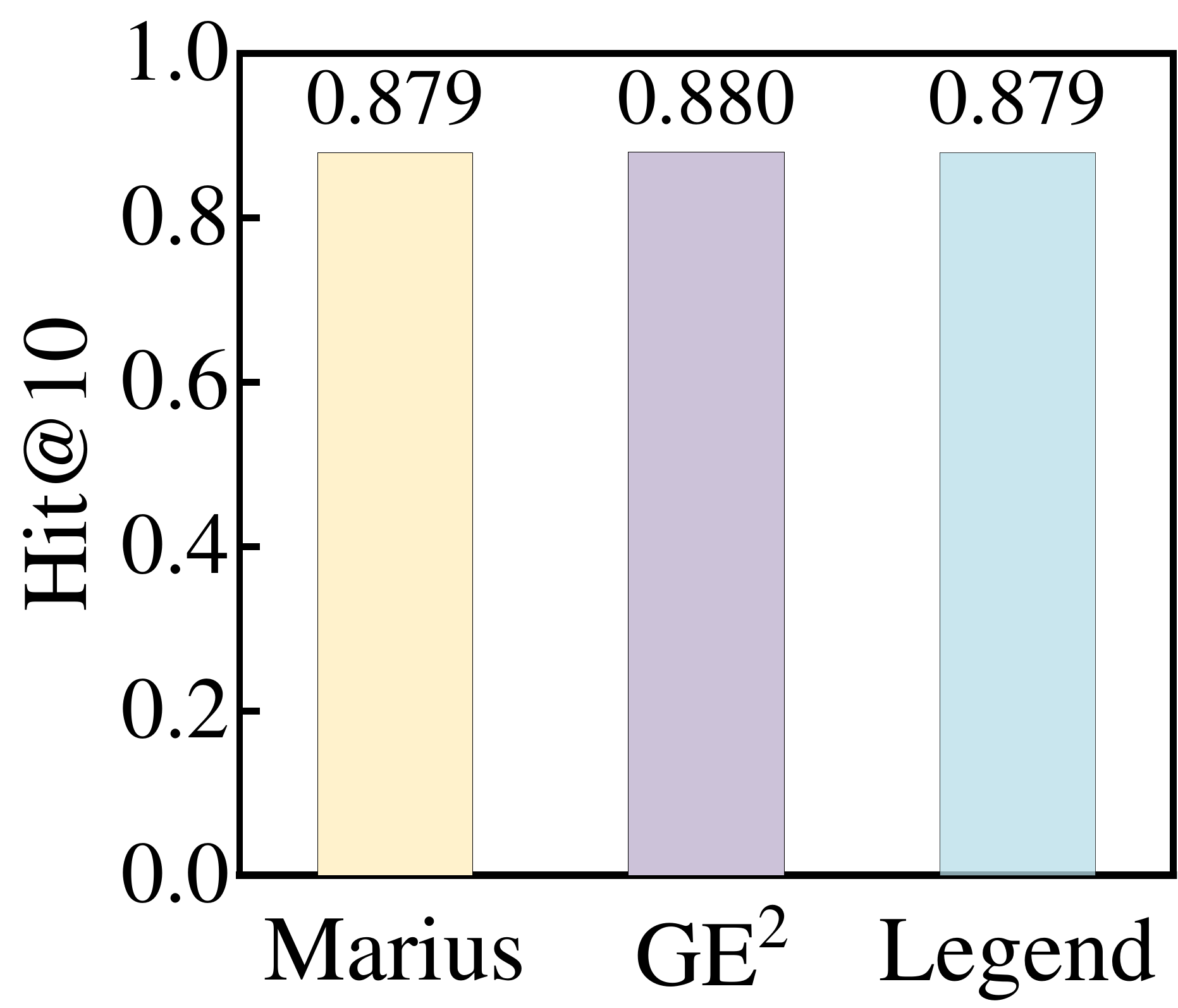}}
    \hspace{2mm}
    \subfigure[\textit{TW}]{
    \includegraphics[width=0.22\textwidth]{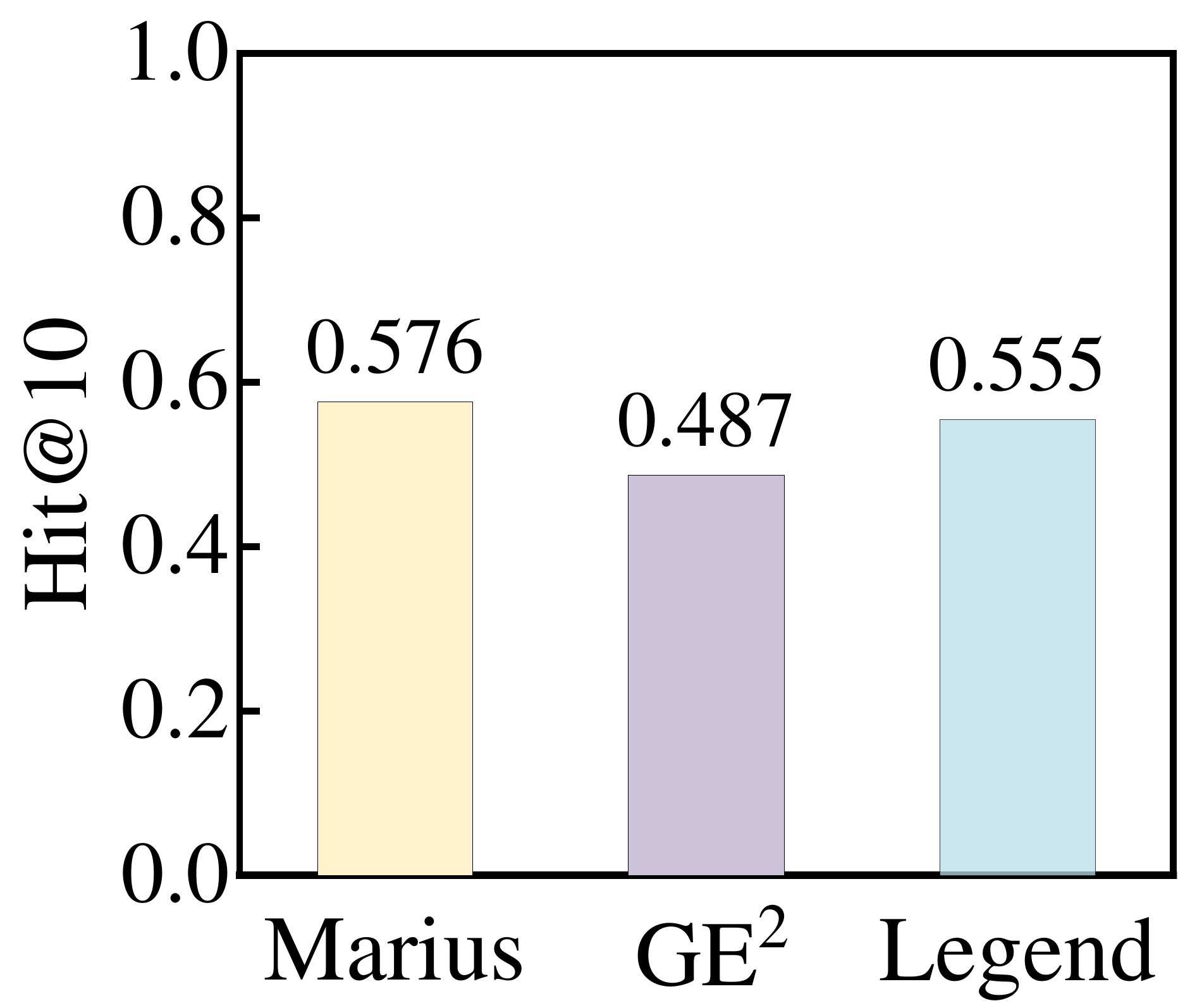}}
    \hspace{2mm}
    \subfigure[\textit{FM}]{
    \includegraphics[width=0.22\textwidth]{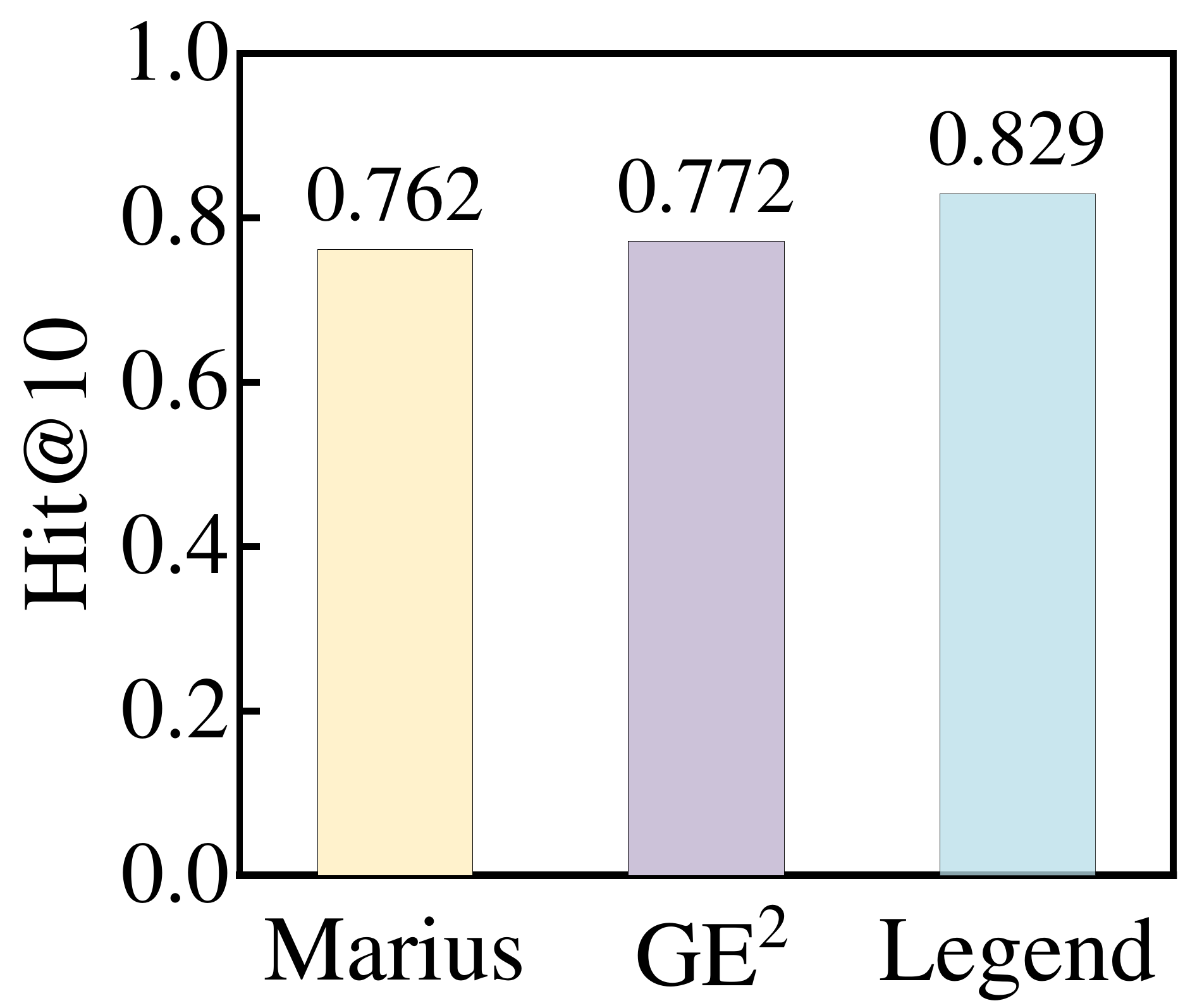}}

    \caption{Comparison of Hit$@10$ using a single GPU. }
    \label{fig:over_performance_hit}
    \vspace{-1.8mm}
\end{figure*}

\vspace{1mm}
\noindent \textbf{Embedding models}. Following {\sf Marius} and {\sf GE$^2$}, on datasets \textit{LJ} and \textit{TW} we employ the popular model \textsf{Dot}~\cite{jure2018} as they lack relation types. On \textit{FB} and \textit{FM}, we utilize \textsf{ComplEx}~\cite{trouillon2016complex}. 

\vspace{1mm}
\noindent \textbf{Baseline systems}. We compare \textsf{Legend} with two state-of-the-art graph embedding systems, i.e., \textsf{Marius}~\cite{mohoney2021marius} and \textsf{GE$^2$}~\cite{zheng2024ge2}, which are disk-based and RAM-based systems respectively. Among the two methods, \textsf{Marius} supports only a single GPU, while \textsf{GE$^2$} can leverage multiple GPUs for embedding training. We exclude \textsf{DGL-KE}~\cite{zheng2020dgl} and \textsf{PBG}~\cite{lerer2019pytorch} from the comparison, as \textsf{Marius} and \textsf{GE$^2$} have been demonstrated to outperform them. To ensure a fair comparison, we maintain identical hyperparameters across the three graph embedding learning systems, including a learning rate of 0.1, a batch size of $10^5$, $10^3$ negative samples per positive edge, 10 epochs for \textsf{TW} and \textsf{FM}, 30 epochs for \textsf{FB} and \textsf{LJ}, etc. 
{As the batch size affects the efficiency and quality, fixing batch size uniformly at $10^5$ and negative samples at $10^3$ enables a) direct hardware efficiency comparison and b) quality differences attributable solely to system design. }
\textsf{Legend} and \textsf{Marius} use 8 node partitions with a buffer capacity of 3 (12GB of the GPU global memory) for \textit{TW} and 12 node partitions with a buffer capacity of 3 (17GB of the GPU global memory) for \textit{FM}. \textsf{GE$^2$} uses 16 node partitions and a buffer capacity of 4 on \textit{TW} and \textit{FM} because it only supports the number of partitions of $4^L$ and a fixed buffer capacity of 4. This comparison is fair as the restricted support for flexible partition numbers is exactly the limitation of {\sf GE$^2$}. 
Nonetheless, we also apply the order in {\sf GE$^2$} with 16 partitions to {\sf Legend}, as referenced in Figure \ref{fig:order_time}. 

\vspace{1mm}
\noindent \textbf{Metrics}. We employ Mean Reciprocal Rank (MRR) and Hits$@k$ as the quality metrics, which are widely used to evaluate the embeddings~\cite{qiu2021lightne,kochsiek2021parallel,zheng2024ge2,mohoney2021marius}. Higher MMR and Hits$@k$ values indicate better embedding quality. Consistent with \textsf{GE$^2$}, we utilize part of the test edges ($10^6$) to compute MRR and Hit$@k$, as using the entire test set would be time-prohibitive.

\vspace{1mm}
\noindent \textbf{Platforms}. All experiments are conducted on a server with the system of Ubuntu 20.04, featuring an Intel Xeon Silver 4216 CPU@2.10GHz with 64 cores, Nvidia A100 GPU (40G), and Samsung 980 NVMe SSD (1T). We implement {\sf Legend} in C++/CUDA under Nvidia CUDA 11.1 and {\sf LibTorch} 1.7.1. {\sf Legend} can be easily integrated into {\sf Pytorch} by {\sf pybind}, but the host running {\sf Pytorch} needs to be equipped with an NVMe SSD and a GPU supporting GPUDirect RDMA. 

\subsection{Comparisons with Existing Systems}
\label{subsec:overall}

\begin{figure}
\centering 
    \subfigure[\textit{TW}]{
    \includegraphics[width=0.225\textwidth]{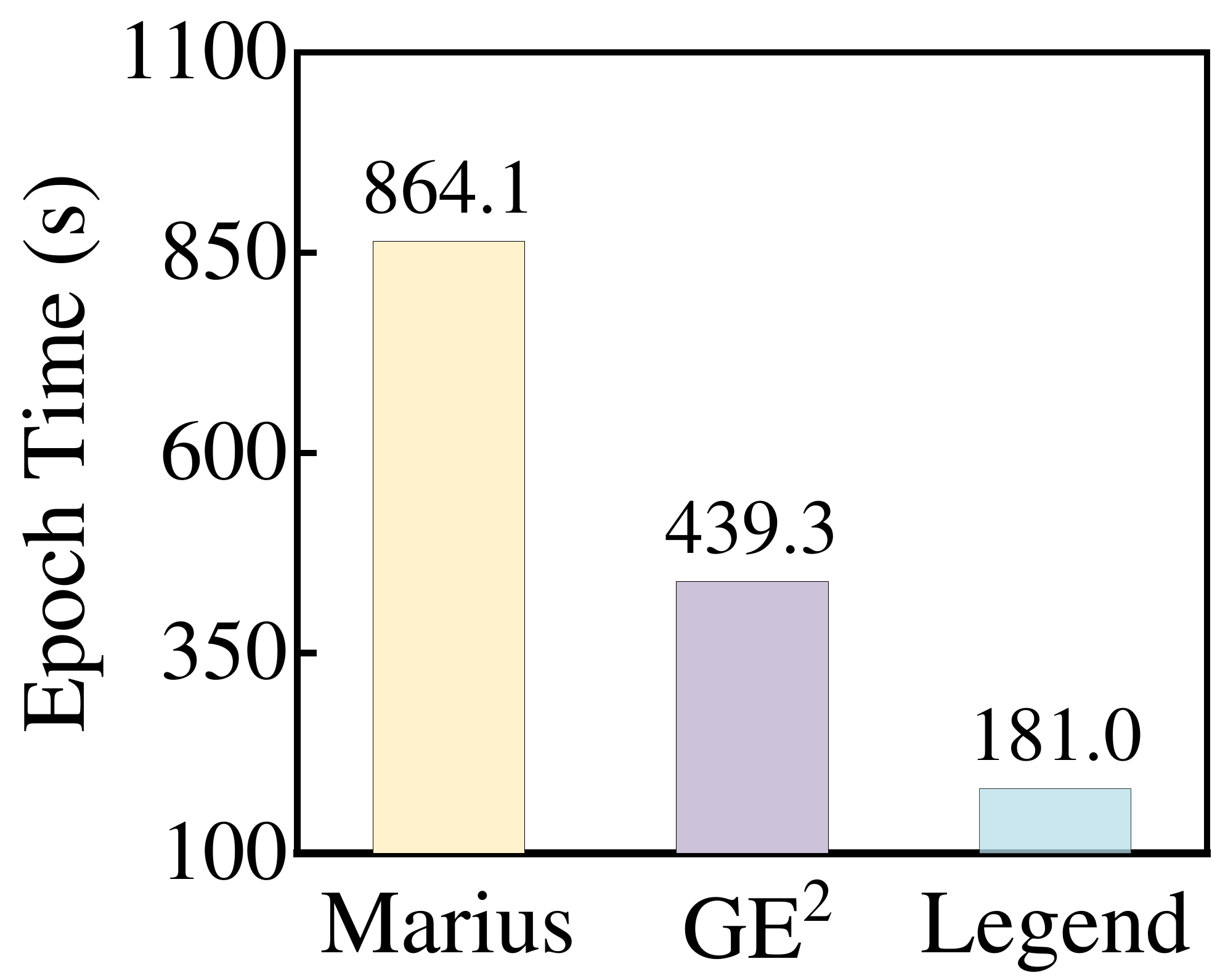}}
    \hspace{2mm}
    \subfigure[\textit{FM}]{
    \includegraphics[width=0.22\textwidth]{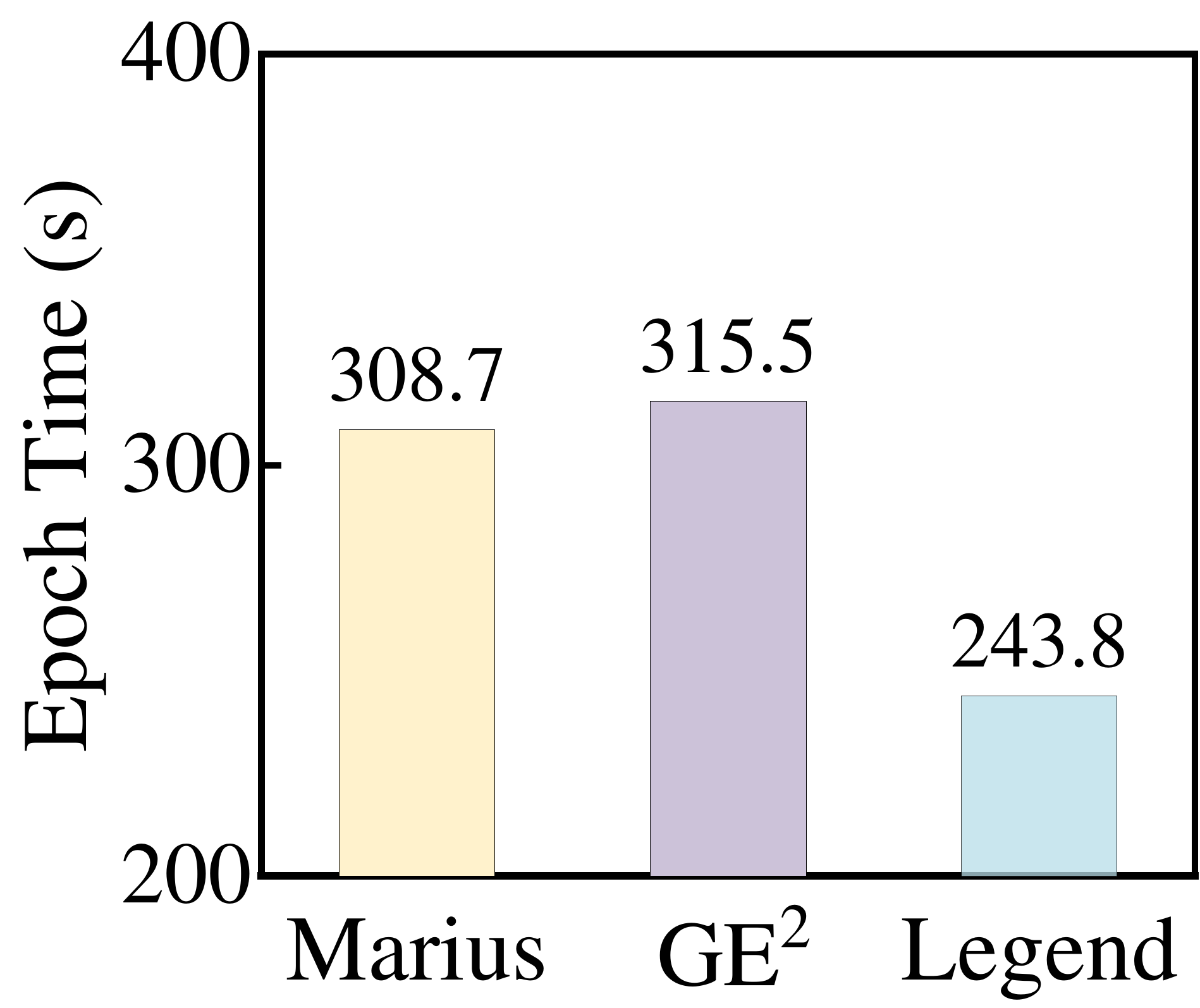}}
    \caption{{Comparison of the average epoch time with 3/4 node partitions residing in memory for {\sf Marius}.}}
    \label{fig:expandBuffer}
\end{figure}

\begin{figure}
\centering 
    \subfigure[\textit{TW}]{
    \includegraphics[width=0.225\textwidth]{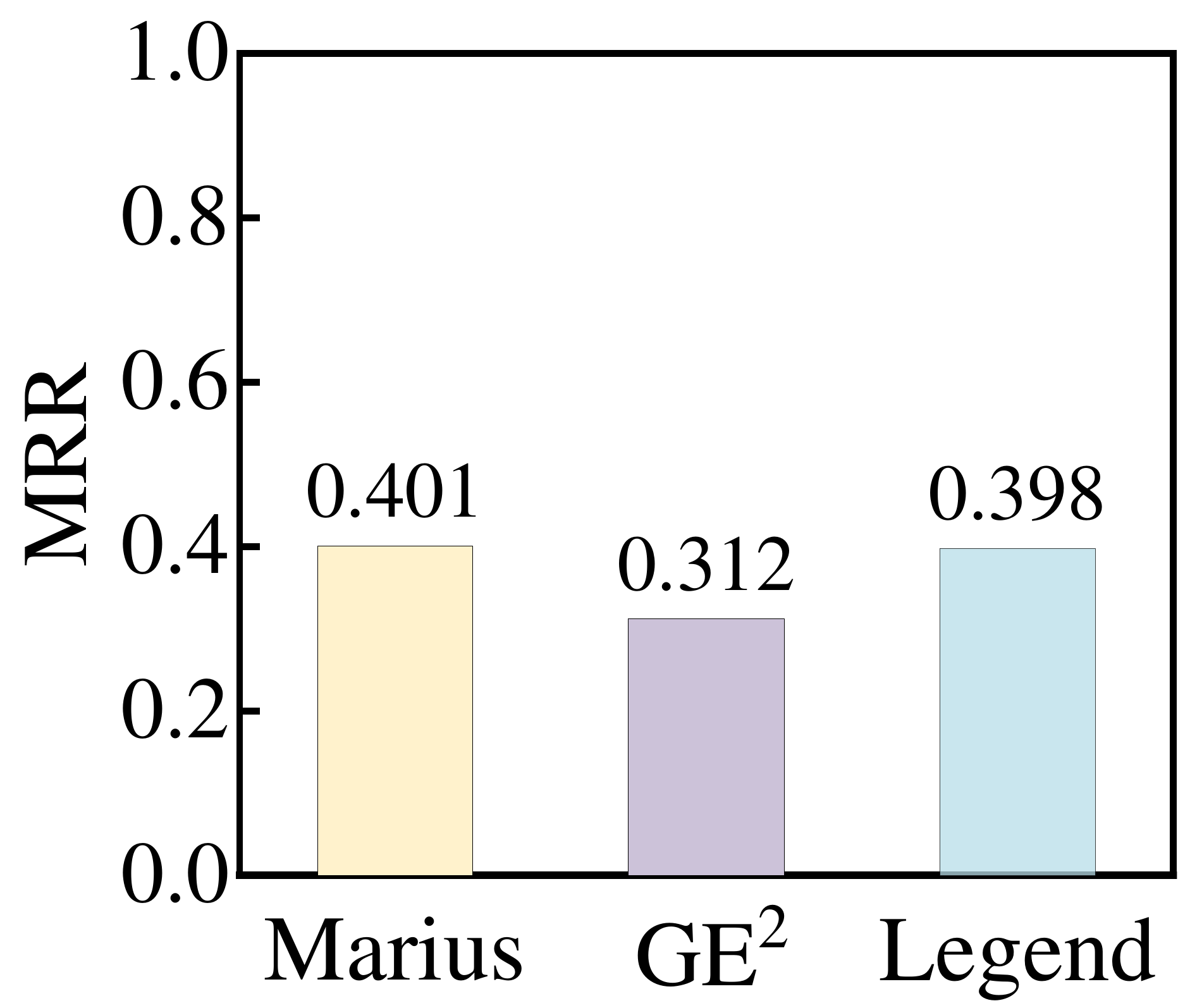}}
    \hspace{2mm}
    \subfigure[\textit{FM}]{
    \includegraphics[width=0.22\textwidth]{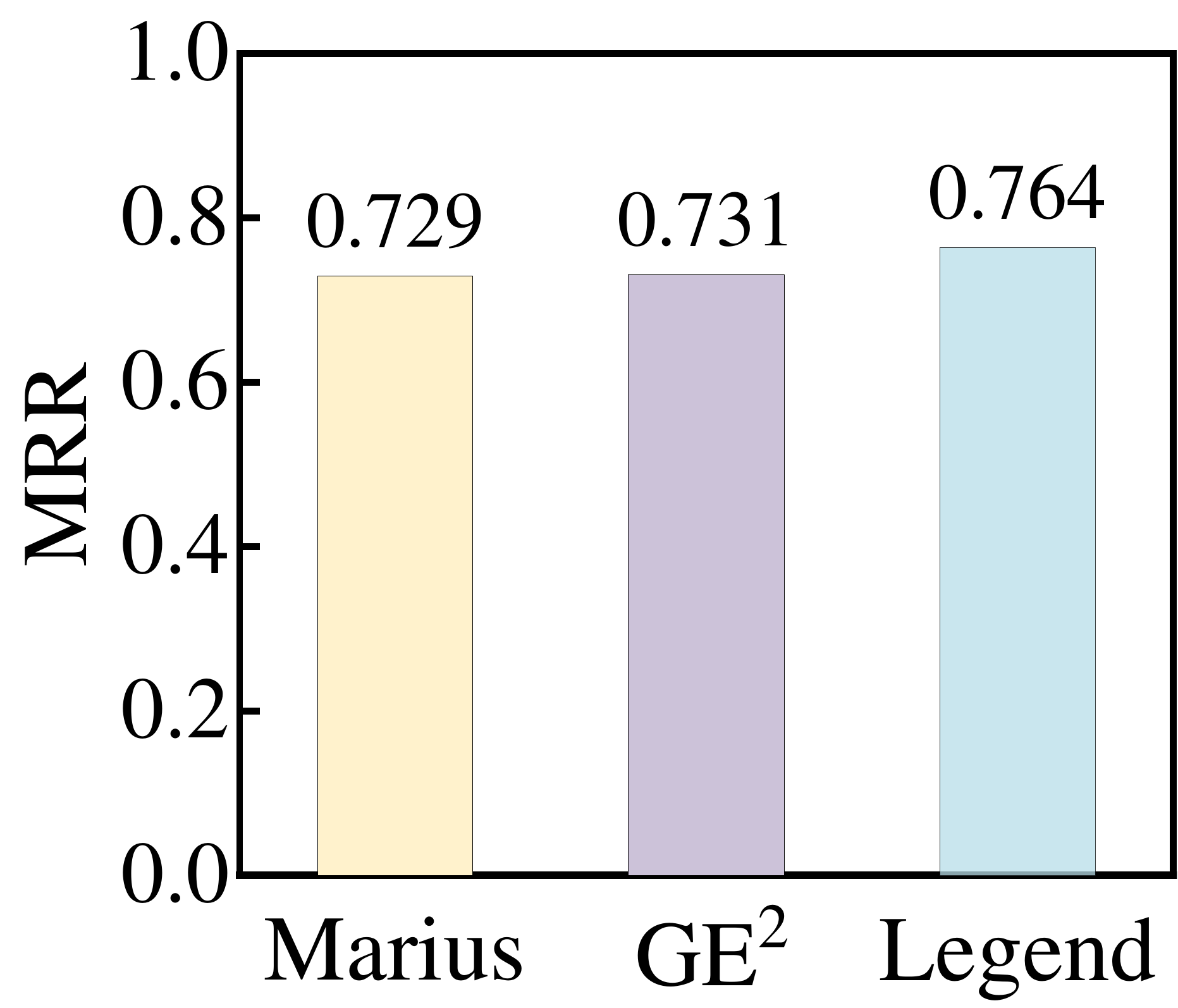}}
    \caption{{Comparison of MRR with 3/4 node partitions in memory for {\sf Marius}.}}
    \label{fig:expandBufferMRR}
\end{figure}

\begin{figure}
\centering 
    \subfigure[\textit{TW}]{
    \includegraphics[width=0.225\textwidth]{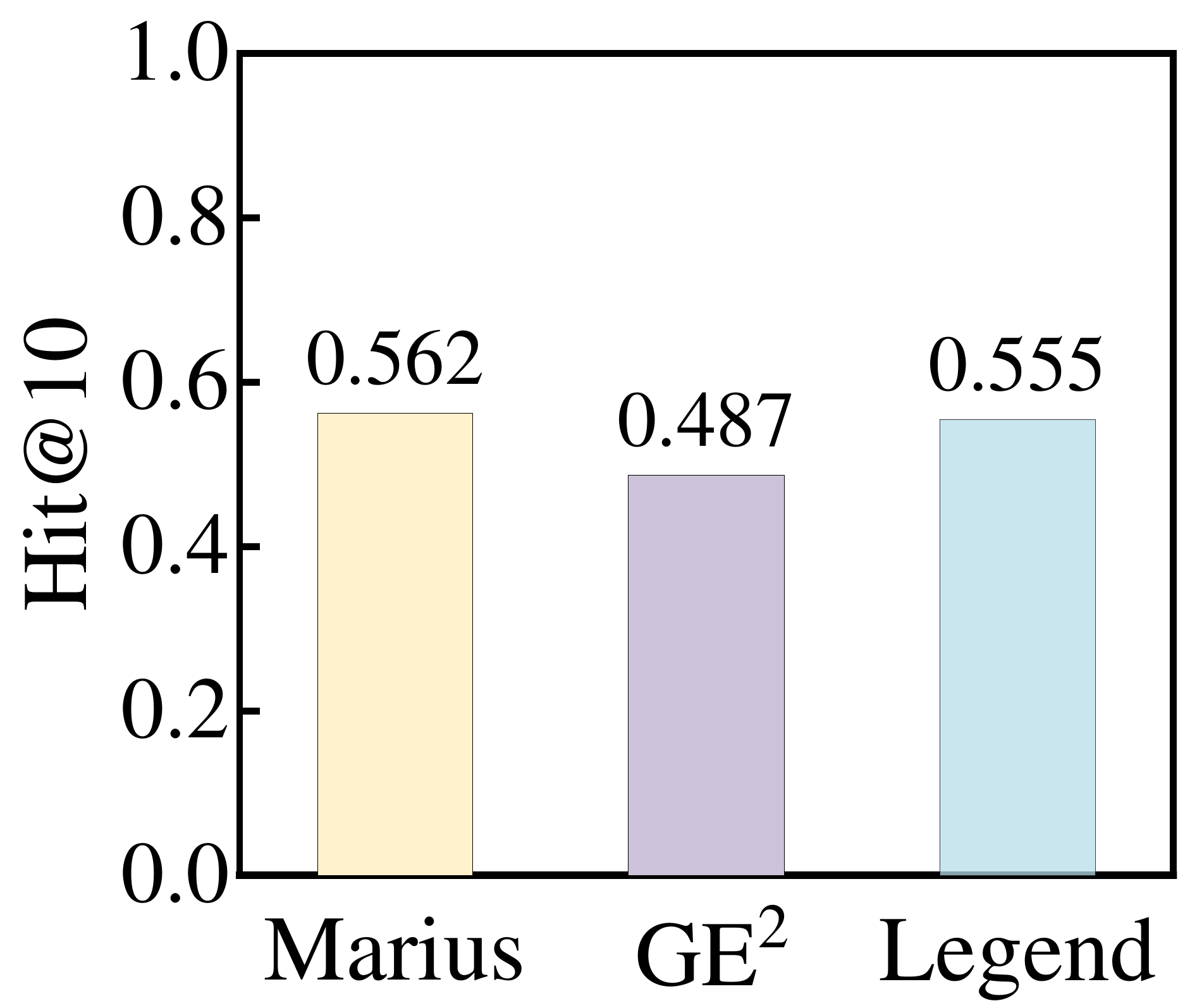}}
    \hspace{2mm}
    \subfigure[\textit{FM}]{
    \includegraphics[width=0.22\textwidth]{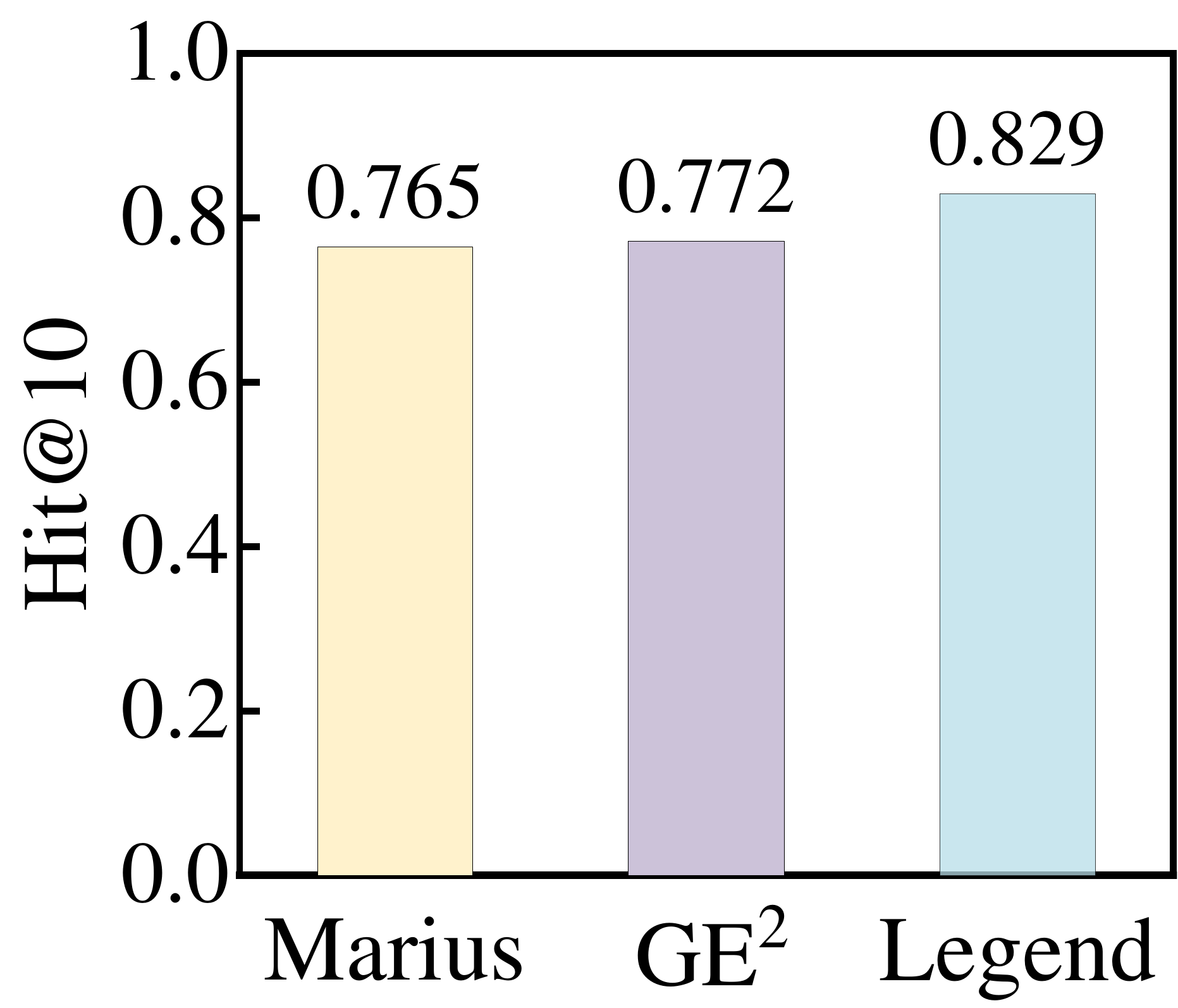}}
    \caption{{Comparison of Hit@10 with 3/4 node partitions in memory for {\sf Marius}.}}
    \label{fig:expandBufferREC}
\end{figure}

Firstly, we evaluate the overall performance of the compared systems. The training overhead for the three systems using a single GPU is reported in Figure \ref{fig:over_performance_time}. The time reported in Figure \ref{fig:over_performance_time} is the average epoch duration. To demonstrate the embedding quality trained by {\sf Legend}, we report the MRR and Hit$@10$ in Figure \ref{fig:over_performance_mrr} and Figure \ref{fig:over_performance_hit}, respectively. 
{On average, {\sf Legend} achieves a speedup of 2.6$\times$ over {\sf Marius} and 2.0$\times$ over {\sf GE$^2$} while maintaining similar embedding quality. }
In optimal scenarios, {\sf Legend} achieves a remarkable speedup of 4.8$\times$ over {\sf Marius} on \textit{TW} and 2.4$\times$ speedup over {\sf GE$^2$} on \textit{TW} and \textit{FB}. 
{Noted that a larger host-side buffer could further benefit {\sf Marius}, we also expand the buffer size for {\sf Marius} to 6 and 9 on datasets \textit{TW} and \textit{FM}, respectively. As shown in Figure \ref{fig:expandBuffer}, \ref{fig:expandBufferMRR} and \ref{fig:expandBufferREC}, with 3/4 of the partitions retained in memory, {\sf Marius} achieves epoch times of 864.1s (876.0s with 3/8 partitions in memory) on \textit{TW} and 308.7s (382.2s with 3/12 partitions in memory) on \textit{FM}. The improvement is marginal for \textit{TW}, where computation dominates, but significant for \textit{FM}, where I/O is the primary bottleneck. Besides, the comparable MRR and Hit@10 with expanded buffer indicate that the buffer size has no significant impact on model quality.}
It's worth noting that {\sf GE$^2$} stores embeddings and optimizer states in RAM, while {\sf Legend} stores them in the NVMe SSD. 
{\sf Legend} exhibits excellent scalability and efficiency on the four datasets with various volumes. 
This is attributed to the optimization of each hardware component and the workflow that orchestrates each hardware in the heterogeneous system by making full use of its unique characteristics. 
Although the systems load all data into the GPU memory without I/O overhead during embedding learning on datasets \textit{FB} and \textit{LJ}, {\sf Legend} still demonstrates superior training speed. This indicates that the workflow and optimizations on the GPU contribute to accelerating the training process, except for the I/O optimization. 
As shown in Figure \ref{fig:over_performance_mrr} and \ref{fig:over_performance_hit}, the embedding quality trained by {\sf Legend} is comparable with those of {\sf Marius} and {\sf GE$^2$}. Moreover, on \textit{FM}, {\sf Legend} exhibits relatively higher embedding quality. 
As introduced in Section \ref{sec:workflow}, {\sf Legend} loads the entire node partitions into the global memory of the GPU and constructs batches on the GPU, ensuring that the updated embeddings from the last batch can be used immediately in the current batch, avoiding the problem of staleness present in {\sf Marius}. This advantage is more apparent when more nodes are in a node partition. As a result, {\sf Legend} achieves better performance on \textit{FM} in the metrics of MRR and Hit$@10$. 

On \textit{FM}, the speedup of {\sf Legend} is relatively insignificant, which is due to the graph properties. The number of edges in \textit{FM} is relatively small compared to the number of vertices, where $\frac{|E|}{|V|^2}\approx4\times 10^{-8} < 10^{-7}$. According to Theorem \ref{theo:3}, the I/O overhead between the GPU and NVMe SSD cannot be entirely covered by computation. Furthermore, the I/O times can reach 36 even though the node partition ordering algorithm is applied, exacerbating the I/O overhead. In contrast, the training speed of {\sf Legend} is more significant on \textit{TW}. Using the records in Table \ref{tab:dataset}, $\frac{|E|}{|V|^2}\approx 8\times 10^{-7}>10^{-7}$ on \textit{TW}, which indicates the I/O overhead can be covered by computation. Consequently, this alleviates the bandwidth constraints between the GPU and NVMe SSD, leading to improved performance.

\begin{table}[tbp]
\renewcommand{\arraystretch}{1.2}
    \centering
    \small
    \caption{{Performance comparison with different buffer size. The second column reports the ratio of buffer size to the total partitions (buffer size / total partitions).} }
    \begin{tabular}{p{1.0cm}<{\centering}p{1.54cm}<{\centering}p{1.7cm}<{\centering}p{0.9cm}<{\centering}p{1.0cm}<{\centering} }
        \toprule
         Methods & Ratio & Epoch time & MRR & Hit@10 \\
         \toprule
         \multirow{3}{*}{\sf Marius} & 2 / 8 & 918.0s & 0.418 & 0.579 \\
          & 3 / 8 & 876.2s & 0.414 & 0.576  \\
          & 4 / 8 & 864.2s & 0.401 & 0.566  \\
          \toprule
         \multirow{3}{*}{\sf Legend} & 3 / 12 & 230.6s & 0.412 & 0.574  \\
          & 3 / 8 & 181.0s & 0.398 & 0.555  \\
          & 3 / 6 & 171.4s & 0.406 & 0.570  \\
        \bottomrule
    \end{tabular}
    \label{tab:buffersize}
\end{table}

{To evaluate the effects of various buffer sizes on embedding quality, we conduct additional experiments. Considering the larger capacity of RAM compared to the GPU, we conduct the experiments with {\sf Marius} and {\sf Legend}, which use the same embedding models and negative sampling strategies. The experiments evaluate the model quality and epoch time using different buffer sizes, while holding other parameters constant. As shown in Table \ref{tab:buffersize}, the model quality (MRR \& Hit@10) and buffer size have no definite correlation, indicating that the small buffer size does not substantially hinder the negative sampling diversity or downstream model quality. This is because for each vertex $v\in V$ in the graph, all edges involving $v$ are traversed within one epoch. Regardless of buffer size, every node partition containing vertices $u\in V\backslash\{v\}$ will appear in the buffer at least once together with the partition containing $v$ during that epoch. Thus, each vertex in $V$ has the opportunity to be paired with $v$ as a negative sample, ensuring the diversity of the negative sampling. }

To further validate the GPU utilization improvements from our proposed techniques, we assess GPU utilization on the dataset \textit{TW}. Figure \ref{fig:utilization} depicts the variation in GPU utilization across the three systems over time. {The average GPU utilization of {\sf Legend} is 96.79\%, compared to 60.14\% for {\sf Marius} and 59.85\% for {\sf GE$^2$}, even with prefetching enabled.} As shown in Figure \ref{fig:utilization}, GPU utilization periodically drops to zero for {\sf Marius} and {\sf GE$^2$}, indicating that the GPU is idle during data loading from the disk or RAM. In contrast, GPU utilization of {\sf Legend} remains consistently above 55\%, exceeding 90\% for most of the duration. This enhanced utilization can be attributed to three key factors. First, we offload batch construction and negative sampling to the GPU, which improves the batch construction speed. Second, we prefetch the node embeddings and optimize the bandwidth between the GPU and NVMe SSD, which reduces the data transfer overhead and the GPU idle time. Third, we further optimize the training on the GPU by a customized parallel strategy and data reuse to maximize the resource utilization of the GPU.

\begin{figure}
    \centering
    \includegraphics[width=0.98\linewidth]{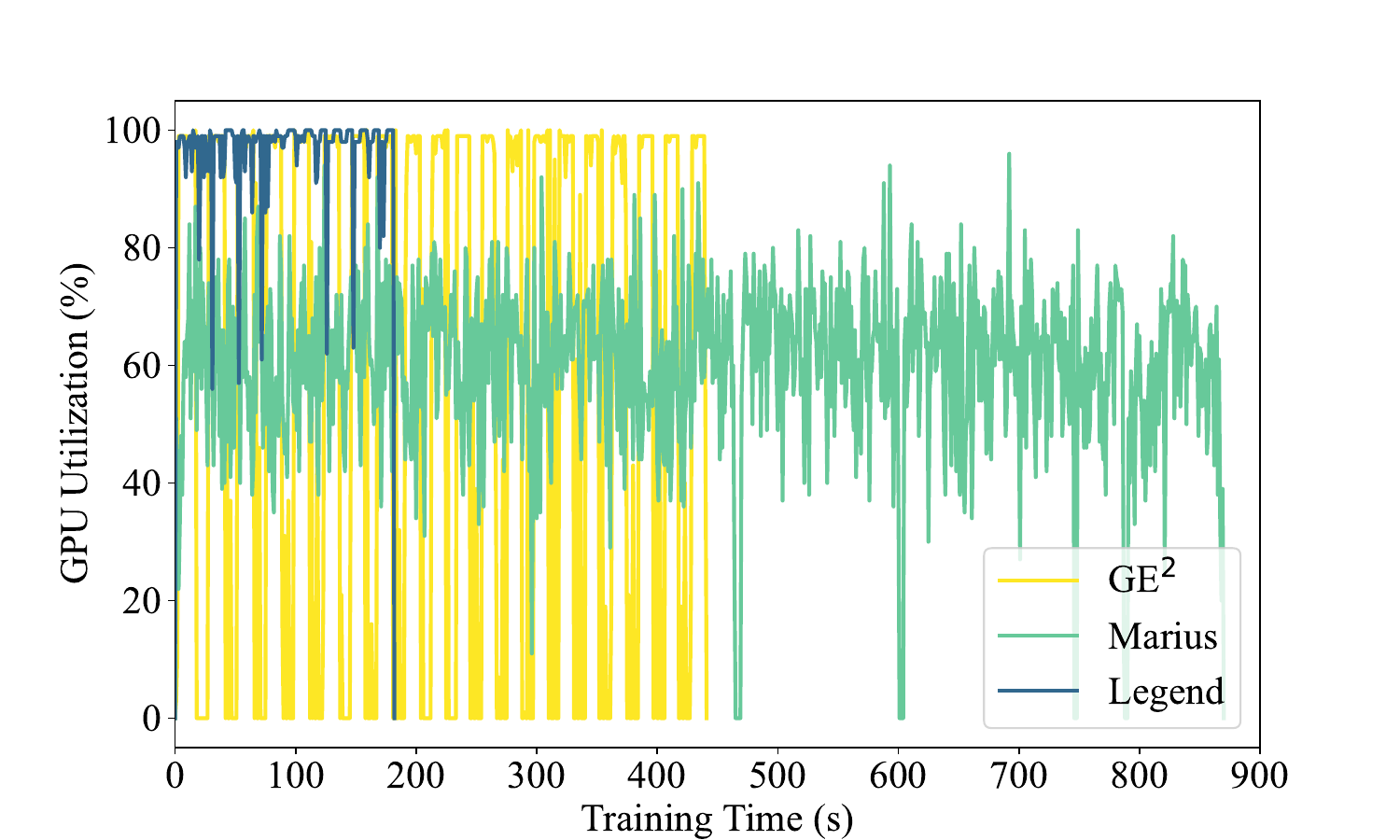}
    \caption{{GPU utilization of {\sf Legend}, {\sf GE$^2$} and {\sf Marius} on \textit{TW}.} }
    \label{fig:utilization}
\end{figure}

Please note that the partition order in {\sf Legend} doesn't support prefetching across multiple GPUs, and accessing data from the NVMe SSD to multiple GPUs adversely affects the throughput. We leave the support for multi-GPU graph embedding for future work. Nevertheless, we compare {\sf Legend} on a single GPU with {\sf GE$^2$} on multi-GPU using dataset \textit{TW}. Table \ref{tab:multigpus} presents the experimental results. 
{\sf Legend} exhibits superior performance compared to {\sf GE$^2$}. 
Particularly, when {\sf GE$^2$} employs 4 GPUs, {\sf Legend} still shows comparable performance. 
Note that as the number of GPUs increases, the time overhead of {\sf GE$^2$} does not decrease proportionally. This phenomenon arises from the limited I/O bandwidth between host and device memory, which constrains data transfer rates to multiple GPUs. This issue can be mitigated by employing the NVMe SSD as the primary data storage device. Since the NVMe SSD is much cheaper than RAM, it is feasible to allocate one NVMe SSD per GPU, thereby eliminating competition for limited bandwidth, which represents a promising direction for future research, and further demonstrates the scalability of {\sf Legend}.

\begin{table}[tbp]
\renewcommand{\arraystretch}{1.2}
    \centering
    \small
    \caption{Performance on various number of GPUs on \textit{TW}.}
    \begin{tabular}{p{1.0cm}<{\centering}p{1.0cm}<{\centering}p{1.1cm}<{\centering}p{1.0cm}<{\centering}p{2.2cm}<{\centering} }
    \toprule
    Systems & GPU(s) & MRR & Hit$@10$ & Time (s) \\
       \toprule
       \multirow{3}{*}{\sf GE$^2$} & 1  & 0.312 & 0.487 & 439.3 (2.43$\times$) \\
       \cline{2-5}
        & 2 & 0.299 & 0.473 & 315.2 (1.74$\times$)\\
        \cline{2-5}
        & 4 & 0.284 & 0.456 & 192.5 (1.06$\times$)\\
        \toprule
        {\sf \textbf{Legend}} & \textbf{1} & \textbf{0.398} & \textbf{0.555} & \textbf{181.0} \\ 
        \bottomrule
    \end{tabular}
    \label{tab:multigpus}
\end{table}

\subsection{Evaluations of the Workflow in Legend}

\begin{figure}
\centering 

    \subfigure[\textit{{FB}}]{
    \includegraphics[width=0.21\textwidth]{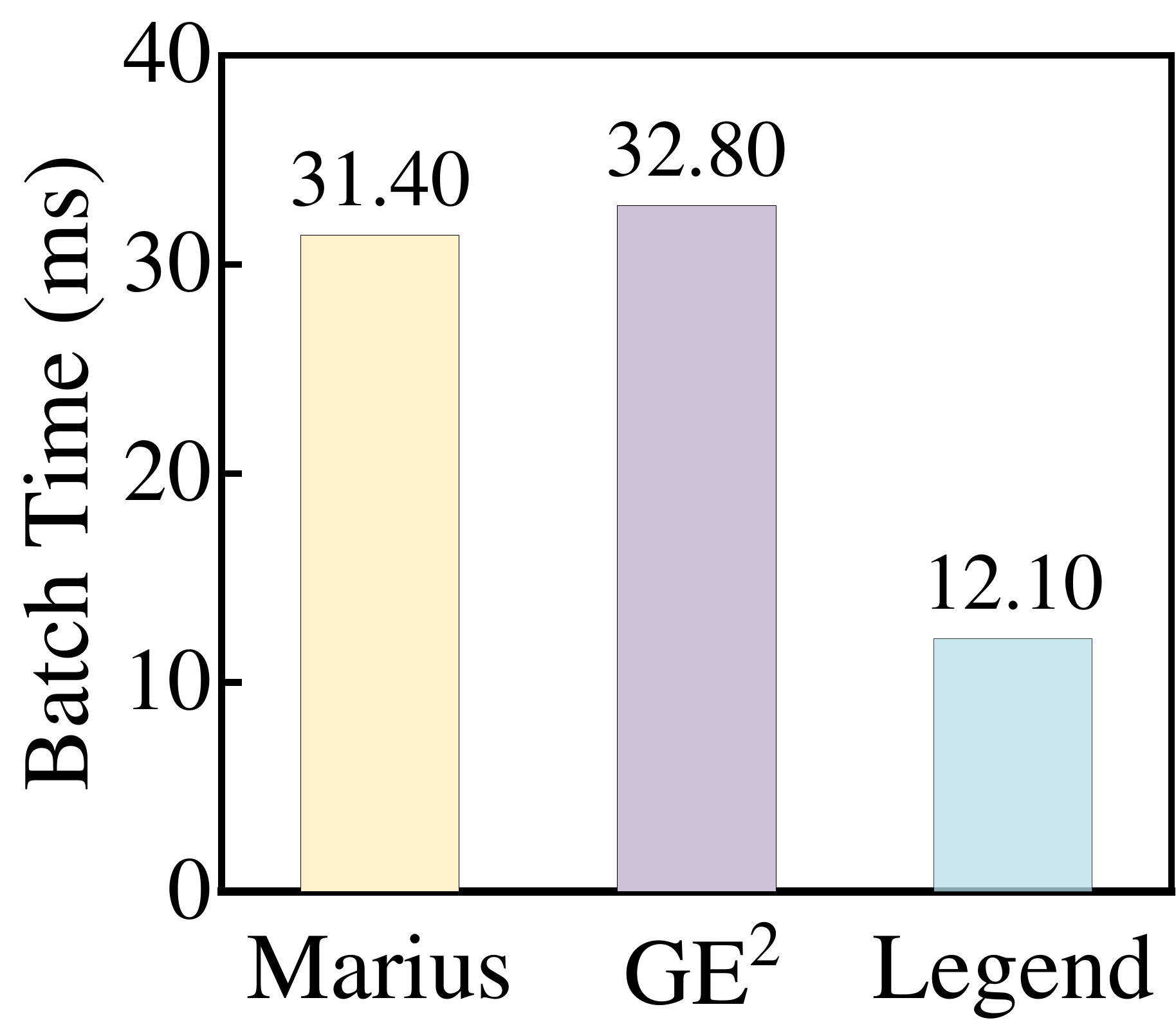}}
    \hspace{2mm}
    \subfigure[\textit{LJ}]{
    \includegraphics[width=0.21\textwidth]{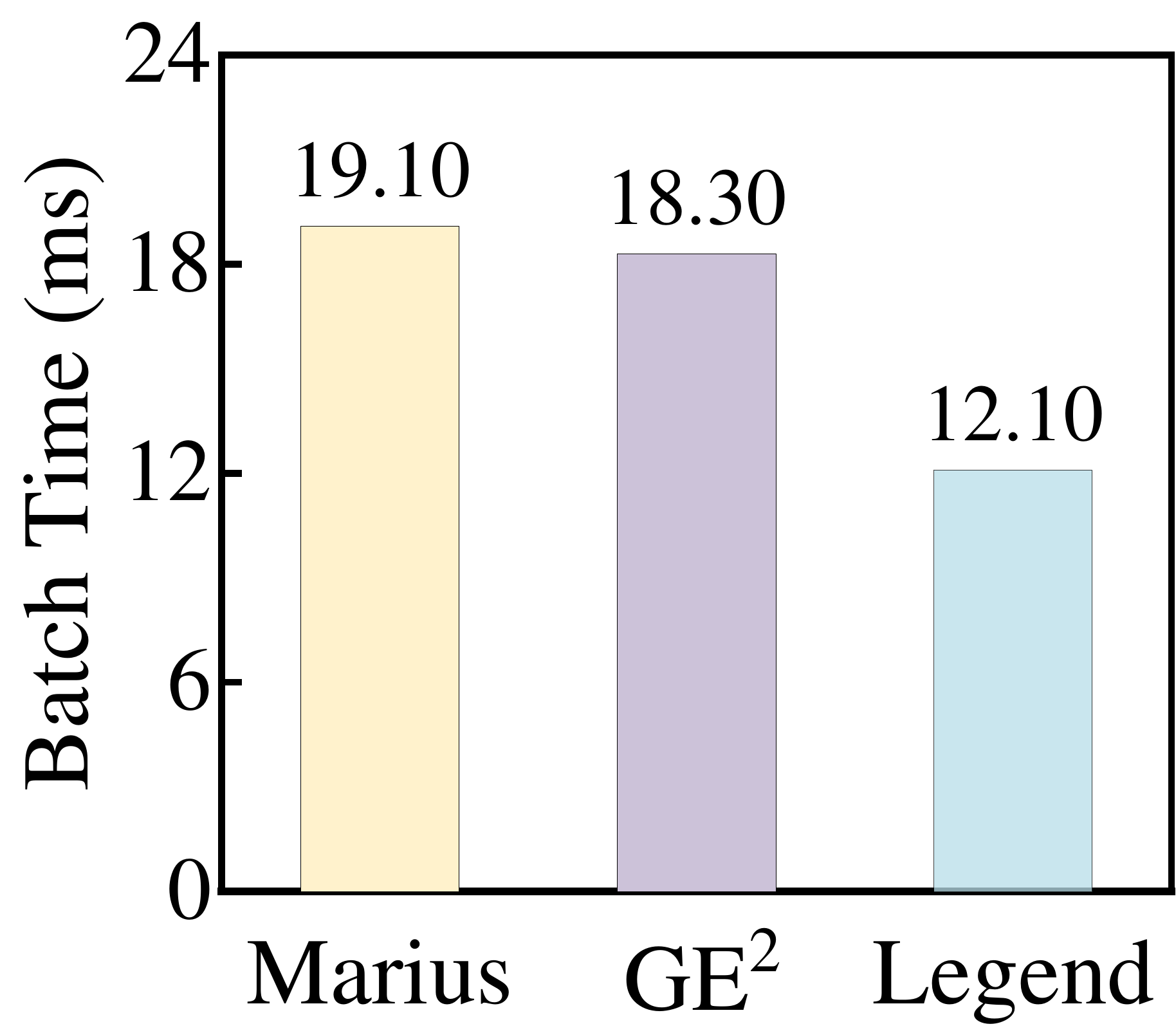}}
    \hspace{2mm}
    \subfigure[\textit{TW}]{
    \includegraphics[width=0.22\textwidth]{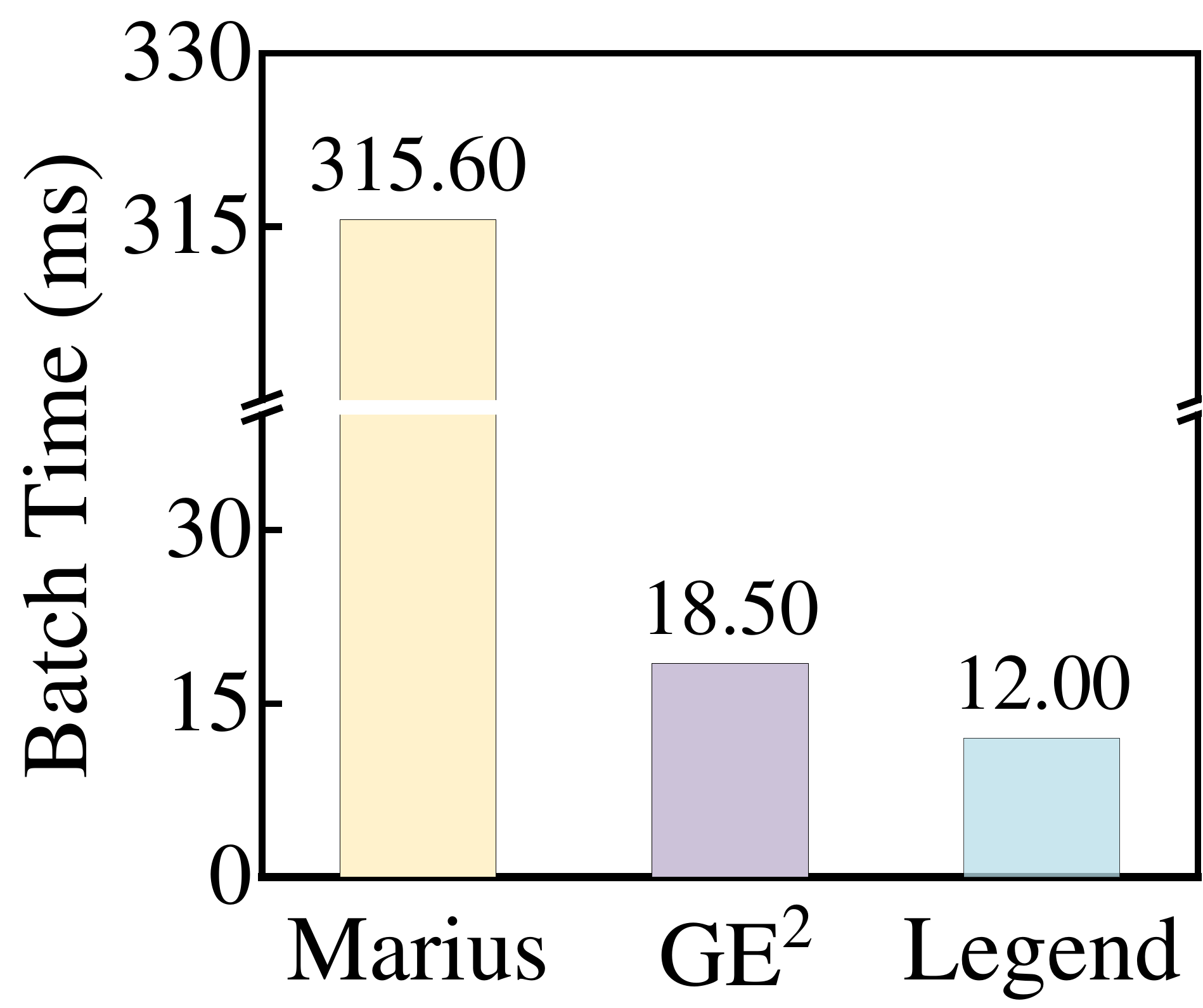}}
    \hspace{2mm}
    \subfigure[\textit{FM}]{
    \includegraphics[width=0.22\textwidth]{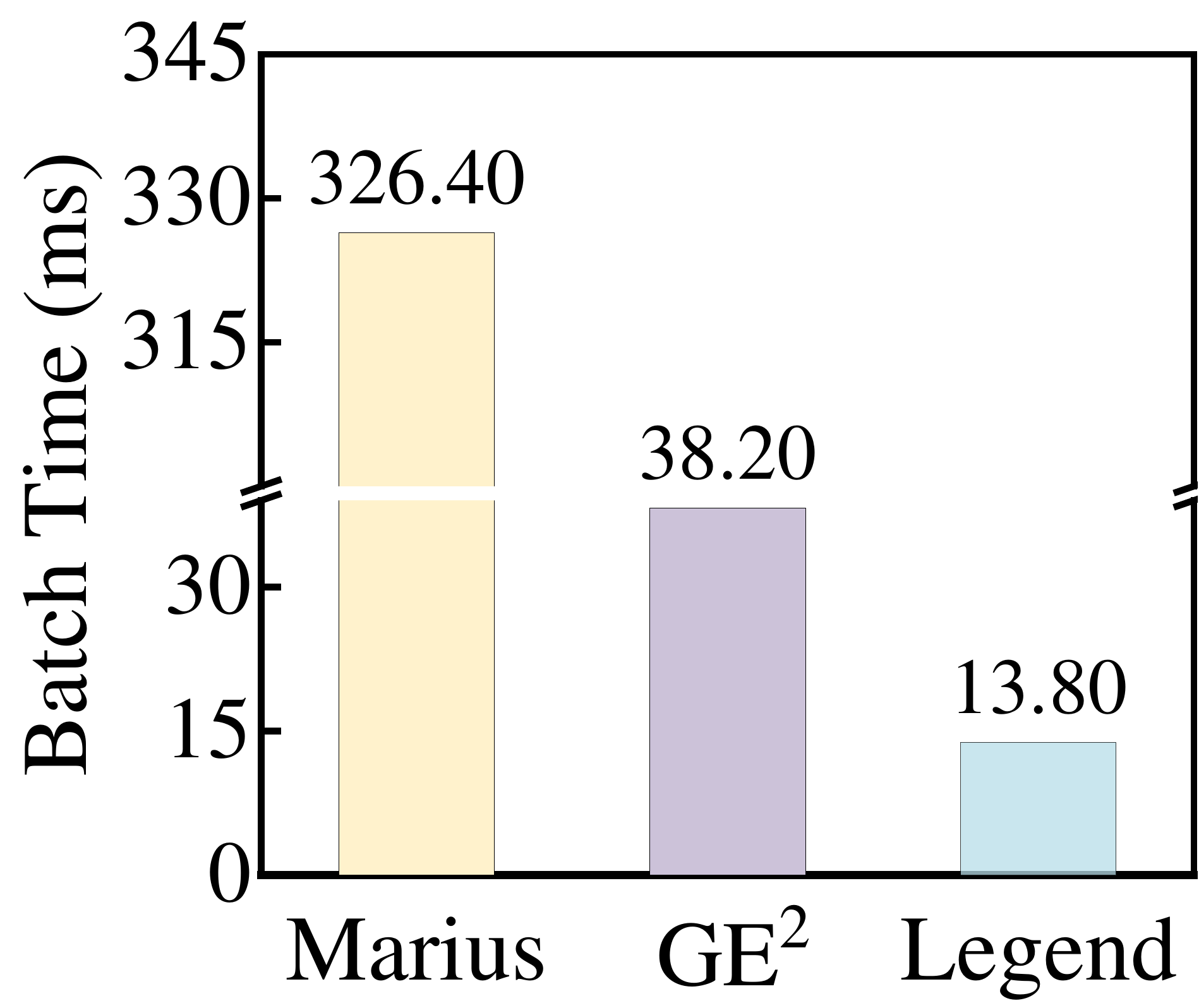}}

    \caption{Average batch time of the compared systems.}
    \label{fig:batch_time}
\end{figure}

To demonstrate the superiority of our proposed workflow introduced in Section \ref{sec:workflow}, we first evaluate the average batch time across the three systems. Batch time encompasses the entire process for a batch, including batch construction, batch computation, and embedding updates. The results are presented in Figure \ref{fig:batch_time}. 
{\sf Legend} exhibits superior performance for each batch compared to {\sf GE$^2$} and {\sf Marius}. On average, {\sf Legend} achieves a speedup of 2.13$\times$ and 17.18$\times$ over {\sf GE$^2$} and {\sf Marius}, respectively. 
Notably, GE$^2$ exhibits higher overhead on datasets \textit{FB} and \textit{FM}. This is because the adopted embedding model is {\sf ComplEx}, as outlined in Section \ref{subsec:setting}. This embedding model includes more complex computing operations, which are harmful to computing efficiency. 
In {\sf Legend}, we design a generalized parallel execution strategy for graph embedding models in Section \ref{subsec:optGPU}, significantly enhancing the efficiency. A detailed analysis of our proposed optimization on GPU computing will be discussed in Section \ref{subsec:eva_gpu}.
On datasets \textit{TW} and \textit{FM}, {\sf Marius} constructs batches on the CPU and subsequently transfers them to the GPU, resulting in considerable communication overhead. Therefore, the average batch time of {\sf Marius} is more than 20$\times$ over {\sf Legend} on datasets \textit{TW} and \textit{FM}, as shown in Figure \ref{fig:batch_time}(c) and (d). In contrast, both {\sf Legend} and {\sf GE$^2$} offload the tasks of batch construction and negative sampling to the GPU, which achieves significant speedups.


However, the batch time reported in Figure \ref{fig:batch_time} only partially reflects the advantages of our workflow. To conduct a comprehensive evaluation, we omit all optimization modules that can be removed, including the modules of GPU optimization, edge bucket iteration order, and prefetching mechanism. The remaining components can completely reflect the performance of the workflow proposed in Section \ref{sec:workflow}. The epoch times on datasets \textit{FB}, \textit{LJ}, \textit{TW} and \textit{FM} are 0.12s, 13.06s, 291.89s and 331.40s, respectively. Compared with the epoch time in Figure \ref{fig:over_performance_time}, it still exhibits superiority over {\sf Marius} and {\sf GE$^2$} on most datasets. 

\subsection{Prefetch-friendly Order}

\begin{table}[tbp]
\renewcommand{\arraystretch}{1.2}
    \centering
    \small
        \caption{Epoch time of {\sf Legend} with and without prefetching.}
    \begin{tabular}{p{1.4cm}<{\centering}p{2.4cm}<{\centering}p{1.5cm}<{\centering}p{1.4cm}<{\centering} }
     \toprule
         Graphs & w/o Prefetching & Prefetching & Speedup\\
    \midrule
        \textit{TW} & 235.0s & 181.0s & 29.83\%\\
        \textit{FM} & 271.2s & 243.8s & 11.24\%\\
    \bottomrule
    \end{tabular}
    \label{tab:prefetch}
\end{table}

Prefetching is one of the key strategies that alleviates the limited bandwidth between the NVMe SSD and the GPU in {\sf Legend}. To evaluate the effectiveness of prefetching, we compare the performance of {\sf Legend} with and without prefetching on \textit{TW} and \textit{FM}. The results are reported in Table \ref{tab:prefetch}. 
{\sf Legend} benefits more from prefetching on \textit{TW} than on \textit{FM}, which can be attributed to the properties of the graphs. As calculated in subsection \ref{subsec:overall}, $\frac{|E|}{|V|^2}\approx4\times 10^{-8}$ for \textit{FM} while $\frac{|E|}{|V|^2}\approx8\times 10^{-7}$ for \textit{TW}. The sparsity of \textit{FM} results in an incomplete covering of I/O overhead, leading to reduced benefits from prefetching. Nonetheless, prefetching remains effective on \textit{FM}, demonstrating the scalability of the prefetching strategy in {\sf Legend}. Moreover, the edge distribution of real-world graph datasets such as \textit{TW} and \textit{FM} is always skewed, and the degree is in a power-law distribution, which is not friendly to the prefetching strategy. Even though in this scenario, the results in Table \ref{tab:prefetch} demonstrate the effectiveness of the prefetching strategy, as it can significantly reduce the computation and I/O overhead for dense edge buckets. {Notably, while {\sf Legend} achieves an overall speedup of 4.8$\times$, the prefetching gain of 29.8\% on \textit{TW} highlights that its breakthrough primarily stems from architectural innovations. Specifically, direct GPU-SSD access and computational optimization contribute more substantially to the performance improvement, whereas prefetching serves as a complementary optimization. }

\begin{figure}
\centering 

    \subfigure[\textit{TW}]{
    \includegraphics[width=0.22\textwidth]{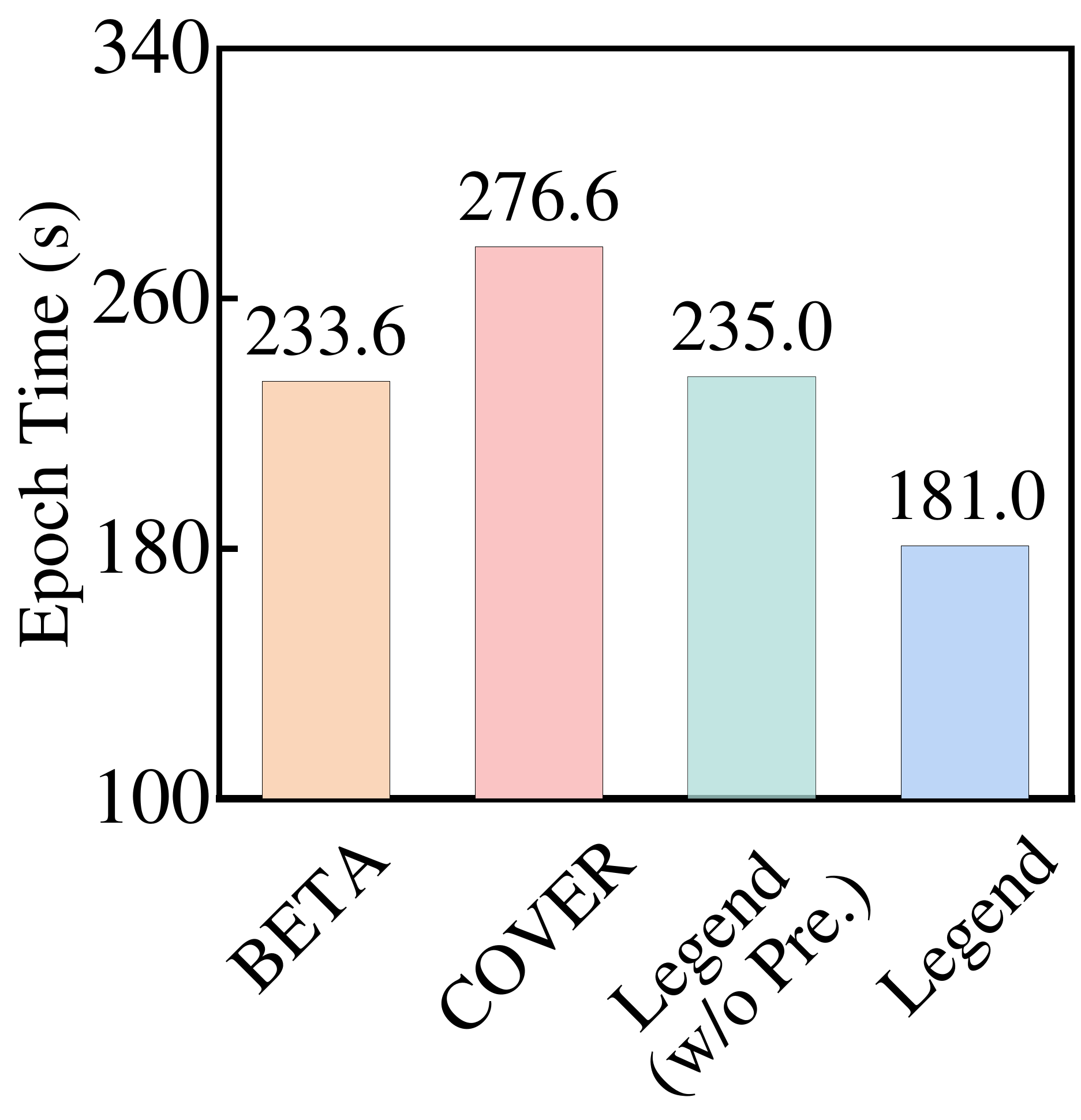}}
    \hspace{2mm}
    \subfigure[\textit{FM}]{
    \includegraphics[width=0.22\textwidth]{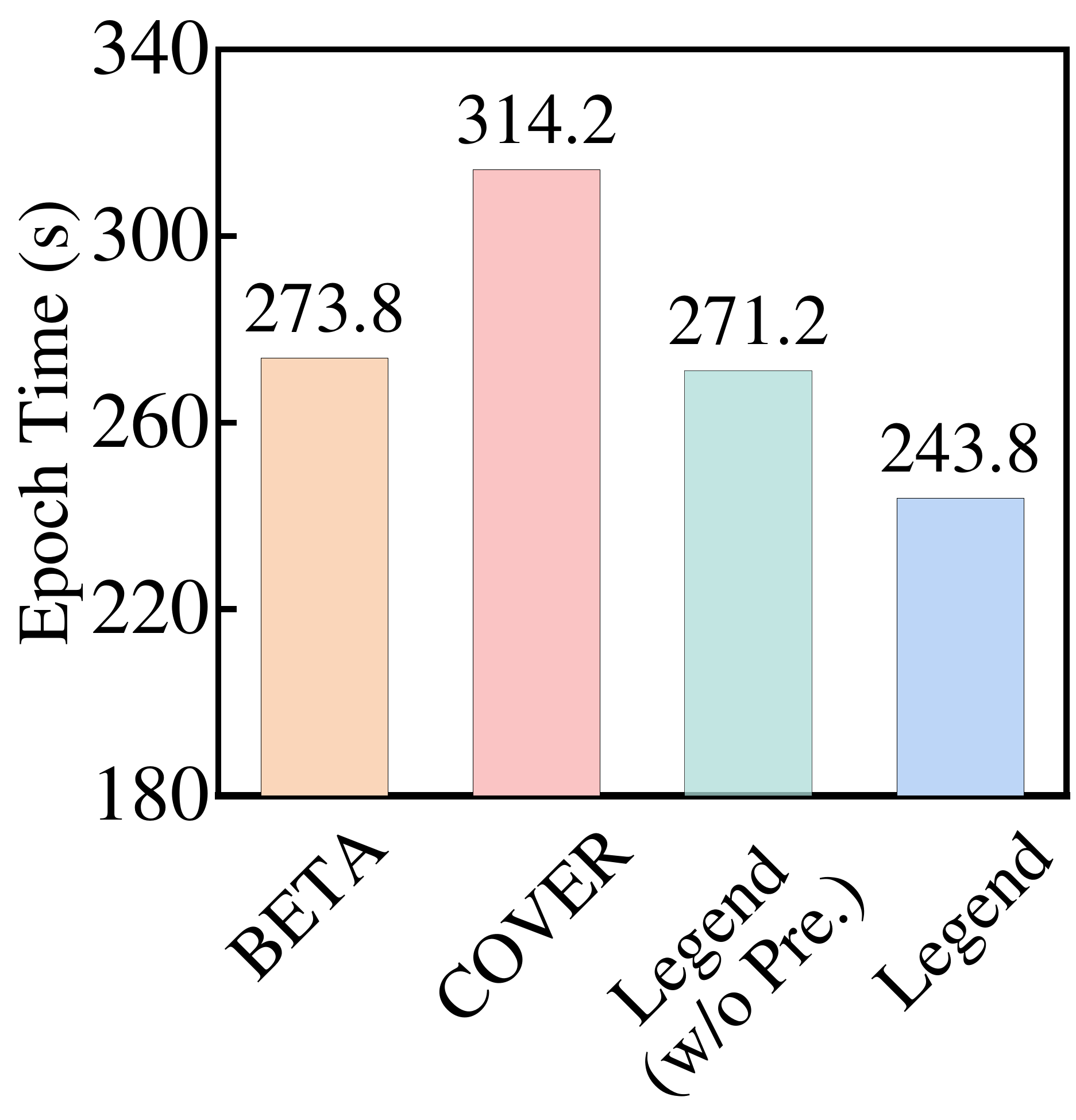}}

    \caption{Epoch time of replacing the edge buckets iterating order in {\sf Legend} with {\sf BETA} and {\sf COVER}. {\sf Legend} (w/o Pre.) denotes the epoch time that computing edge buckets using the order produced by {\sf Legend} but without prefetching.}
    \label{fig:order_time}
\end{figure}

Next, we compare the I/O order algorithms proposed in the state-of-the-art systems by applying the order used in {\sf Marius} named {\sf BETA}, and the order proposed in {\sf GE$^2$} named {\sf COVER} to {\sf Legend} to demonstrate the effectiveness of our prefetch-friendly order. Using the same settings as in subsection \ref{subsec:overall}, {\sf BETA} and {\sf COVER} divide the node embeddings into 8 and 16 partitions for \textit{TW}, and into 12 and 16 partitions for \textit{FM}. {\sf BETA} has the buffer capacity of 3, while {\sf COVER} has a buffer capacity of 4. The results are summarized in Figure \ref{fig:order_time}. 
Recall that the prefetch-friendly order generating algorithm aims to generate an order that supports prefetching while minimizing I/O times. A comparison among {\sf BETA}, {\sf COVER}, and {\sf Legend} without prefetching reveals the comparable I/O overhead between {\sf Legend}(w/o Pre.) and {\sf BETA}, which highlights the I/O efficiency of our proposed order. 
Although {\sf BETA} has I/O times close to the theoretical lower bound, its design is not conducive to prefetching, as discussed in Section \ref{sec:order}. In contrast, the ordering algorithm used in {\sf Legend} exhibits similar I/O overhead while supporting effective prefetching. Additionally, {\sf COVER} used in {\sf GE$^2$} has higher I/O overhead when applied to {\sf Legend}. This is because it is specifically designed for training with multiple GPUs, which is not optimized for single-GPU scenarios. 

{To further analyze the superiority of our proposed node partition loading order, we summarize the I/O times (counts of partition transfers) between the storage and the computing device, and calculate the communication volume for the three ordering algorithms in Table \ref{tab:iotimes}. The communication volume uses the same units as the partition size $S$. Since {\sf COVER} can only accommodate partition numbers of $4^L$, we report its metrics when the number of partitions is 16.} 
As shown in Table \ref{tab:iotimes}, {\sf BETA} and {\sf Legend} have similar I/O times and communication volumes within the evaluated partitions. This indicates that our proposed partition loading order achieves comparable I/O times with {\sf BETA}, which has I/O times close to the theoretical lower bound. However, {\sf BETA} does not support embedding prefetching as illustrated in Section \ref{sec:order}. In contrast, the loading order proposed in {\sf Legend} supports efficient embedding prefetching while achieving low I/O times.   
Additionally, {\sf COVER} is adopted by {\sf GE$^2$} to overcome the issue of I/O overhead within multiple GPUs. It is not optimized for a single GPU. In contrast, the communication volume remains unchanged with the increasing number of GPUs. Devising a prefetching-friendly and low-overhead ordering algorithm that supports multiple GPUs like {\sf COVER} is left to future work.


\begin{table}[tbp]
\renewcommand{\arraystretch}{1.2}
\belowrulesep=0pt
\aboverulesep=0pt
    \caption{I/O times and communication volume of different ordering algorithms with various numbers of partitions. $S$ denotes the size of node embeddings and optimizer states. }
    \centering
    \small
    \begin{tabular}{p{0.5cm}<{\centering}|p{0.8cm}<{\centering}p{0.8cm}<{\centering}p{0.9cm}<{\centering}|p{0.8cm}<{\centering}p{0.8cm}<{\centering}p{0.9cm}<{\centering}}
    \toprule
        \multirow{2}{*}{Par.} & \multicolumn{3}{c|}{I/O times} & \multicolumn{3}{c}{Communication volume} \\
    \cline{2-7}
        & {\sf BETA} & {\sf COVER} & {\sf \textbf{Legend}}  & {\sf BETA} & {\sf COVER} & {\sf \textbf{Legend}}  \\
    \midrule
        6 & 8 & - & {8}  & 1.33S & - & {1.33S} \\
        8 & 15 & - & {16} & 1.88S & - & {2S} \\
        10 & 24 & - & {24}  & 2.4S & - & {2.4S} \\
        12 & 34 & - & {36}  & 2.83S & - & {3S} \\
        14 & 48 & - & {50}  & 3.43S & - & {3.57S} \\
        16 & 63 & 80 & {66} & 3.94S & 5S & {4.13S} \\
    \bottomrule
    \end{tabular}
    \label{tab:iotimes}
\end{table}

\begin{figure}
\centering 
    \subfigure[Read]{
    \includegraphics[width=0.225\textwidth]{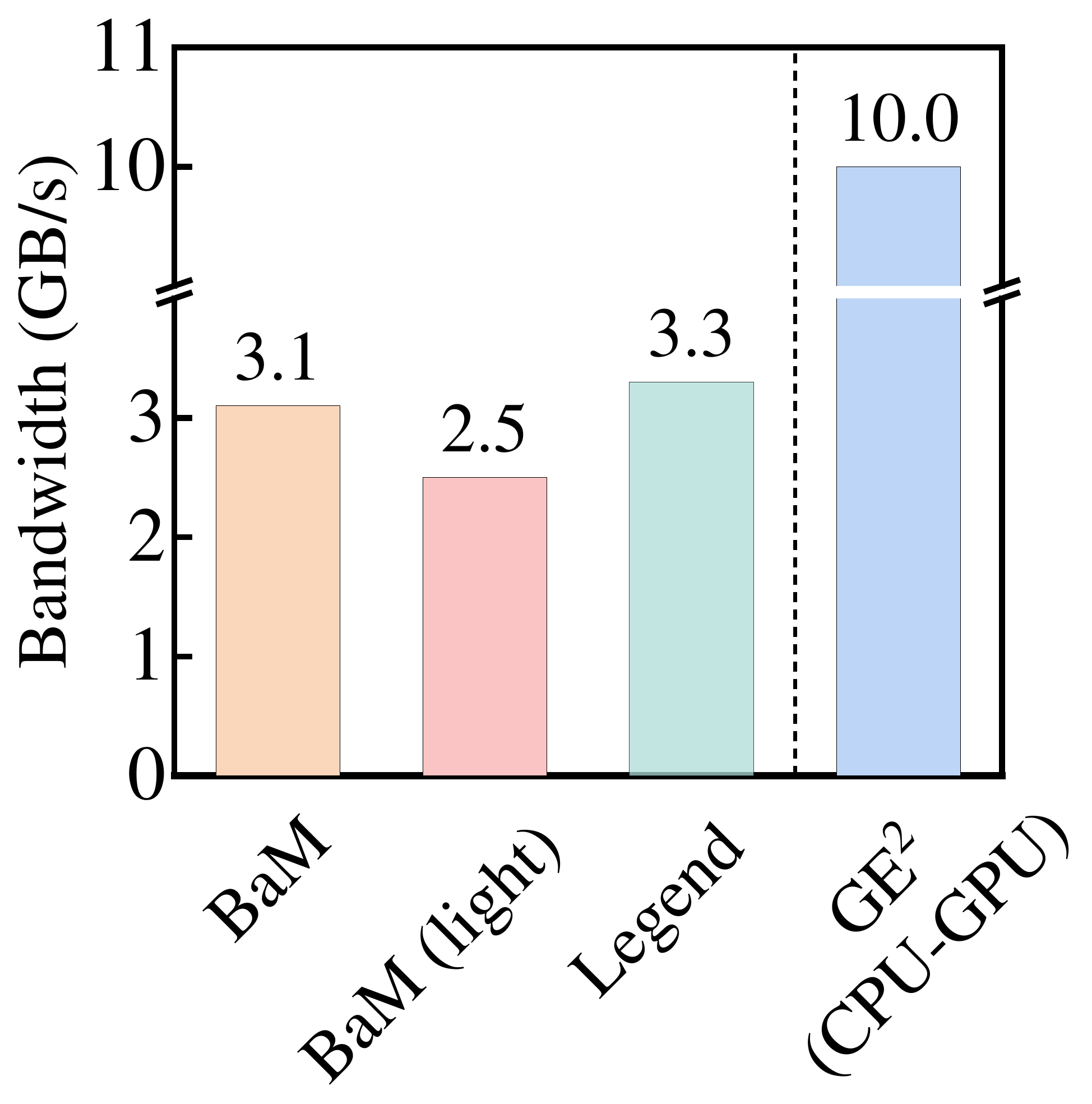}}
    \hspace{1mm}
    \subfigure[Write]{
    \includegraphics[width=0.225\textwidth]{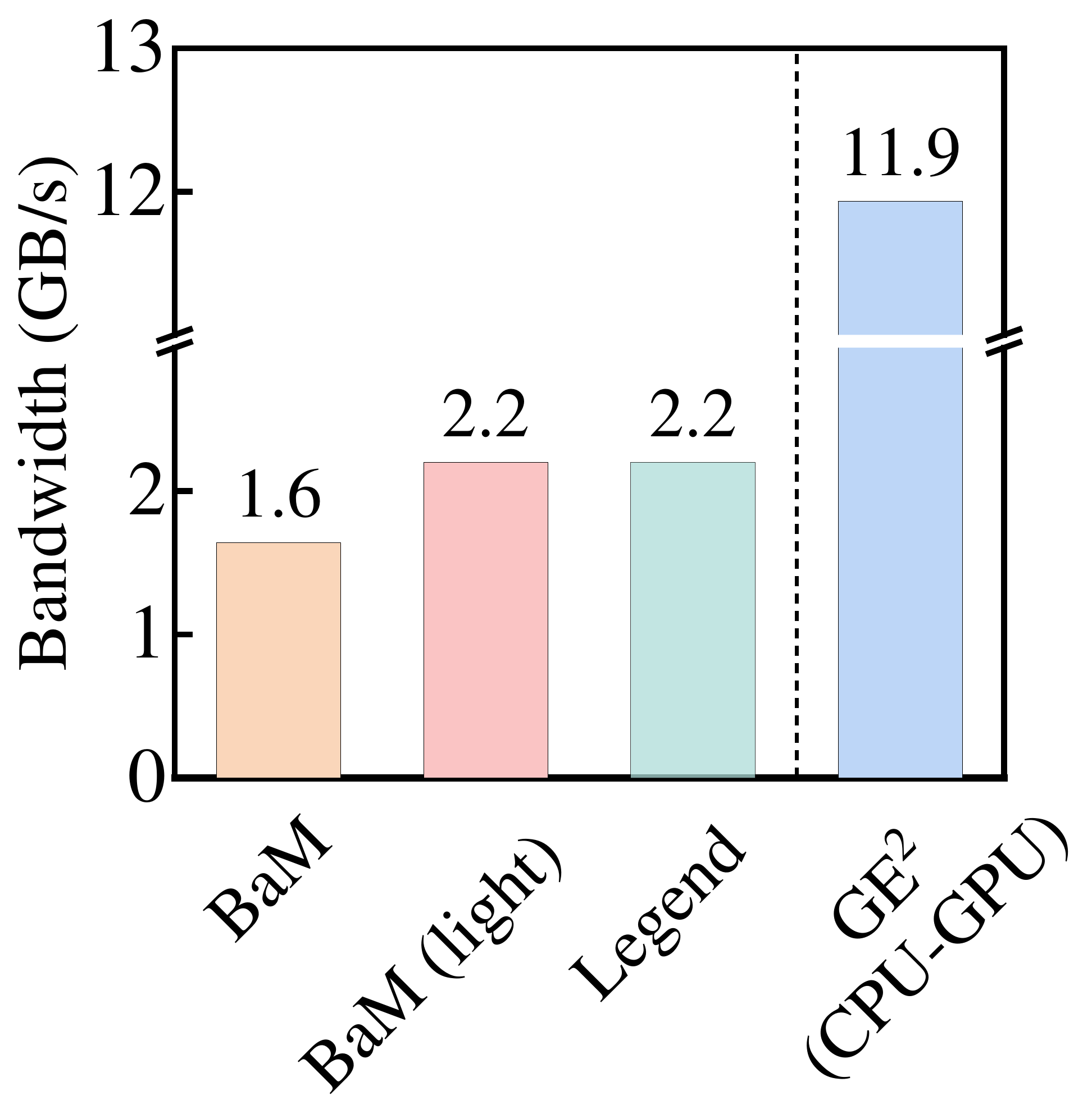}}

    \caption{{Bandwidth of GPU direct access to SSD on test data with volume of 128GB. {\sf GE$^2$} is reported as the bandwidth between CPU and GPU.} }
    \label{fig:bandwidth}
\end{figure}

\subsection{GPU Direct Access to NVMe SSD}
As discussed in Section \ref{subsec:optNVMe}, we aim to design a GPU direct access strategy to NVMe SSD that achieves high performance as well as supports the simultaneous execution of data access and batch calculation kernels. To this end, we separately evaluate the bandwidth of read/write and the ability to simultaneously execute together with the calculation kernel. 

{We first compare the bandwidth of {\sf Legend} with the state-of-the-art GPU direct access method named {\sf BaM}. Moreover, we also evaluate the bandwidth between the CPU and GPU achieved by {\sf GE$^2$}. For our evaluation, the test data volume is set to 128 GB.} We employ 4096 thread blocks for {\sf BaM}, each containing 32 threads, while for {\sf Legend}, we employ 8 thread blocks, each containing 32 threads. We also evaluate {\sf BaM} with the same settings as {\sf Legend}, referred to as {\sf BaM} (light). 
As presented in Figure \ref{fig:bandwidth}, {\sf Legend} achieves comparable I/O bandwidth to {\sf BaM}. Notably, the writing bandwidth of {\sf Legend} outperforms {\sf BaM} due to its high parallel queue management mechanism and low-cost doorbell ringing strategy. 
Under the same settings, {\sf Legend} achieves higher I/O bandwidth than {\sf BaM} (light). This is because we propose a lightweight NVMe SSD driver in Section \ref{subsec:optNVMe}, which utilizes fewer resources on the GPU while achieving better performance. {The results demonstrate the effectiveness of our novel strategies—lock-free batch enqueuing, fully coalesced doorbell writes, and batch polling, avoiding the complex locks and doorbell operations in typical GPU-SSD direct access drivers.} 
For {\sf GE$^2$}, the communication bandwidth between the CPU and GPU is over 3 times higher than that between the GPU and NVMe SSD. This gap can be mitigated by carefully prefetching data, as demonstrated in our previous experiments. 

\begin{figure}
    \centering
    \includegraphics[width=0.45\textwidth]{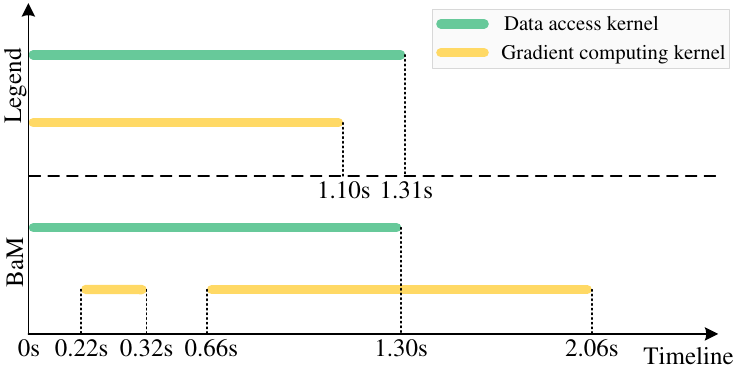}
    \caption{Timeline of simultaneous execution of kernels. }
    \label{fig:timeline}
    \vspace{-2mm}
\end{figure}

To evaluate the capability of the GPU-SSD direct data access kernel to execute concurrently with the batch computing kernel, we run both kernels simultaneously by using CUDA streams. For the batch computing kernel, we fix the batch size at $10^5$ and execute the batch computation 100 times. The execution timelines are depicted in Figure \ref{fig:timeline}. 
In {\sf Legend}, both kernels can be executed concurrently with minimal performance degradation. In contrast, the data access kernel in {\sf BaM} occupies a significant amount of resources, seriously affecting the execution performance of the batch computing kernel. 
By considering the specific workload of graph embedding learning, we have simplified the complexity of the GPU-SSD direct access driver and designed novel direct access strategies, resulting in a lightweight but high-performance GPU-NVMe SSD direct access kernel.

\begin{figure}
\centering 

    \subfigure[\textit{{FB}}]{
    \includegraphics[width=0.22\textwidth]{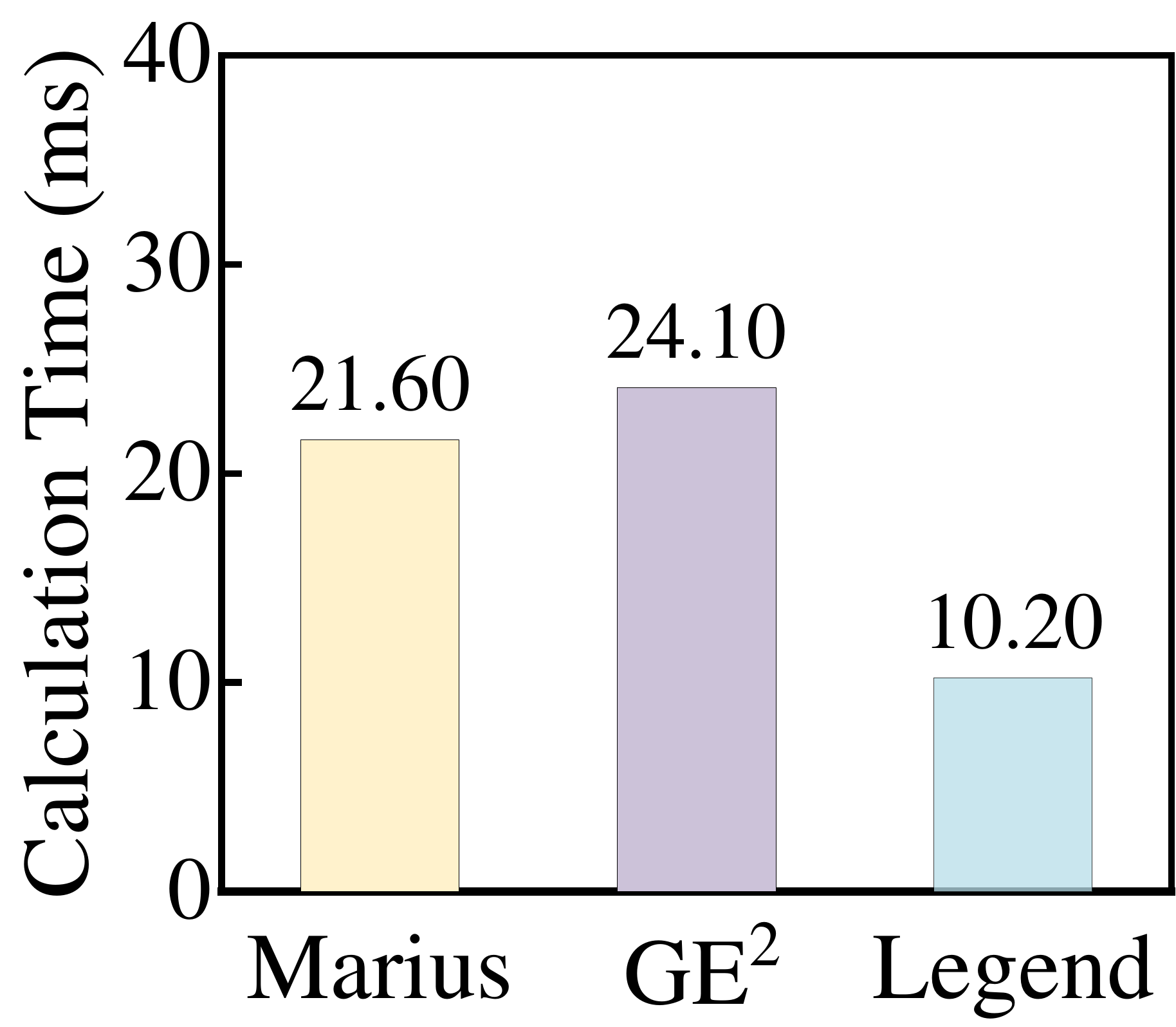}}
    \hspace{2mm}
    \subfigure[\textit{LJ}]{
    \includegraphics[width=0.22\textwidth]{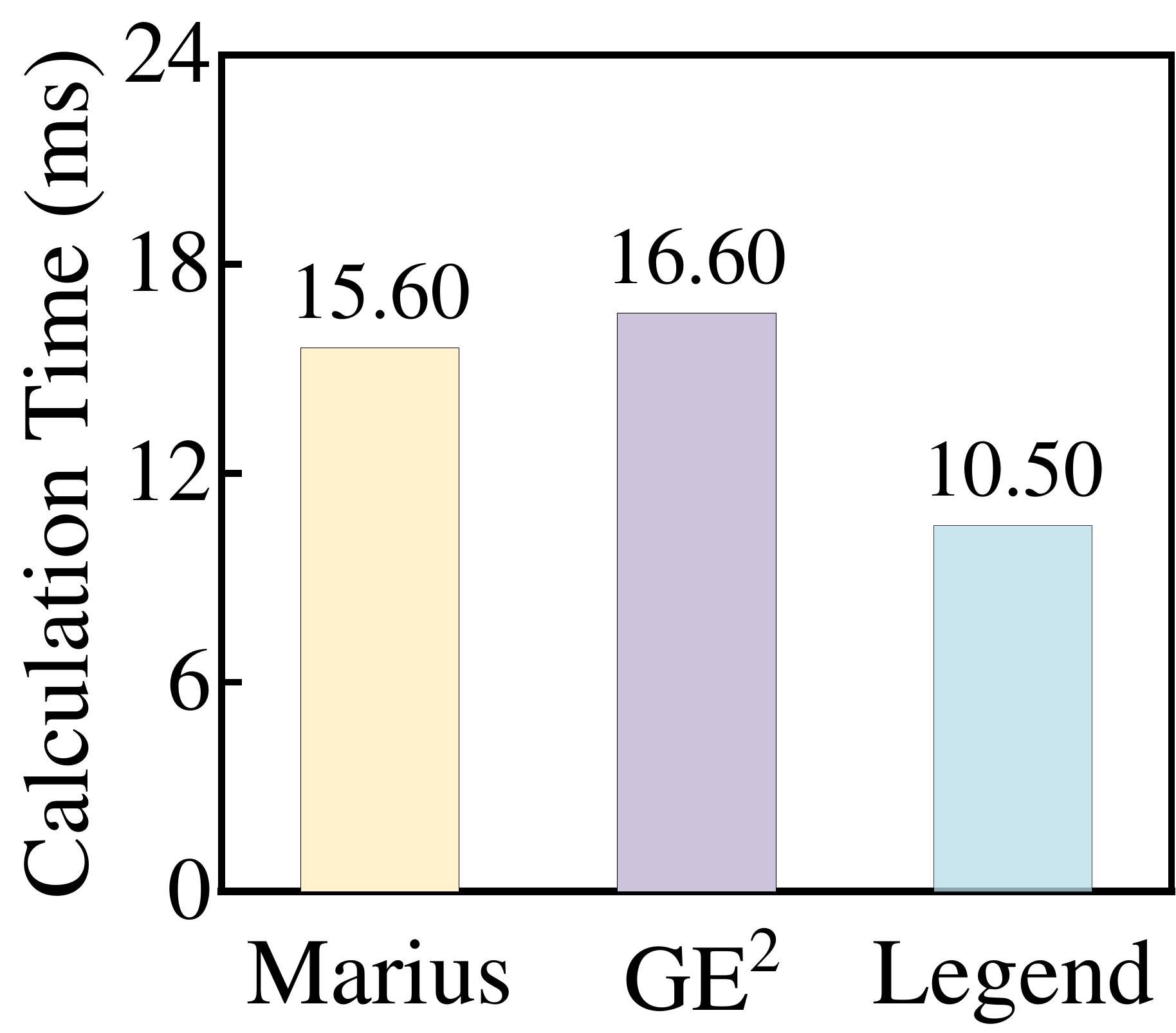}}
    \hspace{2mm}
    \subfigure[\textit{TW}]{
    \includegraphics[width=0.22\textwidth]{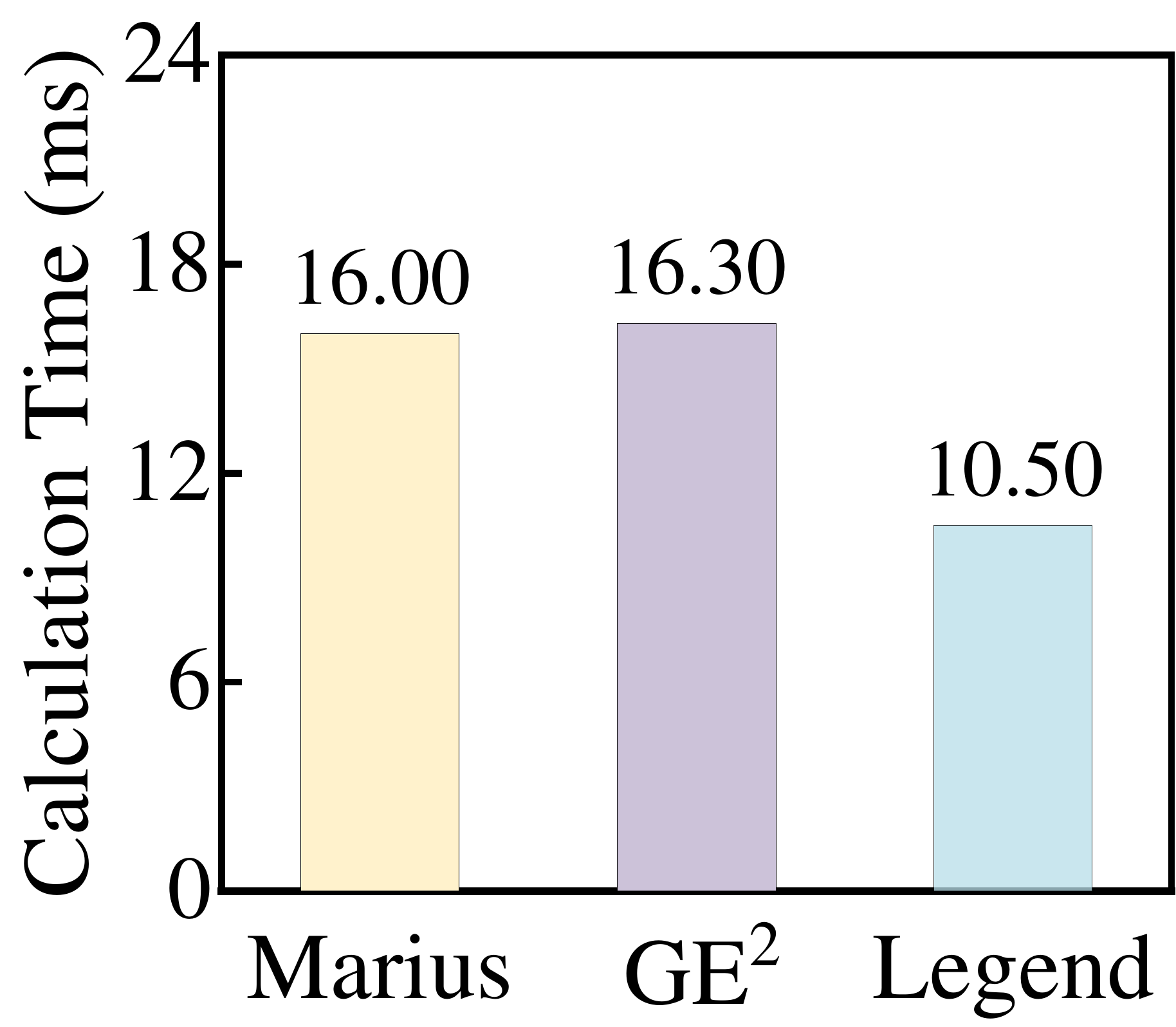}}
    \hspace{2mm}
    \subfigure[\textit{FM}]{
    \includegraphics[width=0.22\textwidth]{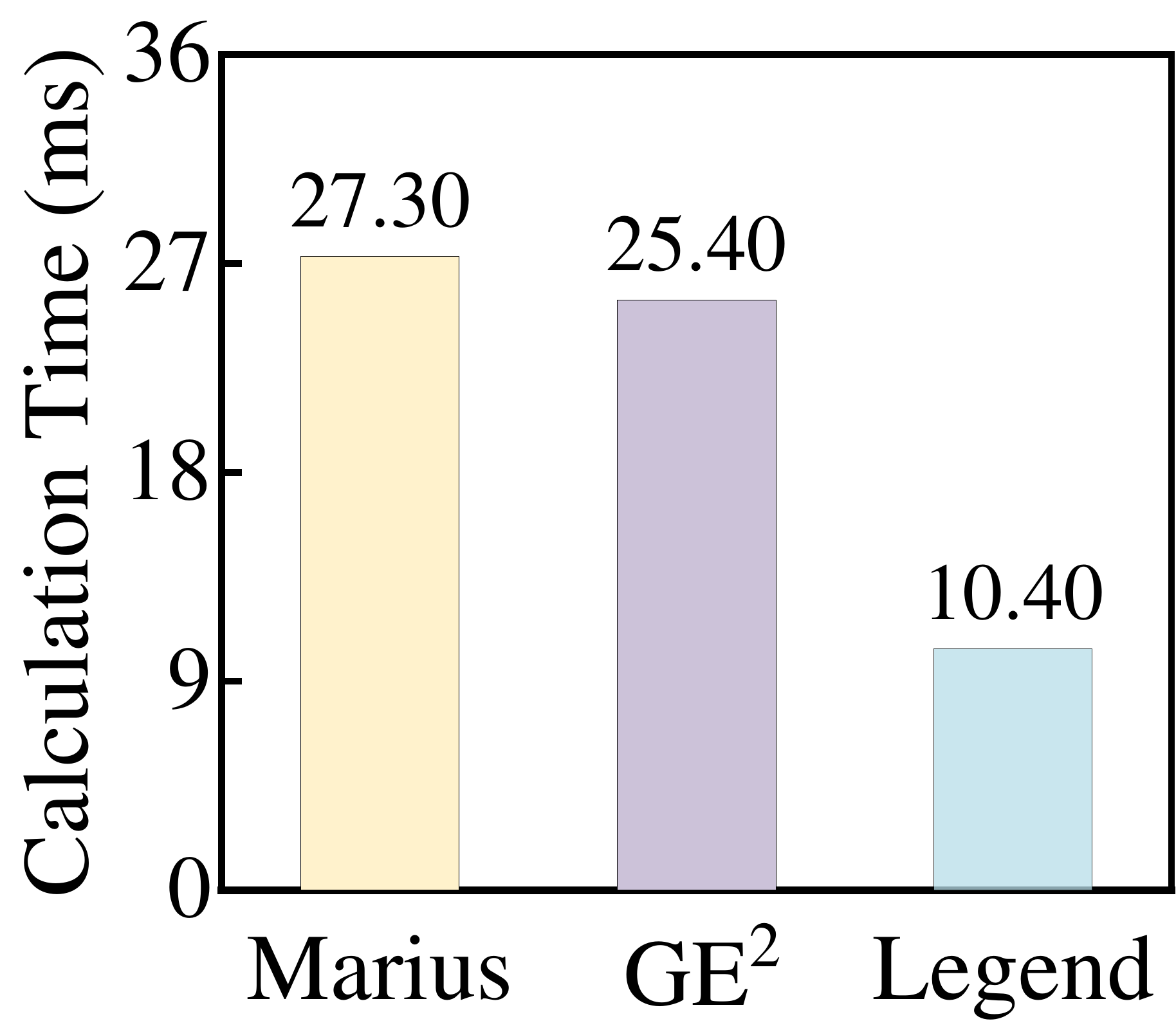}}

    \caption{Comparison of the average batch calculation time for a single batch.}
    \label{fig:cal_time}
    \vspace{-2mm}
\end{figure}

\subsection{Optimizations on the GPU}
\label{subsec:eva_gpu}
In this subsection, we evaluate the optimization techniques applied to GPU computing, as proposed in Section \ref{subsec:optGPU}. To achieve this, we measure the average calculation time per batch. The results are reported in Figure \ref{fig:cal_time}. 
{\sf Marius} and {\sf GE$^2$} exhibit similar performance across the four datasets, as they utilize the same training engine. The training overhead for both systems on \textit{FB} and \textit{FM} is greater compared to \textit{LJ} and \textit{TW}. This discrepancy arises because these datasets are different types of graphs and employ distinct embedding models, as discussed in subsection \ref{subsec:setting}. Specifically, \textit{FB} and \textit{FM} are knowledge graphs with multiple types of edges, which use {\sf ComplEx} model. These methods calculate embeddings of edges while \textit{LJ} and \textit{TW} don't. Moreover, the computing process of {\sf ComplEx} involves cross calculation, which is more complex and leads to higher overhead.  
In {\sf Legend}, we combine the batch calculation processes for edge embeddings with those of node embeddings, eliminating redundant calculation. We also devise a generalized parallel strategy to fully utilize the computing and storage resources on the GPU. As a consequence, the calculation overhead remains consistent across the 4 datasets. 

\begin{figure}
\centering 
    \includegraphics[width=0.43\textwidth]{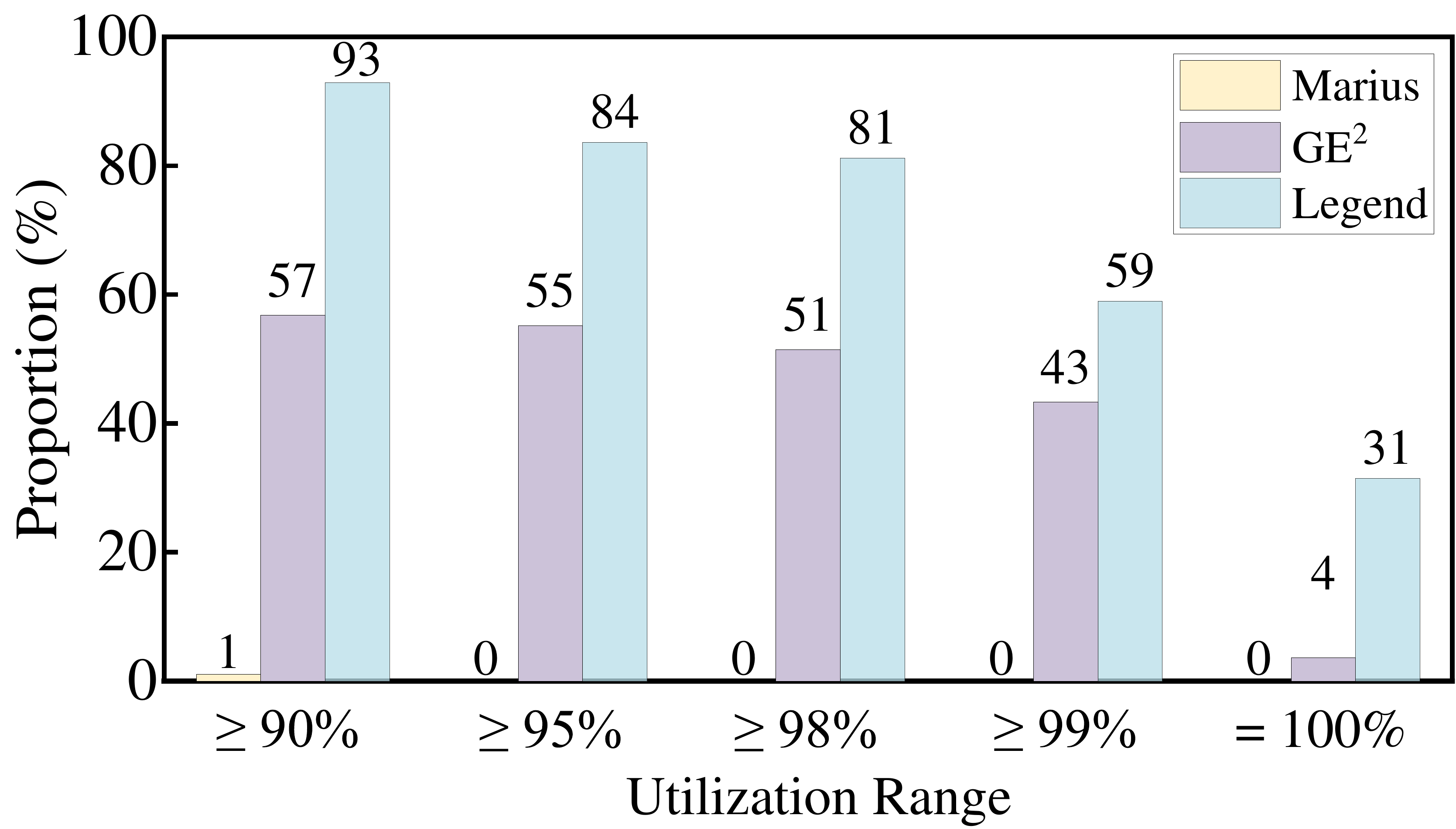}
    \caption{{Statistical information on the GPU utilization of graph embedding systems during training on \textit{TW}. The utilization of {\sf Marius} remains below 95\%.}}
    \label{fig:uti_fig}
    \vspace{-3mm}
\end{figure}

\begin{figure}
\centering 

    \subfigure[\textit{Dim=50}]{
    \includegraphics[width=0.22\textwidth]{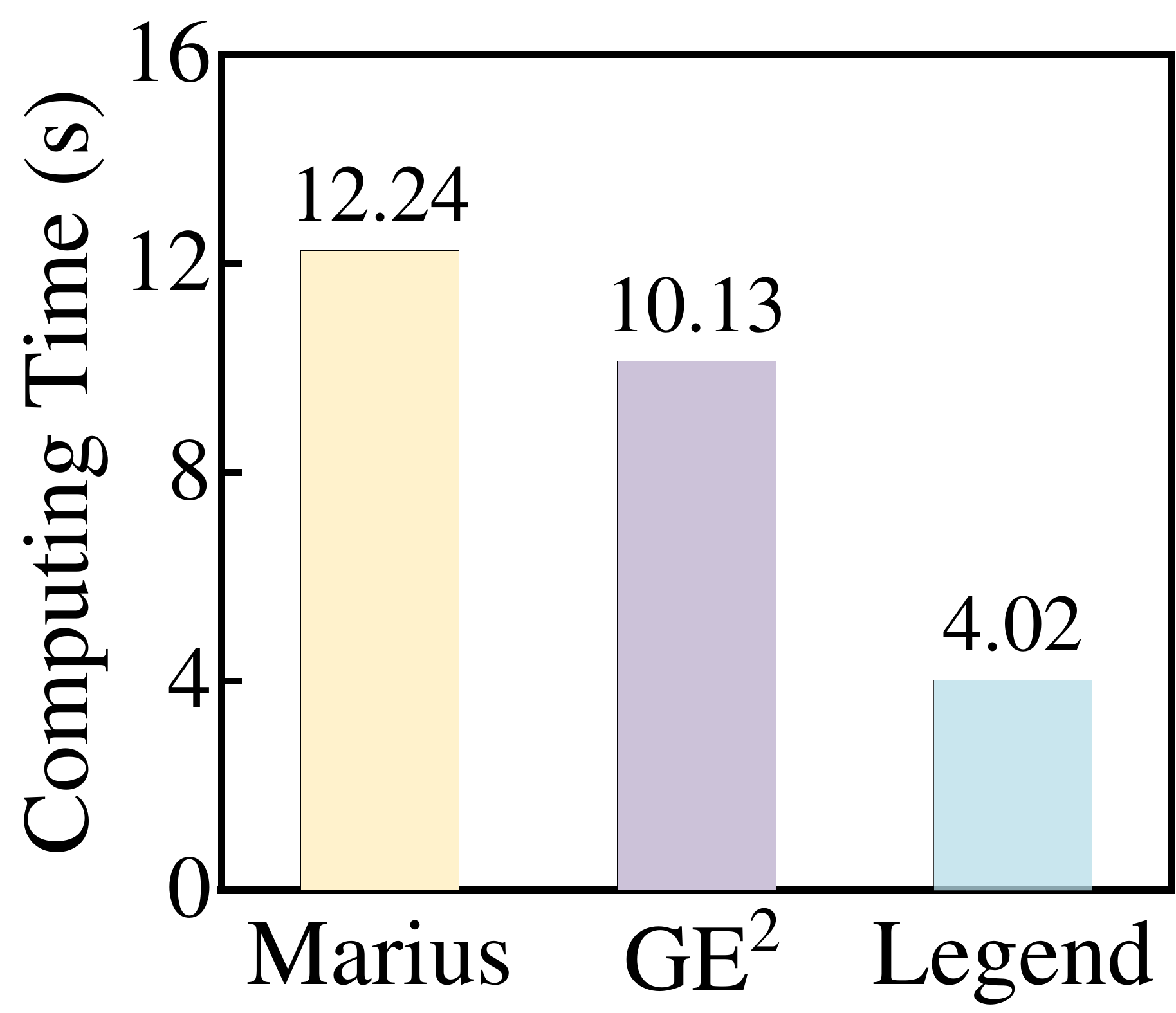}}
    \hspace{2mm}
    \subfigure[\textit{Dim=100}]{
    \includegraphics[width=0.22\textwidth]{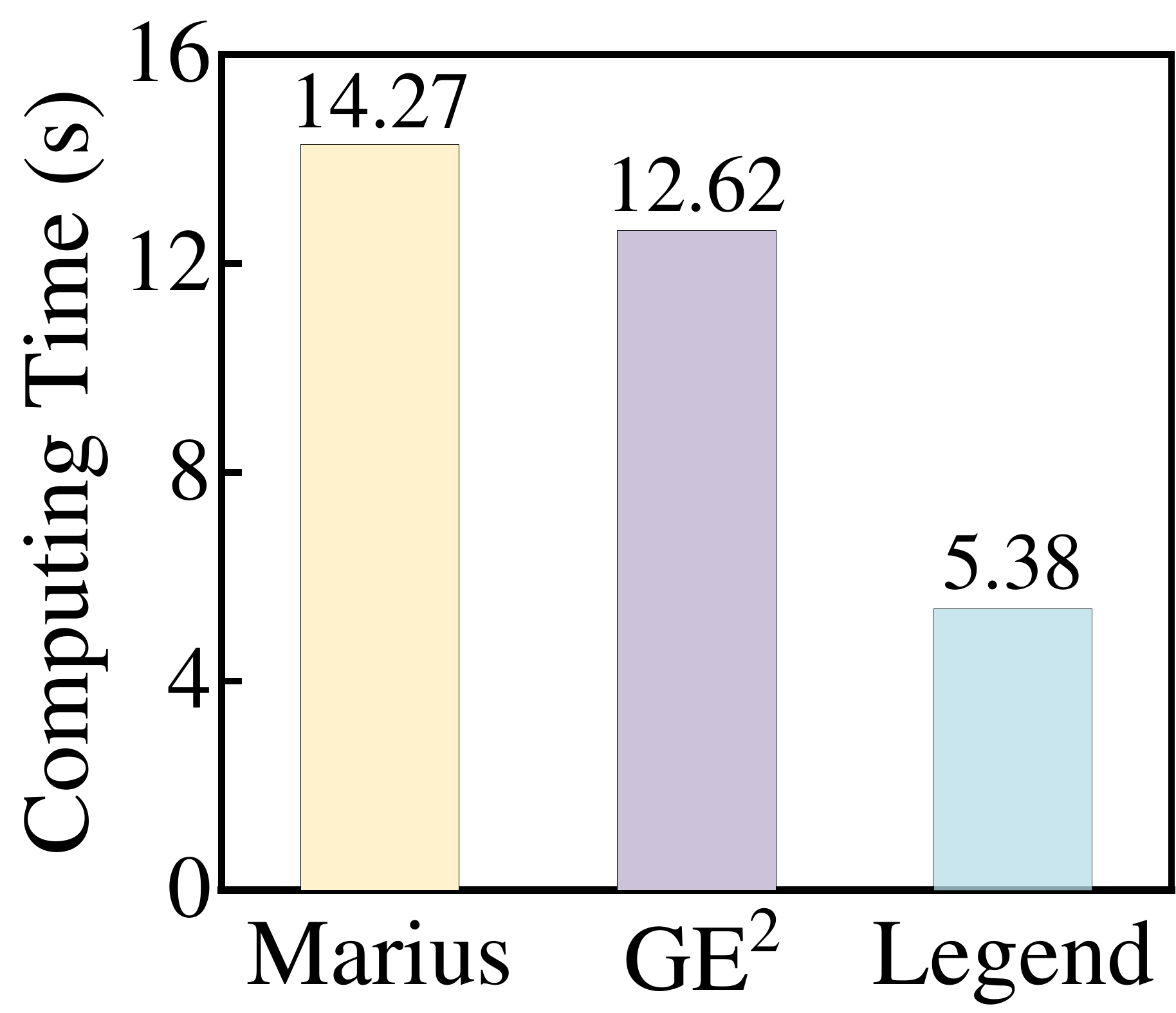}}
    \hspace{2mm}
    \subfigure[\textit{Dim=150}]{
    \includegraphics[width=0.22\textwidth]{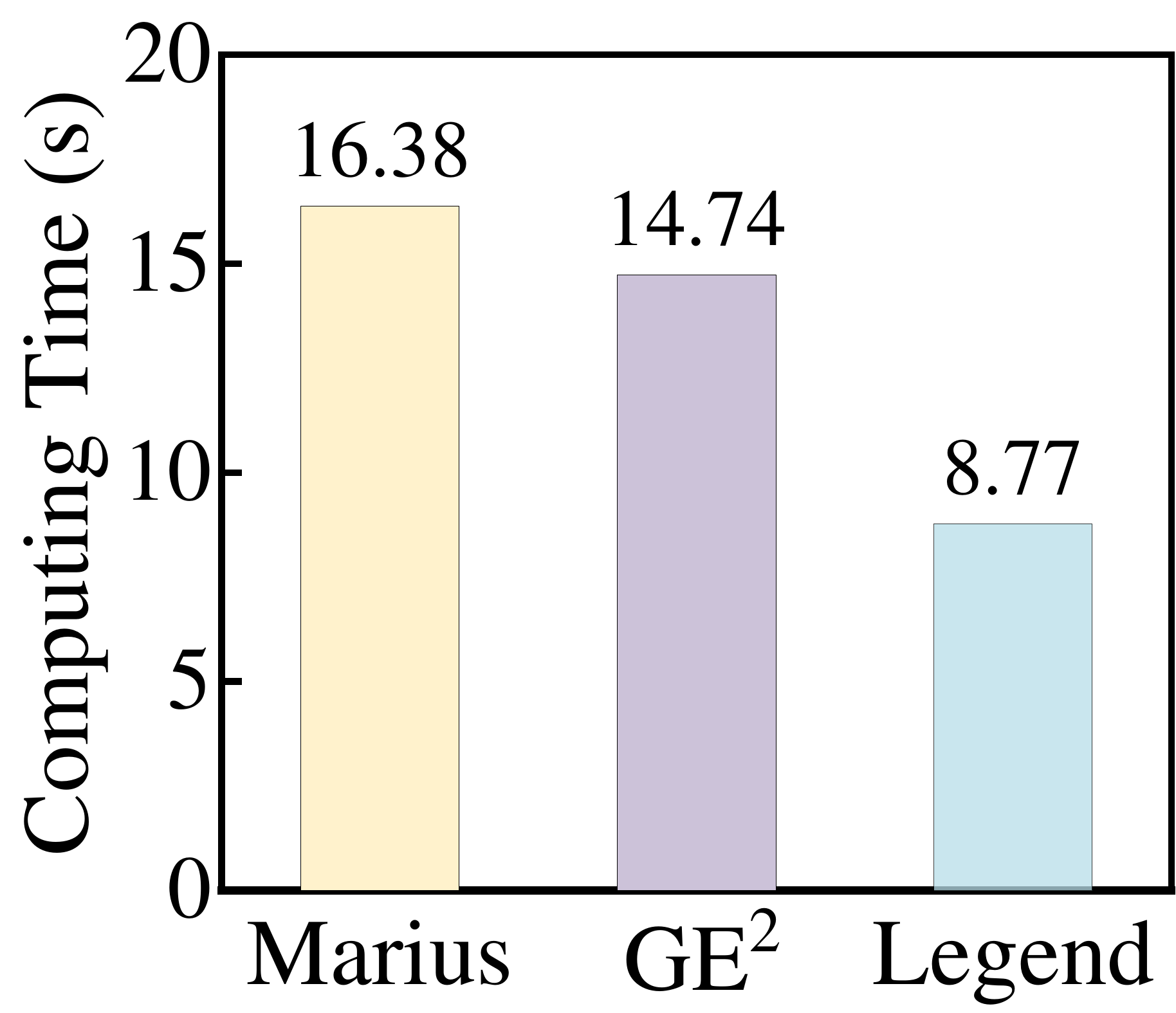}}
    \hspace{2mm}
    \subfigure[\textit{Dim=200}]{
    \includegraphics[width=0.22\textwidth]{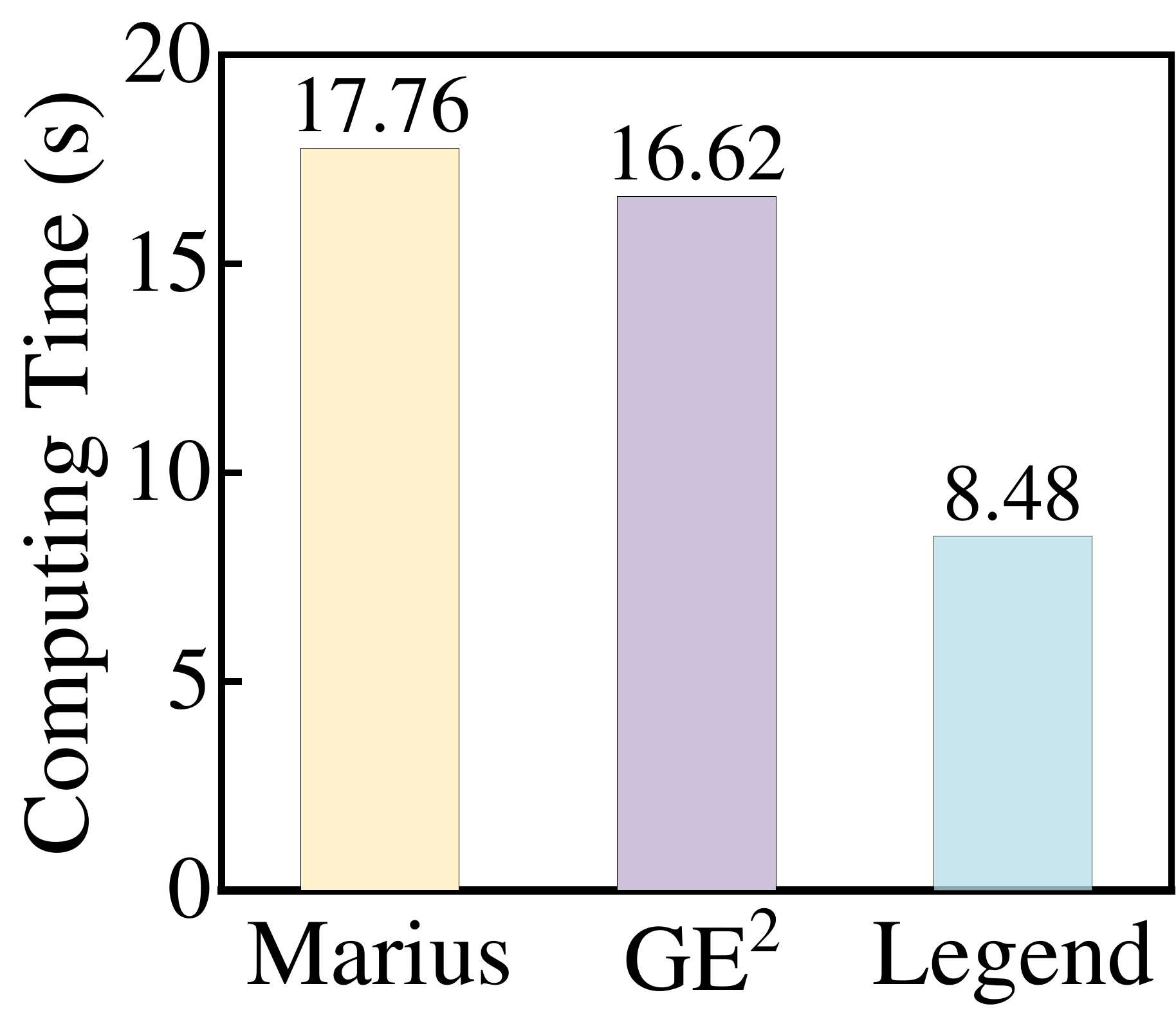}}

    \caption{{Comparison of the batch computing overhead with various embedding dimensions in the batch size of 50000 using the {\sf ComplEx} embedding model on \textit{FB}.}}
    \label{fig:dim_cal_time}
\end{figure}

On \textit{LJ} and \textit{TW}, {\sf Legend} achieves a speedup ratio of 1.5$\times$-1.6$\times$, while the speedup exceeds 2$\times$ on \textit{FB} and \textit{FM}. This indicates that the parallel strategy and intermediate results reuse techniques proposed in Section \ref{subsec:optGPU} are effective. Furthermore, we conduct statistics on the GPU utilization of each graph embedding system during the training process. As shown in Figure \ref{fig:uti_fig}, the GPU utilization for {\sf Legend} remains above 98\% for 81.22\% of the time, above 99\% for 59\% of the time, and reaches 100\% for 31.49\% of the time. In contrast, the proportions of time during which GPU utilization exceeds 98\%, 99\%, and 100\% for {\sf GE$^2$} are 51.48\%, 43.28\%, and 3.64\%, respectively. Thus, the optimization of training on the GPU significantly improves GPU utilization, leading to a substantial reduction in calculation overhead. 

{Furthermore, to evaluate the scalability of the three graph embedding systems with varying embedding dimensions, we conduct experiments using {\sf ComplEx} embedding model with a batch size of 50000 by setting the embedding dimension to 50, 100, 150, and 200. 
The execution overheads of the three graph embedding systems are depicted in Figure \ref{fig:dim_cal_time}. 
{\sf Legend} consistently outperforms the other graph embedding systems with various embedding dimensions. It achieves a speedup ratio of 2.13$\times$ over {\sf GE$^2$} and 2.42$\times$ over {\sf Marius}, which demonstrates the computational scalability of our proposed {\sf Legend} with various embedding dimensions. 
The superior performance of {\sf Legend} is attributed to the optimizations of parallel strategy, memory access, and computing process in Section \ref{subsec:optGPU}.
When the embedding dimension is 150, the computing overhead of {\sf Legend} is similar to that when the dimension is 200. This is attributed to the warp scheduling model of the GPU, where a small number of warps can be scheduled to execute concurrently on a single SM. }

\section{Related Work}
\label{relatedworks}


\noindent \textbf{Graph embedding models}. Extensive studies have been conducted to enhance the quality of graph embeddings. For general graphs, existing algorithms typically sample edges based on random walks~\cite{ribeiro2017struc2vec,wang2018billion}. For example,
\textsf{DeepWalk}~\cite{perozzi2014deepwalk} employs the idea of \textsf{Word2Vec}~\cite{tom2013efficient}, generating a series of random walk paths and training the embedding by the skip-gram model. \textsf{Node2Vec}~\cite{grover2016node2vec} improves \textsf{DeepWalk}, which balances the embedding results between homogeneity and structure of the network by adjusting the weights of random walks. \textsf{LINE}~\cite{tang2015line} defines first-order and second-order similarity on the graph to constrain the learning of embeddings. 
In the context of multi-relation graphs, all edges in the graph are used for embedding learning without sampling~\cite{lerer2019pytorch,zheng2020dgl,mohoney2021marius,zheng2024ge2}. Extensive multi-relation graph embedding models have been developed, categorized into two primary types: translational distance models and semantic matching models~\cite{wang2017knowledge}. Translational distance models, such as \textsf{TransE}~\cite{bordes2013translating} and \textsf{TransH}~\cite{wang2014knowledge}, employ distance-based scoring functions to evaluate the plausibility of facts between entities. Semantic matching models, such as \textsf{DistMult}~\cite{yang2015embedding} and \textsf{ComplEx}~\cite{trouillon2016complex}, utilize similarity-based scoring functions to assess plausibility based on entities' semantics. These multi-relation graph embedding models can be easily integrated into {\sf Legend}. 

\noindent \textbf{Graph embedding systems}. Significant efforts have been dedicated to developing efficient systems for graph embedding training. \textsf{GraphVite}~\cite{zhu2019graphvite} employs the CPU to generate random walks and sample negative edges on GPUs to achieve efficient embedding learning on general graphs. \textsf{HET-KE}~\cite{dong2022het} proposes a cache embedding table to reduce the communication overhead among distributed machines. \textsf{DistGER}~\cite{fang2023distributed} proposes an efficient distributed graph embedding system. For massive knowledge graphs, \textsf{PBG}~\cite{lerer2019pytorch} proposes a batched negative sampling method to reduce memory access overhead. \textsf{DGL-KE}~\cite{zheng2020dgl} overlaps the gradient update with batch processing to reduce GPU idle time. Kochsiek et al. provide a comprehensive experimental study of the existing knowledge graph embedding training techniques~\cite{kochsiek2021parallel}. Following \textsf{PBG}, \textsf{Marius}~\cite{mohoney2021marius} proposes a partition loading order \textsf{BETA} to reduce the I/O times and pipelines the training procedure on the CPU and GPU. \textsf{GE$^2$}~\cite{zheng2024ge2} designs a general negative sampling execution model, and proposes a loading order to reduce I/O overhead between RAM and GPUs. 
Different from them, our proposed graph embedding system, {\sf Legend}, employs GPU-SSD direct access and prefetch-friendly order to optimize the I/O efficiency, while utilizing a customized GPU kernel to optimize the computing efficiency. {{\sf Legend} is specifically designed for graph embedding, but the framework also has the potential to be adopted in other areas such as out-of-core GNN training~\cite{jeongmin2024accelerating}, DNN training~\cite{bae2021flashneuron}, and large-scale vector search~\cite{huang2024neos} with customized optimizations. }

\noindent \textbf{GPU direct access to NVMe SSD}. Recent research has studied the GPU-SSD direct access to meet the demand for low latency and large capacity. GPUDirect Storage (\textsf{GDS})~\cite{refgds} is a library supporting data transmission between GPU and NVMe SSD through a bounce buffer in the CPU’s memory and a direct memory access (DMA) engine. However, \textsf{GDS} is still restricted by the high-overhead software stacks. To completely break free from the limitations of the Linux software stacks, \textsf{BaM}~\cite{qureshi2023gpu} proposes a queue management mechanism and caching strategy completely on the GPU to achieve high-throughput access to storage. There are also various customized GPU direct access methods for specific applications such as DNN training~\cite{bae2021flashneuron}, GNN training~\cite{jeongmin2024accelerating}, vector retrieval~\cite{huang2024neos}, and data analysis in OLAP~\cite{li2016hippogriffdb}. These methods simply adopt a GPU-SSD direct access library without optimizing the underlying access mechanism. The graph embedding workflow has its specific I/O pattern as discussed in Section \ref{subsec:optNVMe}.
To maximize the bandwidth between the GPU and SSD during graph embedding, {\sf Legend} customizes the queue management, doorbell ringing, and polling mechanism for the GPU-SSD direct access driver according to the graph embedding workload, significantly reducing the data access overhead. 

\noindent {\textbf{Specialized GPU kernels for graph learning.} Computational overhead is always a primary challenge for graph learning. To tackle this challenge, recent studies have designed custom GPU kernels for each module in graph learning. To achieve efficient graph sampling, {\sf gSamlper}~\cite{gong2023gsampler} proposes efficient sampling kernels with operator fusion, while {\sf FlowWalker}~\cite{mei2024flowwalker} implements a sampling kernel framework with memory-efficient optimizations. 
Another general operation in graph learning is SDDMM/SpMM. {\sf FusedMM}~\cite{rahman2021fusedmm} design a unified SDDMM-SpMM kernel to avoid redundant intermediate results. {\sf TC-GNN}~\cite{wang2023tc} employs Tensor cores to achieve more efficient SpMM on graphs. Moreover, {\sf HC-SpMM}~\cite{li2025hc} proposes a hybrid CUDA-Tensor kernel to fully utilize the heterogeneous GPU cores. 
GNN, as a primary framework in graph learning, has also been studied to design specialized GPU kernels. {\sf GNNAdvisor}~\cite{wang2021gnnadvisor} designs GPU kernels specially for GNN workload to improve memory access and GPU utilization. {\sf PruneGNN}~\cite{gurevin2024prunegnn} proposes SIMD-aware kernels to exploit matrix-operator-level parallelism. 
The graph embedding models, such as {\sf DistMult}, have a different computing paradigm from graph sampling and SpMM. {\sf Legend} designs an optimized GPU kernel to enhance the computing efficiency of these models. }
\section{Conclusion}
\label{sec:conclusion}

We introduce {\sf Legend}, a lightweight graph embedding system. {\sf Legend} systematically integrates CPU, GPU, and NVMe SSD resources, which perform efficient and scalable embedding training. 
We carefully design the workflow to enable a seamless introduction of the NVMe SSD into the system and distribute tasks according to the unique characteristics of each hardware component. Meanwhile, we design an edge bucket iteration order that minimizes the I/O times between GPU and SSD while supporting efficient prefetching, and a customized GPU-SSD direct access driver to significantly reduce I/O overhead. 
Furthermore, we propose an efficient parallel strategy for graph embedding workload to optimize the computation on the GPU, ensuring efficient handling of billion-scale datasets. Experimental results consistently demonstrate the superiority of {\sf Legend}.



%
%
\bibliographystyle{spmpsci}
\bibliography{ref}

\end{document}